\documentclass[aap,preprint]{imsart}

\RequirePackage{amsthm,amsmath,amsfonts,amssymb}
\RequirePackage[numbers,sort&compress]{natbib}
\RequirePackage[T1]{fontenc}
\RequirePackage[utf8]{inputenc}
\RequirePackage{mathtools,mathrsfs}
\RequirePackage[colorlinks,citecolor=blue,urlcolor=blue]{hyperref}
\makeatletter
\renewcommand*{\nocontentsline}{%
  \let\@@addcontentsline\addcontentsline
  \def\addcontentsline##1##2##3{%
    \let\addcontentsline\@@addcontentsline}%
}
\makeatother

\RequirePackage{microtype,tikz}
\definecolor{replicaone}{HTML}{205E8C}
\definecolor{replicatwo}{HTML}{BA5A21}

\startlocaldefs
\theoremstyle{plain}
\newtheorem{theorem}{Theorem}[section]
\newtheorem{lemma}[theorem]{Lemma}
\newtheorem{proposition}[theorem]{Proposition}
\newtheorem{corollary}[theorem]{Corollary}
\theoremstyle{definition}
\newtheorem{definition}[theorem]{Definition}
\newtheorem{remark}[theorem]{Remark}

\endlocaldefs

\begin{document}
\begin{frontmatter}
\title{Signed FK percolation in Ising spin glasses}
\runtitle{FK percolation and spin glasses}
\begin{aug}
\author[A]{\fnms{Alberto}~\snm{Gandolfi}\ead[label=e1]{albertogandolfi@nyu.edu}}
\author[B]{\fnms{Charles}~\snm{Newman}\ead[label=e2]{cn3@nyu.edu}}
\author[B]{\fnms{Daniel}~\snm{Stein}\ead[label=e3]{daniel.stein@nyu.edu}}
\address[A]{Science Division, New York University Abu Dhabi\printead[presep={,\ }]{e1}}
\address[B]{Courant Institute School, New York University\printead[presep={,\ }]{e2,e3}}
\end{aug}

\begin{abstract}
We introduce signed Fortuin--Kasteleyn (FK) percolation for the \(\pm J\)
random-interaction Ising model to characterize nonzero uniform magnetization
density and two-replica overlap within the class of Gibbs distributions whose infinite-volume
Edwards--Sokal (ES)/FK representation has at most
one infinite FK cluster almost surely. The construction refines ordinary FK
connectivity by retaining the parity of coupling products along open paths.

Using the infinite-volume ES representation and its DLR equations, we study
pairs of companion objects and the compatible spin assignments on them.
In one replica, companions are defined by positive-parity FK
connections. For extremal states satisfying the infinite-cluster assumption
and almost-sure existence of companion densities along the chosen van~Hove
sequence, nonzero limiting expected magnetization density is equivalent to
density distinguishability of the companions of the infinite FK cluster.

For two independent replicas, companions are defined by FK connections
coherent in parity. For products of extremal one-replica states satisfying
the infinite-cluster assumption and almost-sure existence of the companion
densities along the chosen sequence, density distinguishability of the
companions is equivalent to nonzero limiting expected mutual overlap.
For two independent replicas drawn from the same extremal state, this
specializes to the Edwards--Anderson self-overlap.
\end{abstract}

\end{frontmatter}

\tableofcontents

\section{Introduction}

Although the onset of Fortuin-Kasteleyn (FK)~\cite{fortuin1972random} random cluster percolation signifies the 
transition from the paramagnetic to ferromagnetic phase 
in uniform Ising ferromagnets, it does not similarly imply the onset of a low-temperature broken-symmetric  phase in the Edwards-Anderson (EA) ~\cite{edwards1975theory}  Ising spin glass model on $\mathbb{Z}^d$~\cite{swendsen1987nonuniversal}.  Some special
cases have been studied: on the one hand,
the onset of {\it double\/} FK percolation captures the spin glass phase on the Bethe lattice~\cite{carlson1990bethe}. (By ``double FK percolation'' we mean percolation of the bonds occupied in both of two independent FK configurations.) On the other hand, numerical studies of two-replica representations in the EA model support a density-imbalance description of spin-glass ordering~\cite{machta2008percolation,munster2023cluster}.
Rigorous results relating this ordering to the densities of large clusters
in the infinite-range Sherrington--Kirkpatrick model are given
in~\cite{machta2008percolation}.
A related density-imbalance criterion is established for
Chayes--Machta--Redner (CMR) blue clusters by
Pei~\cite{pei2026blue}: persistent imbalance would imply distinct
spin-flip-related Gibbs states, while the existence of such imbalance remains open.
These examples motivate an FK description of both uniform magnetization
density and replica overlap in systems with competing ferro- and
antiferromagnetic interactions.

In addition, while the mathematical theory of percolation, even
in its bond-variables version, is by now very well developed
\cite{duminil2018sixty}, there seems to be no natural 
extension beyond cases with only positive or zero values.  In 
particular, while there is a well-developed deterministic theory of signed graphs and gain graphs, where signs or group labels are multiplied along paths \cite{harary1953balance,zaslavsky1982signed,zaslavsky1989biased}, there does not seem to be an established  percolation theory that combines signed graphs and FK distributions into a signed-connectivity theory.

In this paper, we develop signed percolation mechanisms for these two
order parameters. We use the infinite-volume Edwards--Sokal (ES) joint
spin--FK representation, defined in Section~\ref{sec:ES-infinite-volume},
and restrict to states with at most one infinite FK cluster almost surely.
Our main results characterize nonzero limiting expected uniform
magnetization density and mutual overlap. They concern extremal
one-replica Gibbs states, or products of such states, and assume
almost-sure existence of the relevant companion densities along a fixed
van~Hove sequence of volumes. Extremal states and the one-state order parameters
are defined in Section~\ref{subsec:order-parameters}; mutual overlap is
defined in Section~\ref{subsec:mutual-overlap}.
The precise characterizations are Theorems~\ref{EquivalenceOneCopy}
and~\ref{thm:product_extremal_pc_overlap}.
The signed connections and their companions are outlined below.

Our results expand the idea, first suggested in
\cite{biskup2000gibbs}, that spin order can be studied through the
infinite-volume ES joint spin--FK distribution
\cite{edwards1988generalization}. Under the hypotheses of the
order-parameter theorems, spin assignments on the infinite FK objects
determine the limiting expected uniform magnetization density and overlap.
We first develop the required infinite-volume ES theory for the
random-interaction Ising model, and then identify the companion objects
that enter these formulas. Projection onto the FK variables gives the
corresponding signed percolation criteria.

As already highlighted by several authors \cite{grimmett2004random,biskup2000gibbs}, for the
projection onto the FK bond variables to work properly 
we need uniqueness of the infinite FK cluster, a condition
which is  satisfied in several interesting and relevant contexts
\cite{gandolfi1992uniqueness} and plays a 
role similar to absolute convergence of interactions
in Gibbs distributions. 

With uniqueness in place, we dig deeper into the types of long-range vertex connections
which are generated by the FK-occupied bonds in either one copy (for uniform magnetization)
or two copies (for replica overlap).  We call such new 
long-range ``percolative'' mechanisms ``signed percolation'',
which can be either $+$connected (i.e., FK-connected along a path with either zero or else an even number of antiferromagnetic bonds,  also referred to as a positive-parity path), $-$connected (odd number of antiferromagnetic bonds, or a negative-parity path), or disconnected.
Transitivity of $+$connections and the fact
that two $-$connections in series determine a 
$+$connection imply that $+$clusters organize
into pairs of companion long-range clusters, as described in Section~\ref{SFFK11}.

We define signed FK percolation through well-defined, unequal densities of
companion sets associated with infinite FK objects. There are two
versions, described below, depending on whether we consider one or two
replicas. In either version, one companion may be finite or empty.

The first version of signed percolation we analyze,
useful for a single replica,
is  signed
ferromagnetic FK percolation, in which FK unfrustrated configurations are considered with
random $\pm J$ couplings, and the connectivity $r(x,y)$
between any two vertices $x,y$ is defined by 
\begin{eqnarray} \label{SignedFerrConn}
r(x,y)=\begin{cases}
    +1 \hskip 0.1 cm \text{ if  an FK-occupied ferromagnetic path connects 
     } x \text{ and } y  \\
     -1  \hskip 0.1 cm \text{ if  an FK-occupied antiferromagnetic path  connects
     } x \text{ and } y \\
     0  \hskip 0.1 cm \text{ if  no FK-occupied path connects
     } x \text{ and } y\, . 
\end{cases}
    \end{eqnarray}
By FK-occupied ferromagnetic path we mean an FK-connected path containing zero or an even number of antiferromagnetic bonds, and by FK-occupied antiferromagnetic path we mean an FK-connected path containing an odd number of antiferromagnetic bonds.

We consider signed ferromagnetic-FK-clusters, or sf-FK-clusters,
which are maximal sets of vertices 
which are pairwise signed ferromagnetically-FK-connected,
or sf-FK-connected;
i.e. $r(x,y)=1$ for each $x,y$ in the cluster.
We include the word "signed" in the name to distinguish this notion of
connection from what is generally understood as FK ferromagnetic connection,
which in turn involves only FK occupied ferromagnetic edges.
We will show that clusters come in pairs, in the 
sense that for each sf-FK-cluster $C$, 
there is a possibly empty companion set $\tilde C$, which is an
sf-FK-cluster when nonempty,
such that each pair of vertices, one in $C$ and one in $\tilde C$,
is signed antiferromagnetically FK-connected, or sa-FK-connected;
i.e. $r(x,y)=-1$ for each $x
\in C,y \in \tilde C$.
While an sf-FK-cluster $C$ need not be connected in the traditional
nearest-neighbor sense, the union $C \cup \tilde C$ is.
In addition, in the $d$-dimensional lattice 
$\mathbb Z^d$, under a joint periodic-boundary limit satisfying the
one-replica infinite-cluster assumption, at most one companion pair
contains an infinite sf-FK-cluster, almost surely. One member of this
pair may be finite or empty, so the number of infinite sf-FK-clusters
is $0$, $1$ or $2$.

We will say that signed ferromagnetic FK percolation occurs when the two
companions of the infinite FK cluster have well-defined, unequal densities
along the chosen sequence; see Definition~\ref{def:percolation-selector}.
A companion may be finite or empty; in that case the other must have
positive density. Other tail observables are discussed in
Section~\ref{SfPerc}.

The second version we consider is 
signed paired coherent  FK percolation.
Here, FK configurations of two replicas of the 
random interaction Ising model with
random $\pm J$ couplings are considered;
the connectivity $s(x,y)$
between any two vertices $x,y$ is defined
as follows. Let
\(x \xLeftrightarrow{c} y\) indicate
that $x$ and $y$ are
FK-connected in both copies with the same parity and \(x \xLeftrightarrow{i} y\) indicate
that $x$ and $y$ are
FK-connected in both copies with the opposite parity, where
"c" stands for coherent and "i" for incoherent; then
\begin{eqnarray} \label{SignedDoubleConn}
s(x,y)=\begin{cases}
    +1 \hskip 0.5 cm \text{ if } x \xLeftrightarrow{c} y, \\
     -1 \hskip 0.5 cm \text{ if } x \xLeftrightarrow{i} y, \\
     0 \hskip 0.5 cm \text{ otherwise }.
\end{cases}
\end{eqnarray}

This time we consider paired coherent FK-clusters, or pc-FK-clusters,
which are maximal sets of vertices
whose pairs are coherently connected in both replicas,
i.e. $s(x,y)=1$ for each $x,y$ in the cluster.
As before, for any pc-FK-cluster $C$
there is a possibly empty companion set $\tilde C$, which is a
pc-FK-cluster when nonempty,
such that each pair of vertices $x$ in $C$ and $y$ in $\tilde C$
is incoherently connected, i.e. $s(x,y)=-1$. 
This time, the union of a coherent cluster and its
companion need not be connected in the traditional sense.
However,  in the $d$-dimensional lattice 
$\mathbb Z^d$, under a joint two-replica periodic-boundary limit satisfying
the infinite-cluster assumption in each replica, at most one companion
pair contains an infinite coherent cluster, almost surely.
The union of such a pair is the vertex-set intersection of the two
infinite ordinary FK clusters; this is proved in
Lemma~\ref{lem:pairedAsIntersectionOfFKClusters}.
The number of infinite coherent clusters is again $0$, $1$ or $2$.
We restrict our discussion to states satisfying the one-replica
infinite-cluster assumptions.

We say that paired coherent FK percolation occurs when the coherent
companions have well-defined, unequal densities along the chosen sequence;
see Definition~\ref{def:subseq-pc-density-imbalance}.

M{\"u}nster and Weigel~\cite{munster2026cluster} report numerical
evidence for a related density-imbalance mechanism in the
three-dimensional $\pm J$ model: two large
Chayes--Machta--Redner--J{\"o}rg (CMRJ, or blue) clusters develop unequal
densities at the disordered ferromagnetic or spin-glass transition,
below their percolation temperature.
CMRJ bonds can be occupied only when the coupling sign agrees with
the product of the endpoint spins in both replicas; their occupation
probability differs from that for the common occupied bonds of two
independent FK configurations, another construction studied there.
The paired connectivity used here allows different occupied paths in
the two replicas, and coherence compares their coupling-sign products.
Consequently, the coherent companions need not coincide with either
of these bond clusters.
Theorem~\ref{thm:product_extremal_pc_overlap} gives an exact identity
for the limiting expected mutual overlap in terms of our companion
densities under the stated extremality, infinite-cluster and
density-existence hypotheses, with the diagonal Edwards--Anderson
specialization in Corollary~\ref{cor:diagonal_pc_EA}.

The sf-FK and pc-FK companions relate spin assignments at infinity to
the order parameters considered here. When the two companion densities
exist and differ, the denser companion provides a reference for choosing
a spin orientation in one replica, or an agreement orientation in two
replicas. The density
selectors are defined in Sections~\ref{SfPerc}
and~\ref{subsec:pc-percolation}. Their stability under finite
modifications permits the corresponding orientations to be imposed
without changing the local ES specification.
Under the hypotheses of Theorems~\ref{EquivalenceOneCopy}
and~\ref{thm:product_extremal_pc_overlap}, equal companion densities
give a vanishing corresponding limiting expected order parameter.
Equality of densities is a property of the FK geometry; it neither
imposes spin-flip invariance on the ES measure nor establishes Gibbs
uniqueness. The phase interpretation and the implications for coexistence
are stated in Section~\ref{subsec:order-parameters}.

Distinguishability by density, or even by other tail observables, need not
in general imply that the probability of the origin belonging to one companion
is different from the probability of belonging to the other. When such an
origin-probability imbalance is available, the signed ferromagnetic mechanism
recovers the corresponding single-site magnetization. In the two-replica setting, the
corresponding statement is formulated for products of extremal one-replica
states satisfying the hypotheses of
Theorem~\ref{thm:product_extremal_pc_overlap}: paired coherent
percolation is equivalent to a nonzero limiting
expected mutual overlap along the chosen sequence. Choosing the relative
orientation gives the absolute mutual overlap of the two chosen states.
For two independent replicas drawn from the same extremal state, this is
the Edwards--Anderson self-overlap.
Corollary~\ref{cor:diagonal_pc_EA} gives this diagonal specialization and
its disorder average.

As a side remark, we observe that the FK measures of interest here
have $q = 2$, but  the same
signed percolation phenomena can be studied for other values of $q$. In particular,
for $q = 1$, in finite volume this amounts to Bernoulli occupation
variables carrying the prescribed coupling signs, conditioned on no
frustrated open cycle. Signed ferromagnetic percolation in an
infinite-volume version is again defined by well-defined, unequal
densities of the companions of the infinite cluster. When the origin
belongs to that cluster, these are the densities of the vertices
connected to it with positive and negative parity, in either order.
This is a new percolation mechanism, which we call ``signed
percolation'' and which extends the classical notion of Bernoulli
percolation.

\section{Sketch of the argument}

Our aim is to characterize nonzero limiting expected uniform magnetization
density and mutual overlap by FK connectivity, under the infinite-cluster
and density-existence hypotheses stated in
Theorems~\ref{EquivalenceOneCopy}
and~\ref{thm:product_extremal_pc_overlap}. The construction also extends
the multiplicative interpretation of bond percolation to variables
that can take negative values.

We begin with the one-replica magnetization.
As mentioned, 
we extend the joint spin--FK Edwards--Sokal representation to infinite
volume through the DLR equations, following the viewpoint developed for
ferromagnetic interactions and related graphical representations in
\cite{biskup2000gibbs}, and the rest of the work
is carried out in infinite volume.

Using the fact
that the local specifications
of the ES representation  depend on boundary spins but not on the exterior
FK bond configuration, which makes it possible to separate local DLR
properties from tail choices made on infinite FK objects,
Theorem~\ref{SpinProjSincleCopy} first shows the one-to-one correspondence between spin Gibbs
states and infinite-volume ES measures.

To pass from ES measures to their FK marginals, we then impose the \(0/1\)-infinite-cluster property stating that the number of infinite FK clusters is at most one almost
surely. Under this assumption, Theorem~\ref{thm:FK-surjective} shows that
the FK marginal map from ES states to FK Gibbs states is 
surjective. Thus,
within this class, the FK configuration can be used as the primary object of
study, with its possible ES colorings supplying the orientation choices
used in the order-parameter formulas.

As the colorings, or spin assignments, must be consistent
with products of the couplings along each path, it is natural
to introduce the $r(x,y)$ connectivity mentioned above, 
under which the usual
FK clusters then split into two signed companions. 
The density selector is stable outside finite sets under local bond
modifications (Lemma~\ref{lem:tail-pql-density-selector1}). Its
nonemptiness event, and the spin orientation on the selected companion,
are tail-measurable modulo null sets. One can therefore condition on
these events without changing the local ES specification. As mentioned in the Introduction, it is then
natural to define sf--FK percolation, an extension of standard percolation,
 as the
density imbalance of the two companions.

Returning to the analysis of states, the ES state then decomposes canonically according to its infinite signed
companion event. This is the tail-event decomposition of
Theorem~\ref{thm:canonical_sf_decomposition}. For extremal ES states,
Corollary~\ref{cor:extremal_sf_sector} implies that exactly one tail event is
chosen. Combining this decomposition with the spin--ES isomorphism gives the
FK characterization of nonzero uniform magnetization density:
Theorem~\ref{EquivalenceOneCopy} states that, assuming almost-sure
existence of both companion densities along the chosen sequence,
nonzero limiting expected magnetization density is equivalent to
sf--FK percolation along that sequence in the FK marginal.

The mutual overlap is handled in parallel, using two independent
replicas. Two vertices are paired coherently connected if they are FK-connected
in both replicas with the same parity, and paired incoherently connected if the
two parities are opposite. The paired FK cluster is the intersection of the
ordinary FK clusters in the two replicas, and it splits into two coherent
companions. Under the corresponding \(0/1\)-infinite-cluster property,
the two-replica density selector is again stable outside finite sets
under local bond modifications (Lemma~\ref{lem:tail-pql-density-selector3}),
and its nonemptiness event is a tail event modulo null sets.

The two-replica FK marginal theorem is the product analogue of the one-replica
statement: the FK marginal map sends products of one-replica ES states satisfying the
\(0/1\)-infinite-cluster property
onto products of one-replica FK Gibbs states satisfying the same \(0/1\)-infinite-cluster property.
The paired-coherent tail-event decomposition then identifies the possible
orientations of the density-selected infinite coherent companions.
For products of extremal one-replica Gibbs states satisfying the
infinite-cluster and density-existence hypotheses,
Theorem~\ref{thm:product_extremal_pc_overlap} shows that pc--FK
percolation is equivalent to nonzero limiting expected mutual overlap.
Theorem~\ref{Th1_replica_product_extremal} gives the corresponding
oriented averaged mutual overlap for measurable families of extremal states.
For two independent replicas drawn from the same extremal state,
Corollary~\ref{cor:diagonal_pc_EA} identifies the expected density imbalance with
the Edwards--Anderson self-overlap and gives its disorder average.
For the joint periodic-boundary limits constructed in
Lemma~\ref{lem:periodic-joint-ES},
Corollary~\ref{cor:periodic_pc_overlap} gives the corresponding formula
for the expected absolute value of the sampled limiting overlap.

The proofs of the main results follow this roadmap:
\begin{itemize}
\item Theorem~\ref{SpinProjSincleCopy}: spin Gibbs states and ES Gibbs states
are affinely isomorphic.
\item Theorem~\ref{thm:FK-surjective}: under the \(0/1\)-infinite-cluster property, FK Gibbs
states are FK marginals of ES states.
\item Lemma~\ref{lem:tail-pql-density-selector1} and
Lemma~\ref{lem:tail-pql-density-selector3}: the one-replica and two-replica
density selectors are stable outside finite sets under local bond modifications.
\item Theorem~\ref{thm:canonical_sf_decomposition}: ES states decompose by conditioning on
tail events determined by the signed infinite companions.
\item Theorem~\ref{EquivalenceOneCopy}: along a sequence with almost-sure
companion-density limits, sf--FK percolation is equivalent to nonzero
limiting expected magnetization density for extremal one-replica states.
\item Theorem~\ref{thm:product_extremal_pc_overlap}: pc--FK percolation is
equivalent to nonzero limiting expected mutual overlap for products of extremal
one-replica states under the stated infinite-cluster and
density-existence hypotheses.
\item Theorem~\ref{Th1_replica_product_extremal}: the oriented averaged
mutual overlap is given by the expected pc--FK companion-density difference.
\item Corollary~\ref{cor:diagonal_pc_EA}: for two independent replicas
drawn from the same extremal state, the density formula gives the
Edwards--Anderson self-overlap and its disorder average.
\item Corollary~\ref{cor:periodic_pc_overlap}: for a stationary joint
periodic-boundary limit satisfying the one-replica \(0/1\) properties,
the expected absolute sampled overlap equals the mean coherent-companion
density imbalance.
\end{itemize}

The same framework also suggests other possible transitions, based on tail
observables other than density, and signed percolation mechanisms outside the
random-interaction Ising model. We do not pursue those extensions here.

\section{The \texorpdfstring{$\pm 1$}{+/-1} random-interaction Ising model on an infinite graph and its FK representation}\label{sec:model}

\subsection{Summary of notation}\label{subsec:notation}

We fix a countably infinite, locally finite graph $G=(V,E)$.
Throughout the paper, the inverse temperature and bond parameter satisfy
\[
   0\le\beta<\infty,
   \qquad
   p=1-e^{-2\beta}\in[0,1).
\]
For a finite volume $\Lambda\Subset V$, let $\partial\Lambda$ denote
its (exterior) vertex boundary.
Let
\[
   E_{\Lambda}:=\bigl\{\{x,y\}\in E:\ x,y\in\Lambda\bigr\},
   \qquad
   \overline E_{\Lambda}:=\bigl\{\{x,y\}\in E:\ \{x,y\}\cap\Lambda\neq\varnothing\bigr\}.
\]
Thus $\overline E_{\Lambda}$ contains in particular the edges with one endpoint
in $\Lambda$ and the other in $\partial\Lambda$.

\begin{itemize}
\item \textbf{Couplings in $\Lambda$:} $J_{E_{\Lambda}}=\{J_{xy}\}_{\{x,y\}\in E_\Lambda}$.

\item \textbf{Coupling law:} $P_{\lambda}$ is the product measure on
$\{\pm1\}^{E}$ with
$P_\lambda(J_{xy}=+1)=\lambda$ and $P_\lambda(J_{xy}=-1)=1-\lambda$.
Its restriction to a finite edge set $F$ is denoted by $P_{\lambda,F}$.

\item \textbf{Spins:} $\omega_{\Lambda}\in\{-1,+1\}^{\Lambda}$; \textbf{spin boundary condition:}
$\omega_{\partial\Lambda}\in\{-1,+1\}^{\partial\Lambda}$.

\item \textbf{Spin Gibbs measures (quenched):} $\mu_{J,\beta,\Lambda}$ (free b.c.) and
$\mu_{J,\beta,\Lambda,\omega_{\partial\Lambda}}$ (where $\omega_{\partial\Lambda}$ are the boundary spins).

\item \textbf{FK bonds:} $\eta_F\in\{0,1\}^{F}$, with $F=E_\Lambda$
for free boundary conditions and $F=\overline E_\Lambda$ when boundary
conditions are specified.

\item \textbf{ES measures (quenched):} $\pi_{J,\beta,\Lambda}$ on
$(\omega_\Lambda,\eta_{E_\Lambda})$ for free boundary conditions, and
$\pi_{J,\beta,\Lambda,\omega_{\partial\Lambda}}$ on
$(\omega_\Lambda,\eta_{\overline E_\Lambda})$ for fixed boundary spins.

\item \textbf{FK random-cluster measures (quenched):} $\nu_{J,p,\Lambda}$
on $\eta_{E_\Lambda}$ for free boundary conditions, and
$\nu_{J,p,\Lambda,\zeta}$ on $\eta_{\overline E_\Lambda}$ for exterior
bond configuration $\zeta=\eta_{E\setminus\overline E_\Lambda}$.

\item \textbf{Joint disorder--ES law (free boundary conditions):}
\[
\rho_{\lambda,\beta,\Lambda}(dJ,d\omega_\Lambda,d\eta_{E_\Lambda})
:=P_\lambda(dJ)\,\pi_{J,\beta,\Lambda}(d\omega_\Lambda,d\eta_{E_\Lambda}).
\]

\item \textbf{Disorder-averaged spin--bond law (free boundary conditions):}
\[
\overline\rho_{\lambda,\beta,\Lambda}(d\omega_\Lambda,d\eta_{E_\Lambda})
:=\int\pi_{J,\beta,\Lambda}(d\omega_\Lambda,d\eta_{E_\Lambda})\,P_\lambda(dJ).
\]
This averages normalized quenched measures, not unnormalized Gibbs weights.
\end{itemize}

Infinite-volume limits are taken along exhausting sequences of finite
volumes, and the dependence on $\Lambda$ is then dropped from the notation.
The van~Hove assumption will be imposed on the sequences used for density
averages.

\subsection{The random-interaction Ising model}
We recap the usual definitions. Let \( G = (V, E) \) be an infinite graph with vertex set \( V \) and edge set \( E \). The random interaction Ising model on \( G \) is defined by random couplings \( J = \{J_{xy}\}_{\{x,y\} \in E} \) distributed according to a product measure \( P \), where each \( J_{xy} \) takes values in \(\{-1, +1\}\) independently with probability:
\(
P_{\lambda}(J_{xy} = +1) = \lambda =1- P_{\lambda}(J_{xy} = -1) 
\)
for some \( \lambda \in [0,1] \).

Fix a realization of the $J_{xy}$'s. For a finite subgraph with vertices $\Lambda
\subset V$ and  edges \(
E_{\Lambda}\;=\;\bigl\{\{x,y\}\in E \,\colon\, x,y\in\Lambda \bigr\}
\), a spin configuration  
\(
\omega_{\Lambda} = \{\omega_x\}_{x\in\Lambda}\; \in\;\{-1,+1\}^{\Lambda}
\) is assigned energy through the Hamiltonian:
\[
H_{J,\Lambda}(\omega_{\Lambda},\overline {\omega})
  \;=\;
  -\sum_{\{x,y\}\in E_{\Lambda}} J_{xy}\,\omega_x\omega_y
  -\sum_{\{x,y\}\in E, x \in \Lambda, y \notin \Lambda} J_{xy}\,\omega_x\overline {\omega}_y,
\]
where $\overline {\omega}$ is the boundary condition
defined on $\partial \Lambda$.
The corresponding Gibbs measure at inverse temperature \( \beta \geq 0 \) is given by:
\begin{equation}
    \label{eq:Gibbs}
\mu_{J,\beta,\Lambda, \overline \omega}(\omega_{\Lambda})
  \;=\;
  \frac{1}{Z_{J,\beta,\Lambda, \overline \omega}}\,
  e^{-\beta H_{J,\Lambda}(\omega_{\Lambda},\overline {\omega})},
\end{equation} 
 where
\[Z_{J,\beta,\Lambda,\overline {\omega}}
  \;=\!
  \sum_{\omega_{\Lambda}\in\{-1,+1\}^{\Lambda}}
  e^{-\beta H_{J,\Lambda}(\omega_{\Lambda},\overline {\omega})}\]
  is the partition function.

  We now consider an increasing sequence of finite subgraphs 
\[
\Lambda_1 \;\subset\; \Lambda_2 \;\subset\; \cdots 
\subset\; \Lambda_n \;\subset\; \cdots
\]
such that 
\(
\bigcup_{n=1}^\infty \Lambda_n \;=\; V
\) and \(\lim_{n \rightarrow\infty} 
|\Lambda_{n+1}|/|\Lambda_{n}|=1.\)
When a sequence is used for spatial density statements, it is also assumed
to be van~Hove: for every fixed integer $r\ge1$,
\[
 \frac{|\{x\in\Lambda_n:\operatorname{dist}(x,V\setminus\Lambda_n)\le r\}|}
 {|\Lambda_n|}\longrightarrow0.
\]
Such a sequence is an additional geometric assumption, not a property of
every locally finite graph. On $\mathbb Z^d$, standard cubes satisfy it.

For each fixed realization of the couplings $J$, we consider weak subsequential
limits of the finite-volume Gibbs measures, allowing deterministic boundary
conditions that may vary with the volume. We write
$\mathcal G_{J,\beta}^{\mathrm{SPIN}}$ for the infinite-volume Gibbs measures
for $J$ and $\beta$, defined by the DLR equations with kernels
\eqref{eq:Gibbs}. Equivalently, this set is the closed convex hull of the
finite-volume limit points just described; compare
Theorem~\ref{SpinProjSincleCopy}.

\subsection{Order parameters and phase interpretation}
\label{subsec:order-parameters}

For phase interpretations on \(\mathbb Z^d\), fix the couplings \(J\),
the inverse temperature \(\beta\), and a deterministic van~Hove sequence
\(\Gamma=(\Lambda_n)_{n\ge1}\) as above. A Gibbs state is
\emph{extremal} if it cannot be expressed as a nontrivial convex
combination of distinct Gibbs states. For an extremal state
\(\mu\in\mathcal G^{\mathrm{SPIN}}_{J,\beta}\), write
\(m_x(\mu):=\mu(\omega_x)\) and, whenever the limits exist, define
\[
\begin{aligned}
 m_\Gamma(\mu)
 &:=
 \lim_{n\to\infty}\frac1{|\Lambda_n|}
       \sum_{x\in\Lambda_n}m_x(\mu),\\
 q_{\mathrm{EA},\Gamma}(\mu)
 &:=
 \lim_{n\to\infty}\frac1{|\Lambda_n|}
       \sum_{x\in\Lambda_n}m_x(\mu)^2.
\end{aligned}
\]
The first quantity is the limiting expected uniform magnetization
density. The second is the Edwards--Anderson self-overlap of the chosen
extremal state; its interpretation through two independent samples from
that same state is given in Corollary~\ref{cor:diagonal_pc_EA}.
For every finite volume, Cauchy--Schwarz bounds the square of the average
of the local means by the average of their squares. Hence
\[
                 m_\Gamma(\mu)^2\le q_{\mathrm{EA},\Gamma}(\mu).
\]
Thus uniform ferromagnetic order, \(m_\Gamma(\mu)\neq0\), implies
positive self-overlap. A nonzero local mean \(m_x(\mu)\) at one vertex
does not by itself establish nonzero uniform magnetization density.

Theorem~\ref{EquivalenceOneCopy} and
Corollary~\ref{cor:diagonal_pc_EA} establish existence of these limits
and characterize their nonvanishing under their respective hypotheses
on infinite clusters and companion densities. The corresponding
percolation notions are introduced in
Definitions~\ref{def:percolation-selector} and
\ref{def:subseq-pc-density-imbalance}. Each characterization refers
to the stated class of states and the same density sequence \(\Gamma\).

At the level of the disorder parameters, uniform ferromagnetic order
is detected when an extremal state in the relevant class with
\(m_\Gamma\neq0\) exists with positive \(P_\lambda\)-probability.
Positive self-overlap is detected analogously using
\(q_{\mathrm{EA},\Gamma}>0\).
Its interpretation as specifically spin-glass order concerns regimes
without uniform ferromagnetic or other conventional order.
The absence of uniform ferromagnetic order in such a regime concerns
all states in the specified class, rather than a single chosen state.
In particular, the range \(\lambda\in[0,1]\) also includes
antiferromagnetism. At \(\lambda=0\) and sufficiently low temperature
in \(d\ge2\), there are extremal states whose local means alternate
between the two sublattices, with zero uniform magnetization density
and positive self-overlap.

We reserve the term \emph{paramagnetic phase} for Gibbs uniqueness
for \(P_\lambda\)-almost every \(J\). The interaction is invariant
under the global spin flip \(F\), defined by
\((F\omega)_x=-\omega_x\). Uniqueness therefore implies that every
local mean vanishes, and hence that both order parameters vanish.
If either \(m_\Gamma(\mu)\neq0\) or
\(q_{\mathrm{EA},\Gamma}(\mu)>0\), then \(\mu\) and \(F_*\mu\)
are distinct extremal Gibbs states. Vanishing of these order parameters
alone does not establish Gibbs uniqueness.

\subsection{The finite-volume ES measure and FK representation}

The Fortuin--Kasteleyn (FK) representation introduces edge-occupation
variables
\(
   \eta_F=\{\eta_b\}_{b\in F}\in\{0,1\}^F
\)
for finite sets of edges \(F\). In finite volume without spin boundary
conditions we take \(F=E_\Lambda\). With spin boundary conditions on the
exterior vertex boundary \(\partial\Lambda\), we take
\[
   F=\overline E_\Lambda
   :=
   \bigl\{\{x,y\}\in E:\{x,y\}\cap\Lambda\neq\varnothing\bigr\}.
\]
For a bond configuration \(\eta_F\), write
\[
   |\eta^{(1)}|:=|\{b\in F:\eta_b=1\}|,\qquad
   |\eta^{(0)}|:=|\{b\in F:\eta_b=0\}|.
\]
If \(\omega\) is a spin configuration on a set of vertices containing all
endpoints of the open edges of \(\eta_F\), we write
\(
   \eta_F\sim_J \omega
\)
when every open edge \(b=\{x,y\}\in F\) satisfies
\(
   J_{xy}\omega_x\omega_y=+1.
\)
Equivalently, the open FK bonds are satisfied by the spin configuration
\(\omega\).

We also write
\(
   \eta_F\sim J
\)
when the open bond configuration \(\eta_F\) is compatible with the couplings
\(J\), i.e. when every open cycle has coupling product \(+1\). Equivalently,
there exists a spin assignment on each open cluster satisfying all open
edges.

\medskip
The free-boundary finite-volume Edwards--Sokal measure is the probability measure on
\(
   \{-1,+1\}^\Lambda\times\{0,1\}^{E_\Lambda}
\)
given by
\begin{equation}\label{eq:freeES}
\pi_{J,\beta,\Lambda}(\omega_\Lambda,\eta_{E_\Lambda})
=
\frac{
   p^{|\eta^{(1)}|}(1-p)^{|\eta^{(0)}|}
   \mathbf 1_{\{\eta_{E_\Lambda}\sim_J \omega_\Lambda\}}
}{
   Z^{\mathrm{ES}}_{J,\beta,\Lambda}
},
\qquad
p=1-e^{-2\beta}.
\end{equation}
Summing \eqref{eq:freeES} over the bond variables gives the usual Ising
Gibbs measure with free boundary condition:
\[
   \sum_{\eta_{E_\Lambda}\in\{0,1\}^{E_\Lambda}}
   \pi_{J,\beta,\Lambda}(\omega_\Lambda,\eta_{E_\Lambda})
   =
   \mu_{J,\beta,\Lambda}(\omega_\Lambda).
\]

Summing \eqref{eq:freeES} over the spin variables gives the FK
random-cluster marginal
\begin{equation}\label{FiniteFK}
\nu_{J,p,\Lambda}(\eta_{E_\Lambda})
=
\frac{
   p^{|\eta^{(1)}|}(1-p)^{|\eta^{(0)}|}
   2^{\mathrm{Cl}(\eta_{E_\Lambda})}
   \mathbf 1_{\{\eta_{E_\Lambda}\sim J\}}
}{
   \widetilde Z_{J,p,\Lambda}
},
\end{equation}
where \(\mathrm{Cl}(\eta_{E_\Lambda})\) is the number of open FK clusters
of vertices in \(\Lambda\), isolated vertices included. The factor
\(2^{\mathrm{Cl}(\eta_{E_\Lambda})}\) is the number of possible spin
assignments to the open FK clusters, provided that the open bonds are
compatible with \(J\).

\medskip

We next include spin boundary conditions.
If a full exterior configuration
\(
   (\overline\omega_{V\setminus\Lambda},\eta_{E\setminus\overline E_\Lambda})
\)
is given, then the conditional ES distribution of
\((\omega_\Lambda,\eta_{\overline E_\Lambda})\) is
the
probability measure \(
   \pi_{J,\beta,\Lambda,\bar\omega_{\partial\Lambda}}
   (\omega_\Lambda,\eta_{\overline E_\Lambda}),
\)
on
\(
   \{-1,+1\}^{\Lambda}\times\{0,1\}^{\overline E_\Lambda}
\)
where \(\overline\omega_{\partial\Lambda}\) is the restriction of the exterior spin
configuration to \(\partial\Lambda\). It is the case, though,
that, as the above notation suggests, the conditional ES
kernel depends on the exterior configuration only through the boundary
spins \(\omega_{\partial\Lambda}\), and not through the exterior FK bonds.
Therefore, if \(\omega_\Lambda\bar\omega_{\partial\Lambda}\) denotes the
spin configuration on \(\Lambda\cup\partial\Lambda\) obtained by joining
\(\omega_\Lambda\) and \(\bar\omega_{\partial\Lambda}\), the finite-volume ES
kernel with spin boundary condition \(\bar\omega_{\partial\Lambda}\) has  
the expression
\begin{equation}\label{ESCond}
\pi_{J,\beta,\Lambda,\bar\omega_{\partial\Lambda}}
   (\omega_\Lambda,\eta_{\overline E_\Lambda})
=
\frac{
   p^{|\eta^{(1)}|}(1-p)^{|\eta^{(0)}|}
   \mathbf 1_{\{
      \eta_{\overline E_\Lambda}
      \sim_J
      \omega_\Lambda\bar\omega_{\partial\Lambda}
   \}}
}{
   Z^{\mathrm{ES}}_{J,\beta,\Lambda,\bar\omega_{\partial\Lambda}}
}.
\end{equation}

The spin marginal of \eqref{ESCond} is the usual finite-volume spin Gibbs
measure with boundary condition:
\begin{equation}\label{CondSpinMarg}
\sum_{\eta_{\overline E_\Lambda}\in\{0,1\}^{\overline E_\Lambda}}
\pi_{J,\beta,\Lambda,\bar\omega_{\partial\Lambda}}
   (\omega_\Lambda,\eta_{\overline E_\Lambda})
=
\mu_{J,\beta,\Lambda,\bar\omega_{\partial\Lambda}}(\omega_\Lambda).
\end{equation}

\medskip

We now describe the corresponding conditional FK kernel. Let
\(
   \zeta\in\{0,1\}^{E\setminus\overline E_\Lambda}
\)
be a fixed exterior FK configuration with no frustrated open cycle. For
\(
   \eta\in\{0,1\}^{\overline E_\Lambda},
\)
write \(\eta\zeta\) for the concatenated bond configuration on all of \(E\).

We need the relative cluster count \(K_\Lambda(\eta,\zeta)\). Let
\(
   \mathscr B_\Lambda(\zeta)
\)
be the set of open FK clusters of the exterior configuration \(\zeta\) in
\(V\setminus\Lambda\) that intersect \(\partial\Lambda\). Singleton boundary
vertices are included as clusters when no exterior open edge is incident to
them. Since \(G\) is locally finite and \(\Lambda\) is finite,
\(\mathscr B_\Lambda(\zeta)\) is finite. Write
\(
   m_\Lambda(\zeta):=|\mathscr B_\Lambda(\zeta)|.
\)

Form a finite quotient graph  \(Q_\Lambda(\eta,\zeta)\)  as follows. Its vertices are the vertices of
\(\Lambda\), together with one extra vertex for each exterior cluster
\(B\in\mathscr B_\Lambda(\zeta)\); each extra vertex is 
only a representative of the exterior cluster, and its
choice or precise location are irrelevant. Given
the FK configuration 
\(\eta_{\overline E_\Lambda}\), if there is
an FK open edge between two vertices \(\{x,y\}\),
then we create a virtual configuration in which 
open virtual FK edges
are inserted as follows:
\begin{itemize}
\item if  \(x,y\in\Lambda\), then a virtual open edge \(\{x,y\}\in E_\Lambda\)  joins
      \(x\) and \(y\);
\item if \(x\in\Lambda\) and \(y\in\partial\Lambda\),
then a virtual open edge \(\{x,y\}\in\overline E_\Lambda\setminus E_\Lambda\) joins \(x\) to the
      vertex representing the exterior cluster \(B\in\mathscr B_\Lambda(\zeta)\)
      containing \(y\).
\end{itemize}
Let
\(
   N_\Lambda(\eta,\zeta)
\)
be the number of connected components of this finite virtual graph \(Q_\Lambda(\eta,\zeta)\), isolated
vertices included. Define
\begin{equation}\label{eq:relativeK}
   K_\Lambda(\eta,\zeta)
   :=
   N_\Lambda(\eta,\zeta)-m_\Lambda(\zeta).
\end{equation}
Thus \(K_\Lambda\) is a relative number of clusters: it counts the number of
spin-coloring degrees of freedom created inside \(\Lambda\), relative to the
degrees of freedom already present in the exterior boundary clusters. In
particular, \(K_\Lambda(\eta,\zeta)\) may be negative when the inside bonds
wire together distinct exterior clusters.

The conditional FK random-cluster kernel with exterior bond condition
\(\zeta\) is
\begin{equation}\label{CondFiniteFK}
\nu_{J,p,\Lambda,\zeta}(\eta_{\overline E_\Lambda})
=
\frac{
   p^{|\eta^{(1)}|}(1-p)^{|\eta^{(0)}|}
   2^{K_\Lambda(\eta,\zeta)}
   \mathbf 1_{\{\eta\zeta\sim J\}}
}{
   \widetilde Z_{J,p,\Lambda,\zeta}
}.
\end{equation}
Here \(\eta\zeta\sim J\) means that the full open bond configuration obtained
by concatenating \(\eta\) and \(\zeta\) has no frustrated open cycle. The
normalizing constant \(\widetilde Z_{J,p,\Lambda,\zeta}\) depends on
\(J,p,\Lambda\), and \(\zeta\). It is positive for \(p<1\), since
closing all bonds in \(\overline E_\Lambda\) is compatible with such
an exterior configuration and has positive weight.

\begin{remark}\label{ESSpinDepOnly}
The conditional FK distribution \eqref{CondFiniteFK} depends on the exterior
FK configuration through the wiring of exterior clusters, encoded by
\(K_\Lambda(\eta,\zeta)\), and through the compatibility constraint
\(\eta\zeta\sim J\). By contrast, the conditional ES kernel
\eqref{ESCond} depends on the exterior configuration only through the
boundary spins \(\omega_{\partial\Lambda}\). This is the key property that
allows us later to separate local DLR specifications from tail choices made
on infinite FK objects.
\end{remark}

We finally record the precise relation between the ES and FK conditional
kernels. Fix an exterior FK configuration \(\zeta\) compatible with \(J\).
For each exterior boundary cluster \(B\in\mathscr B_\Lambda(\zeta)\), choose
one representative vertex \(y_B\in B\). A choice of signs
\[
   \tau=(\tau_B)_{B\in\mathscr B_\Lambda(\zeta)}
   \in\{-1,+1\}^{\mathscr B_\Lambda(\zeta)}
\]
determines a boundary spin condition
\[
   \bar\omega^\tau_{\partial\Lambda}
\]
as follows: if \(y\in\partial\Lambda\cap B\), then
\[
   \bar\omega^\tau_y
   =
   \tau_B
   \prod_{e\in\gamma_{y_B,y}} J_e,
\]
where \(\gamma_{y_B,y}\) is any \(\zeta\)-open path in \(B\) from \(y_B\) to
\(y\). This is well-defined because \(\zeta\sim J\).

Let
\[
   W^{\mathrm{ES}}_{J,p,\Lambda,\bar\omega_{\partial\Lambda}}
   (\omega_\Lambda,\eta_{\overline E_\Lambda})
   :=
   p^{|\eta^{(1)}|}(1-p)^{|\eta^{(0)}|}
   \mathbf 1_{\{
      \eta_{\overline E_\Lambda}
      \sim_J
      \omega_\Lambda\bar\omega_{\partial\Lambda}
   \}}
\]
denote the unnormalized ES weight. Then, for every
\(\eta\in\{0,1\}^{\overline E_\Lambda}\),
\begin{equation}\label{CondFiniteFK2}
\begin{aligned}&\frac{1}{2^{m_\Lambda(\zeta)}}
\sum_{\tau\in\{-1,+1\}^{\mathscr B_\Lambda(\zeta)}}
\sum_{\omega_\Lambda\in\{-1,+1\}^{\Lambda}}
W^{\mathrm{ES}}_{J,p,\Lambda,\bar\omega^\tau_{\partial\Lambda}}
   (\omega_\Lambda,\eta_{\overline E_\Lambda})\\[3pt]
&\qquad=
p^{|\eta^{(1)}|}(1-p)^{|\eta^{(0)}|}
2^{K_\Lambda(\eta,\zeta)}
\mathbf 1_{\{\eta\zeta\sim J\}}.
\end{aligned}
\end{equation}
Indeed, if \(\eta\zeta\) is not compatible with \(J\), both sides vanish. If
\(\eta\zeta\sim J\), then the number of compatible spin assignments on the
finite quotient graph defined above is \(2^{N_\Lambda(\eta,\zeta)}\), while
the exterior boundary clusters already carry \(2^{m_\Lambda(\zeta)}\)
possible colorings. Dividing by \(2^{m_\Lambda(\zeta)}\) gives the relative
factor \(2^{K_\Lambda(\eta,\zeta)}\).

After normalization, \eqref{CondFiniteFK2} gives exactly the conditional FK
kernel \eqref{CondFiniteFK}. This is the precise sense in which the FK
conditional distribution is obtained from the ES representation.

In the rest of the paper we use FK path and FK cluster for paths or clusters
made of edges \(b\) with \(\eta_b=1\).

\subsection{Infinite--volume Edwards--Sokal measures for fixed couplings \texorpdfstring{$J$}{J}}
\label{sec:ES-infinite-volume}

We now develop the infinite-volume theory of the FK distributions and their
related ES joint distributions. It is convenient, as in
\cite{biskup2000gibbs}, to start from the latter.

Let \(G=(V,E)\) be the underlying graph, for instance \(V=\mathbb Z^d\).
Fix an inverse temperature \(0\le\beta<\infty\), set
\(p=1-e^{-2\beta}<1\), and fix a coupling configuration
\(
   J=\{J_{xy}\}_{\{x,y\}\in E}\in\{\pm1\}^E.
\)
Write
\[
   \Omega=\{-1,+1\}^V,
   \qquad
   \mathcal H=\{0,1\}^E.
\]

For a finite \(\Lambda\Subset V\), recall that
\(
   \overline E_\Lambda
   =
   \bigl\{\{x,y\}\in E:\{x,y\}\cap\Lambda\neq\varnothing\bigr\}.
\)
We use the exterior sigma-field
\[
   \mathscr A_{\Lambda^c}^{\mathrm{ES}}
   :=
   \sigma\!\left(
      \omega_x:x\in V\setminus\Lambda;\,
      \eta_e:e\in E\setminus\overline E_\Lambda
   \right).
\]
The ES specification in finite volume is the family of kernels
\(
   \left\{
      \pi_{J,\beta,\Lambda,\omega_{\partial\Lambda}}
   \right\}_{\Lambda\Subset V},
\)
where \(\pi_{J,\beta,\Lambda,\omega_{\partial\Lambda}}\) is the finite-volume
ES kernel from \eqref{ESCond}. Notice that, by
Remark~\ref{ESSpinDepOnly}, the ES kernel depends on the exterior
configuration only through the boundary spins \(\omega_{\partial\Lambda}\);
the exterior bond configuration does not enter the ES kernel.

For a bounded measurable function \(f\) on \(\Omega\times\mathcal H\), define
the local ES kernel operator
\[
   \Gamma^{\mathrm{ES}}_\Lambda f(\omega,\eta)
   :=
   \sum_{\omega'_\Lambda\in\{-1,+1\}^{\Lambda}}
   \sum_{\eta'_{\overline E_\Lambda}\in\{0,1\}^{\overline E_\Lambda}}
   f\!\left(
      \omega'_\Lambda\omega_{V\setminus\Lambda},
      \eta'_{\overline E_\Lambda}\eta_{E\setminus\overline E_\Lambda}
   \right)
   \pi_{J,\beta,\Lambda,\omega_{\partial\Lambda}}
   \left(
      \omega'_\Lambda,\eta'_{\overline E_\Lambda}
   \right).
\]
Thus \(\Gamma^{\mathrm{ES}}_\Lambda f\) is
\(\mathscr A_{\Lambda^c}^{\mathrm{ES}}\)-measurable.

\begin{definition}[Infinite-volume ES measures]
\label{def:ES-DLR}
A probability measure
\(
   \pi\in\mathcal P(\Omega\times\mathcal H)
\)
is an \emph{infinite-volume Edwards--Sokal measure} for the fixed coupling
configuration \(J\) and inverse temperature \(\beta\) if, for every finite
\(\Lambda\Subset V\) and every bounded local function \(f\),
\begin{equation}\label{ESDLR}
   \int f\,d\pi
   =
   \int \Gamma^{\mathrm{ES}}_\Lambda f\,d\pi,
\end{equation}
or, equivalently,
\(
   \pi\!\left(
      f \mid \mathscr A_{\Lambda^c}^{\mathrm{ES}}
   \right)
   =
   \Gamma^{\mathrm{ES}}_\Lambda f
  \hskip 0.25cm \pi\text{-a.s.}
\).

We denote the set of all such measures by
\(
   \mathcal G^{\mathrm{ES}}_{J,\beta}
   \quad\text{or equivalently}\quad
   \mathcal G^{\mathrm{ES}}_{J,p}.
\)
\end{definition}

We next recall the basic structure of the ES Gibbs simplex. For a probability
measure \(\pi\) on \(\Omega\times\mathcal H\), let
\(
   \Pi_{\mathrm{SPIN}}\pi
\)
denote its spin marginal.

\begin{theorem}\label{SpinProjSincleCopy}
Let \(J\in\{-1,1\}^E\) be arbitrary and let
\(\mathcal G^{\mathrm{ES}}_{J,\beta}\) be the set of infinite-volume ES
measures from Definition~\ref{def:ES-DLR}. Then:

\begin{enumerate}
\item[\textup{(a)}]
\(\mathcal G^{\mathrm{ES}}_{J,\beta}\) is non-empty, compact and convex in the
weak* topology on \(\mathcal P(\Omega\times\mathcal H)\).

\item[\textup{(b)}]
\(\mathcal G^{\mathrm{ES}}_{J,\beta}\) is the closed convex hull of the weak
limits of finite-volume ES measures
\(
   \pi_{J,\beta,\Lambda,\omega_{\partial\Lambda}} .
\)

\item[\textup{(c)}]
\(\mathcal G^{\mathrm{ES}}_{J,\beta}\) is a Choquet simplex: every
\(\pi\in\mathcal G^{\mathrm{ES}}_{J,\beta}\) has a unique decomposition
\[
   \pi
   =
   \int_{\operatorname{ex}\mathcal G^{\mathrm{ES}}_{J,\beta}}
      \widetilde\pi\,
      \alpha_\pi(d\widetilde\pi),
\]
where \(\operatorname{ex}\mathcal G^{\mathrm{ES}}_{J,\beta}\) denotes the set of
tail-trivial, equivalently extremal, ES Gibbs states.

\item[\textup{(d)}]
The projection
\(
   \Pi_{\mathrm{SPIN}}:(\omega,\eta)\mapsto\omega
\)
is an affine isomorphism between the Choquet simplices
\(
   \mathcal G^{\mathrm{ES}}_{J,\beta}
   \quad\text{and}\quad
   \mathcal G^{\mathrm{SPIN}}_{J,\beta}.
\)
In particular,
\[
   |\mathcal G^{\mathrm{SPIN}}_{J,\beta}|=1
   \quad\Longleftrightarrow\quad
   |\mathcal G^{\mathrm{ES}}_{J,\beta}|=1.
\]
\end{enumerate}
\end{theorem}

\begin{proof}
The ES specification is a finite-range specification on the countable
index set \(V\cup E\), with finite single-site spaces and with a local hard-core
compatibility constraint: configurations in which an open FK edge is
incompatible with the spins have weight zero. The finite-volume kernels are
quasilocal because
\(\pi_{J,\beta,\Lambda,\omega_{\partial\Lambda}}\) depends on the exterior only
through the finitely many boundary spins \(\omega_{\partial\Lambda}\).

\smallskip
\noindent\textbf{(a) Non-emptiness, convexity and compactness.}
Non-emptiness follows from the general existence theorem for Gibbs states
for quasilocal specifications with hard cores. Equivalently, one can take weak
limit points of finite-volume ES measures along an exhausting sequence of
finite volumes; compactness of the product configuration space gives subsequential
limits, and quasilocality and consistency of the specification imply that such
limits satisfy \eqref{ESDLR}.

Convexity is immediate from the linearity of the DLR equations. If
\(\pi_1,\pi_2\in\mathcal G^{\mathrm{ES}}_{J,\beta}\) and \(t\in[0,1]\), then, for
every bounded local \(f\),
\[
   \int f\,d\bigl(t\pi_1+(1-t)\pi_2\bigr)
   =
   t\int f\,d\pi_1+(1-t)\int f\,d\pi_2
   =
   \int \Gamma^{\mathrm{ES}}_\Lambda f\,
   d\bigl(t\pi_1+(1-t)\pi_2\bigr).
\]

For compactness, note that \(\Omega\times\mathcal H\) is compact and
metrizable, hence \(\mathcal P(\Omega\times\mathcal H)\) is weakly compact.
For fixed \(\Lambda\) and local \(f\), the map
\[
   \pi\mapsto \int f\,d\pi-\int \Gamma^{\mathrm{ES}}_\Lambda f\,d\pi
\]
is continuous, since \(\Gamma^{\mathrm{ES}}_\Lambda f\) is again a bounded
local function. Therefore the set of measures satisfying the DLR equation
for a given \((\Lambda,f)\) is closed. Intersecting over a countable
determining class of local functions and finite volumes gives a closed subset
of \(\mathcal P(\Omega\times\mathcal H)\), hence compact.

\smallskip
\noindent\textbf{(b) Closed convex hull of finite-volume weak limits.}
This is exactly the standard Gibbs-state statement applied to the present ES
specification. More precisely, the ES kernels form a system of conditional
probabilities in the sense of the DLR theory. Theorem~1.9
in \cite{ruelle2004thermodynamic} gives the
closed-convex-hull characterization of Gibbs states by thermodynamic limits,
and Section~2.8 of the same reference shows that the relevant parts of Theorem~1.9 remain
valid for systems of conditional probabilities. Therefore
\(
   \mathcal G^{\mathrm{ES}}_{J,\beta}
   =
   \overline{\operatorname{co}}
   \left(
      \mathcal L^{\mathrm{ES}}_{J,\beta}
   \right),
\)
where \(\mathcal L^{\mathrm{ES}}_{J,\beta}\) is the set of all weak limits of
finite-volume ES measures
\(
   \pi_{J,\beta,\Lambda_n,\omega^{(n)}_{\partial\Lambda_n}} .
\)
This proves part~\textup{(b)}.

\smallskip
\noindent\textbf{(c) Choquet simplex.}
The same general DLR theory implies that the set of Gibbs measures for a
quasilocal finite-range specification with hard cores is a Choquet simplex,
and that its extremal points are precisely the tail-trivial Gibbs measures.
Applying this to the ES specification gives the stated unique integral
decomposition.

\smallskip
\noindent\textbf{(d) Spin--ES isomorphism.}
We give the explicit construction.

First, let
\(
   \pi\in\mathcal G^{\mathrm{ES}}_{J,\beta}
\)
and set
\(
   \mu:=\Pi_{\mathrm{SPIN}}\pi.
\)
Summing the finite-volume ES kernel over the bond variables gives the
finite-volume spin Gibbs kernel:
\[
   \sum_{\eta_{\overline E_\Lambda}}
   \pi_{J,\beta,\Lambda,\omega_{\partial\Lambda}}
      (\omega_\Lambda,\eta_{\overline E_\Lambda})
   =
   \mu_{J,\beta,\Lambda,\omega_{\partial\Lambda}}(\omega_\Lambda).
\]
Therefore, applying the ES DLR equation to functions depending only on the
spins gives the usual spin DLR equation for \(\mu\). Hence
\(
   \Pi_{\mathrm{SPIN}}\pi\in
   \mathcal G^{\mathrm{SPIN}}_{J,\beta}.
\)

Conversely, let
\(
   \mu\in\mathcal G^{\mathrm{SPIN}}_{J,\beta}.
\)
We define its ES lift explicitly. Given a spin configuration
\(\omega\in\Omega\), let the bond variables be conditionally independent with
single-edge conditional law
\[
   \mathbb P(\eta_{xy}=1\mid\omega)
   =
   p\,\mathbf 1_{\{J_{xy}\omega_x\omega_y=+1\}},
\]
and
\[
   \mathbb P(\eta_{xy}=0\mid\omega)
   =
   1-p\,\mathbf 1_{\{J_{xy}\omega_x\omega_y=+1\}}.
\]
Equivalently, an unsatisfied edge is forced closed, while a satisfied edge is
open with probability \(p\) and closed with probability \(1-p\). Let
\(
   \kappa_J(d\eta\mid\omega)
\)
be the corresponding product measure on \(\mathcal H\), and define
\[
   \mathcal L(\mu)(d\omega,d\eta)
   :=
   \mu(d\omega)\,\kappa_J(d\eta\mid\omega).
\]
We claim that
\(
   \mathcal L(\mu)\in\mathcal G^{\mathrm{ES}}_{J,\beta}.
\)

Indeed, fix \(\Lambda\Subset V\). Under \(\mathcal L(\mu)\), conditional on
the exterior variables in \(\mathscr A_{\Lambda^c}^{\mathrm{ES}}\), the spin
configuration inside \(\Lambda\) has law
\(
   \mu_{J,\beta,\Lambda,\omega_{\partial\Lambda}}
\)
by the spin DLR equations for \(\mu\). Once the inside spins are sampled, the
bonds in \(\overline E_\Lambda\) are sampled independently, opened with
probability \(p\) precisely on satisfied edges and forced closed otherwise.
This is exactly the finite-volume ES kernel
\(
   \pi_{J,\beta,\Lambda,\omega_{\partial\Lambda}}.
\)
Therefore \(\mathcal L(\mu)\) satisfies \eqref{ESDLR}.

Moreover,
\(
   \Pi_{\mathrm{SPIN}}\mathcal L(\mu)=\mu
\)
by construction.

It remains to prove uniqueness of the lift. Let
\(
   \pi\in\mathcal G^{\mathrm{ES}}_{J,\beta}
\)
and let
\(
   \mu:=\Pi_{\mathrm{SPIN}}\pi.
\)
We show that \(\pi=\mathcal L(\mu)\). Fix a finite set of edges
\(F\subset E\), and choose \(\Lambda\Subset V\) so large that
\(F\subset\overline E_\Lambda\). By the ES DLR equation, conditional on the
spin configuration, the law of \(\eta_F\) is the product of the one-edge
conditional laws above: edges satisfied by the spins are opened independently
with probability \(p\), and unsatisfied edges are forced closed. Since this
holds for every finite \(F\), the conditional law of all FK bonds given the
spin configuration is \(\kappa_J(d\eta\mid\omega)\). Hence
\[
   \pi(d\omega,d\eta)
   =
   \mu(d\omega)\,\kappa_J(d\eta\mid\omega)
   =
   \mathcal L(\mu)(d\omega,d\eta).
\]
Thus the ES lift is unique.

We have shown that
\[
   \Pi_{\mathrm{SPIN}}
   :
   \mathcal G^{\mathrm{ES}}_{J,\beta}
   \longrightarrow
   \mathcal G^{\mathrm{SPIN}}_{J,\beta},\qquad
   \mathcal L
   :
   \mathcal G^{\mathrm{SPIN}}_{J,\beta}
   \longrightarrow
   \mathcal G^{\mathrm{ES}}_{J,\beta}
\]
are inverse affine maps. Therefore \(\Pi_{\mathrm{SPIN}}\) is an affine
isomorphism of the two Choquet simplices. The equivalence of uniqueness of
the spin and ES Gibbs states follows immediately.
\end{proof}

It turns out that to consider the projection
on the FK edge variables which will be 
discussed in the next section, 
we need 
the technical 
 requirement 
that, with probability one under the ES measure, the number of infinite FK clusters is at most
one, see \cite{biskup2000gibbs},  a condition
also 
identified as the $0/1$-property \cite{grimmett2004random}.
The $0/1$ property is an assumption on the chosen state and need not
hold on an arbitrary locally finite graph. Our FK-projection and
order-parameter results concern states for which it does hold.
This includes states with no infinite FK cluster.
In some instances, this happens with $P_{\lambda}$-probability
one, as proven for instance in \cite{gandolfi1992uniqueness} for 
the quenched random cluster distributions obtained with periodic boundary conditions in $\mathbb Z^d$. 
For fixed $\beta$ and one such $J$, 
we then further restrict to  
\[
   \mathcal G^{\mathrm{ES}}_{\le 1}
   \;:=\;
   \bigl\{
        \pi\in\mathcal G^{\mathrm{ES}}_{J,\beta}:
        \;\pi\bigl(N_{\infty}\le1\bigr)=1
   \bigr\},\]
throughout the paper.

Before continuing, we make some simple observations about uniqueness. The
following lemma is implicitly used in \cite{biskup2000gibbs}.

\begin{lemma}\label{UniqueOutside}
Fix \(p=1-e^{-2\beta}<1\). Let \(\nu\) be either an FK Gibbs state or
the FK marginal of an ES Gibbs state for the fixed couplings \(J\), and
suppose that \(\nu(N_\infty\le1)=1\). For every finite edge set
\(F\subset E\), deleting the bonds in \(F\) leaves at most one infinite
open component, \(\nu\)-almost surely. In particular, this holds for
\(F=\overline E_\Lambda\) for every finite \(\Lambda\Subset V\).
Finite exterior components may remain.
\end{lemma}

\begin{proof}
For an FK Gibbs state, closing a bond never violates compatibility and
cannot decrease the relative cluster count. Its conditional probability
of being closed, given all other bonds, is therefore at least \(1-p\).
For an ES marginal, the same bound follows by first conditioning on the
spins: a satisfied bond is closed with probability \(1-p\), and an
unsatisfied bond is forced closed. Iterating over \(F\) gives
\(\nu(\eta_e=0\text{ for all }e\in F\mid\eta_{E\setminus F})
\ge(1-p)^{|F|}>0\), almost surely.

The event that the open graph on \(E\setminus F\) has at least two
infinite components is measurable with respect to the bonds outside
\(F\). If it had positive probability, the preceding bound would give
positive probability to the same event with every bond of \(F\) closed.
The full configuration would then have at least two infinite FK clusters,
contrary to \(\nu(N_\infty\le1)=1\). This proves the assertion for each
\(F\); countability gives a common full-measure set for all finite edge
sets.
\end{proof}

In addition, the two-arm events occur with vanishing probabilities for
large volumes:
\begin{corollary}\label{NoTwoLongClusters}
Under the hypotheses of Lemma~\ref{UniqueOutside}, fix
\(\Lambda\Subset V\). For finite \(M\supset\Lambda\cup\partial\Lambda\),
let \(A_{\Lambda,M}\) be the event that the open graph induced by
\((M\cup\partial M)\setminus\Lambda\) has at least two distinct
clusters, each meeting both \(\partial\Lambda\) and \(\partial M\).
Then \(\lim_{M\uparrow V}\nu(A_{\Lambda,M})=0\).
\end{corollary}

\begin{proof}
By Lemma~\ref{UniqueOutside}, with
\(F=\overline E_\Lambda\), the exterior open graph has at most one
infinite component, almost surely. Only finitely many exterior clusters
meet the finite set \(\partial\Lambda\). For all sufficiently large
\(M\), each finite such cluster, together with its neighboring vertices,
is contained in \(M\), and hence cannot meet \(\partial M\).
If the infinite exterior component meets \(\partial\Lambda\), connect
all of its vertices in \(\partial\Lambda\) by finitely many open paths
in \(V\setminus\Lambda\). These paths are contained in \(M\setminus
\Lambda\) for all sufficiently large \(M\). Thus all inner-boundary
vertices that can reach \(\partial M\) belong to a single annular
cluster. Consequently \(\mathbf1_{A_{\Lambda,M}}\to0\) almost surely,
and bounded convergence proves the claim. No monotonicity of
\(A_{\Lambda,M}\) is needed.
\end{proof}

\subsection{Infinite-volume FK measures}
We now define the infinite-volume FK random-cluster measures using the same
conditional kernels as in the finite-volume discussion above.  For a finite
\(\Lambda\Subset V\), let
\(
   \mathscr E_\Lambda
   :=
   \sigma\bigl(\eta_e:e\in E\setminus\overline E_\Lambda\bigr)
\)
be the exterior FK sigma-field.  If
\(
   \zeta\in\{0,1\}^{E\setminus\overline E_\Lambda}
\)
is an exterior bond configuration compatible with \(J\) and
\(
   \xi\in\{0,1\}^{\overline E_\Lambda},
\)
we use the conditional FK kernel
\begin{equation}\label{eq:InfiniteFKKernel}
   \nu_{J,p,\Lambda,\zeta}(\xi)
   =
   \frac{
      p^{|\xi_1|}(1-p)^{|\xi_0|}
      2^{K_\Lambda(\xi,\zeta)}
      \mathbf 1_{\{\xi\zeta\sim J\}}
   }{
      Z^{\mathrm{FK}}_{J,p,\Lambda,\zeta}
   } .
\end{equation}
Here \(K_\Lambda(\xi,\zeta)\) is the relative cluster count defined in
\eqref{eq:relativeK}, and \(\xi\zeta\sim J\) means that the concatenated open
bond configuration contains no frustrated open cycle.  Thus the infinite-volume
FK specification is exactly the conditional random-cluster specification
obtained by conditioning on the exterior bonds in
\(E\setminus\overline E_\Lambda\).

A probability measure \(\nu\) on the compatible bond space
\(\mathcal H_J:=\{\eta\in\mathcal H:\eta\sim J\}\) is an infinite-volume FK Gibbs
state for the fixed couplings \(J\) and parameter \(p<1\) if, for every finite
\(\Lambda\Subset V\) and every bounded local function
\(f=f(\eta_{\overline E_\Lambda})\),
\begin{equation}\label{eq:FKDLR}
   \nu\bigl(f\mid\mathscr E_\Lambda\bigr)(\eta)
   =
   \nu_{J,p,\Lambda,\eta_{E\setminus\overline E_\Lambda}}(f)
   \qquad \nu\text{-a.s.}
\end{equation}
The set of FK distributions satisfying these DLR equations is denoted by
\(\mathcal G^{\mathrm{FK}}_{J,p}\).

With uniqueness of the infinite cluster in place,
all the infinite-volume FK distributions can also be obtained by
projection of the infinite-volume ES measures.
For a configuration of bonds $\eta\in\mathcal H$
write $N_{\infty}(\eta)$ for the number of
infinite FK clusters in $\eta$. Let

\[
   \mathcal G^{\mathrm{FK}}_{\le 1}
   \;:=\;
   \bigl\{
        \nu\in\mathcal G^{\mathrm{FK}}_{J,p}:
        \;\nu\bigl(N_{\infty}\le1\bigr)=1
   \bigr\}.
\]
In view of the one-to-one result in Theorem \ref{SpinProjSincleCopy},
Part (d),
this implies that we are also restricting the spin measures
to the  image $\Pi_{\mathrm{SPIN}}(\mathcal G^{\mathrm{ES}}_{\le 1})=
\mathcal G^{\mathrm{SPIN}}_{J,\beta,\leq 1}$.

Let
\[
\Pi_{\mathrm{FK}}:\Omega\times\mathcal H\to\mathcal H,
\qquad
\Pi_{\mathrm{FK}}(\omega,\eta):=\eta,
\]
denote the projection onto bond variables.

\begin{lemma}[Finite quotient coloring identity]
\label{lem:finite-quotient-coloring}
Fix a finite \(\Lambda\Subset V\), a coupling configuration
\(J\in\{\pm1\}^E\), and an exterior FK configuration
\(
   \zeta\in\{0,1\}^{E\setminus\overline E_\Lambda}
\)
which is compatible with \(J\).

Using the notation introduced in and before 
\eqref{eq:relativeK},
 the following two identities hold.

\begin{enumerate}
\item[\textup{(a)}]
Suppose that at most one cluster in \(\mathscr B_\Lambda(\zeta)\) has its
spin fixed, and that the spins of all other exterior clusters in
\(\mathscr B_\Lambda(\zeta)\) are averaged with the uniform distribution.
Then
\begin{eqnarray}\label{eq:quotient-identity-averaged}
 &&  \frac{1}{2^{m_\Lambda(\zeta)-r}}
   \sum_{\text{free exterior signs}}
   \sum_{\omega_\Lambda\in\{-1,+1\}^{\Lambda}}
   p^{|\xi_1|}(1-p)^{|\xi_0|}
   \mathbf 1_{\{
      \xi\sim_J \omega_\Lambda\omega_{\partial\Lambda}
   \}} \\
  && \hskip 2cm =
   p^{|\xi_1|}(1-p)^{|\xi_0|}
   2^{K_\Lambda(\xi,\zeta)}
   \mathbf 1_{\{\xi\zeta\sim J\}}, \nonumber
\end{eqnarray}
where \(r=0\) if no exterior cluster has fixed spin and \(r=1\) if one
exterior cluster has fixed spin. The boundary spin
\(\omega_{\partial\Lambda}\) is the spin assignment induced on
\(\partial\Lambda\) by the chosen exterior cluster signs and the coupling
configuration \(J\).

\item[\textup{(b)}]
Fix instead a spin assignment on all exterior clusters in
\(\mathscr B_\Lambda(\zeta)\), and let
\(
   \omega_{\partial\Lambda}
\)
be the induced boundary spin configuration. Under the fair coloring of the
connected components of \(Q_\Lambda(\xi,\zeta)\), the joint probability of
seeing these exterior cluster signs and a given
\(\omega_\Lambda\in\{-1,+1\}^{\Lambda}\) is
\(
   2^{-N_\Lambda(\xi,\zeta)}
   \mathbf 1_{\{\xi\sim_J\omega_\Lambda\omega_{\partial\Lambda}\}} .
\)
Consequently,
\begin{equation}\label{eq:quotient-identity-fixed}
   2^{K_\Lambda(\xi,\zeta)}
   \mathbf 1_{\{\xi\zeta\sim J\}}
   \cdot
   2^{-N_\Lambda(\xi,\zeta)}
   \mathbf 1_{\{\xi\sim_J\omega_\Lambda\omega_{\partial\Lambda}\}}
   =
   2^{-m_\Lambda(\zeta)}
   \mathbf 1_{\{\xi\sim_J\omega_\Lambda\omega_{\partial\Lambda}\}} .
\end{equation}
In particular, after normalization over
\(\xi\) and \(\omega_\Lambda\), the FK relative factor
\(2^{K_\Lambda(\xi,\zeta)}\) cancels exactly against the cluster-coloring
probability, leaving the ES weight
\[
   p^{|\xi_1|}(1-p)^{|\xi_0|}
   \mathbf 1_{\{\xi\sim_J\omega_\Lambda\omega_{\partial\Lambda}\}} .
\]
\end{enumerate}
\end{lemma}

\begin{proof}
The open edges of \(\xi\), together with the exterior clusters of \(\zeta\)
meeting \(\partial\Lambda\), define a finite signed quotient graph
\(Q_\Lambda(\xi,\zeta)\). A compatible spin assignment on this quotient graph
exists if and only if the concatenated FK configuration \(\xi\zeta\) has no
frustrated open cycle, i.e. if and only if
\(
   \xi\zeta\sim J.
\)
If this compatibility condition fails, all expressions in
\eqref{eq:quotient-identity-averaged} and
\eqref{eq:quotient-identity-fixed} vanish.

Assume now that \(\xi\zeta\sim J\). Then every connected component of
\(Q_\Lambda(\xi,\zeta)\) has exactly two possible spin assignments, unless
one spin in that component has already been fixed, in which case the whole
component has a unique compatible assignment.

For part \textup{(a)}, suppose first that no exterior cluster spin is fixed.
Then the total number of compatible assignments of the spins on
\(\Lambda\) and on the \(m_\Lambda(\zeta)\) exterior quotient vertices is
\(
   2^{N_\Lambda(\xi,\zeta)}.
\)
Averaging over the \(2^{m_\Lambda(\zeta)}\) possible exterior signs gives the
relative factor
\[
   2^{N_\Lambda(\xi,\zeta)-m_\Lambda(\zeta)}
   =
   2^{K_\Lambda(\xi,\zeta)}.
\]
If one exterior cluster spin is fixed, then the number of compatible
assignments is \(2^{N_\Lambda(\xi,\zeta)-1}\), while the number of exterior
signs still being averaged is \(2^{m_\Lambda(\zeta)-1}\). The quotient is
again
\[
   2^{N_\Lambda(\xi,\zeta)-m_\Lambda(\zeta)}
   =
   2^{K_\Lambda(\xi,\zeta)}.
\]
Multiplying by the common edge weight
\(p^{|\xi_1|}(1-p)^{|\xi_0|}\) proves
\eqref{eq:quotient-identity-averaged}.

For part \textup{(b)}, fix all exterior quotient spins and also a spin
configuration \(\omega_\Lambda\) inside \(\Lambda\). Under fair coloring of
the components of \(Q_\Lambda(\xi,\zeta)\), each connected component has one
independent fair sign. Hence the probability of realizing a prescribed
compatible spin assignment on the quotient graph is
\(
   2^{-N_\Lambda(\xi,\zeta)}.
\)
It is zero if the prescribed spins are incompatible with the open edges of
\(\xi\). Thus the joint probability is exactly
\(
   2^{-N_\Lambda(\xi,\zeta)}
   \mathbf 1_{\{\xi\sim_J\omega_\Lambda\omega_{\partial\Lambda}\}} .
\)
Multiplying by \(2^{K_\Lambda(\xi,\zeta)}\) gives
\(
   2^{N_\Lambda(\xi,\zeta)-m_\Lambda(\zeta)}
   2^{-N_\Lambda(\xi,\zeta)}
   =
   2^{-m_\Lambda(\zeta)},
\)
which proves \eqref{eq:quotient-identity-fixed}.
\end{proof}

\begin{theorem}[Surjectivity of the FK marginal at \(N_\infty\le1\)]
\label{thm:FK-surjective}
For every fixed coupling configuration \(J\in\{\pm1\}^{E}\) and finite
inverse temperature \(0\le\beta<\infty\), with \(p=1-e^{-2\beta}<1\),
\[
   \Pi_{\mathrm{FK}}
   \bigl(
      \mathcal G^{\mathrm{ES}}_{\le1}
   \bigr)
   =
   \mathcal G^{\mathrm{FK}}_{\le1}.
\]
Equivalently, every FK Gibbs state with at most one infinite FK cluster
almost surely is the FK marginal of an ES Gibbs state for the same coupling
configuration \(J\).
\end{theorem}

\begin{remark}
The proof follows the strategy of Biskup et al.\
\cite[Theorem II.2(i)]{biskup2000gibbs}. The only extra point in the present
signed setting is that all finite-volume calculations have to retain the
hard constraint excluding frustrated open cycles. This is exactly what
Lemma~\ref{lem:finite-quotient-coloring} records.
\end{remark}

\begin{proof}
Fix \(J\in\{\pm1\}^E\), \(0\le\beta<\infty\), and set
\(
   p=1-e^{-2\beta}<1.
\)
For a finite \(\Lambda\Subset V\), let
\(
   \mathscr E_\Lambda
   :=
   \sigma\bigl(\eta_b:\ b\in E\setminus\overline E_\Lambda\bigr)
\)
be the sigma-field generated by the FK variables outside
\(\overline E_\Lambda\). 

If
\(
   \zeta\in\{0,1\}^{E\setminus\overline E_\Lambda}
\)
is an exterior FK configuration, we write
\(
   \nu_{J,p,\Lambda,\zeta}
\)
for the finite-volume FK conditional kernel from
\eqref{CondFiniteFK}:
\[
   \nu_{J,p,\Lambda,\zeta}(\xi)
   =
   \frac{
      p^{|\xi_1|}(1-p)^{|\xi_0|}
      2^{K_\Lambda(\xi,\zeta)}
      \mathbf 1_{\{\xi\zeta\sim J\}}
   }{
      Z^{\mathrm{FK}}_{J,p,\Lambda,\zeta}
   },
   \qquad
   \xi\in\{0,1\}^{\overline E_\Lambda}.
\]

\smallskip
\noindent
\emph{Step 1: the FK marginal of an ES state is an FK Gibbs state.}

Let
\(
   \pi\in\mathcal G^{\mathrm{ES}}_{J,\beta},
   \pi(N_\infty\le1)=1,
\)
and set
\(
   \nu:=\Pi_{\mathrm{FK}}\pi.
\)
Then \(\nu(N_\infty\le1)=1\). We prove that
\(
   \nu\in\mathcal G^{\mathrm{FK}}_{J,p}.
\)

Fix a finite \(\Lambda\Subset V\) and a bounded function
\(
   f=f(\eta_{\overline E_\Lambda}).
\)
We must show that
\begin{equation}\label{eq:FK-DLR-to-prove-new}
   \mathbb E_\nu
   \left[
      f(\eta_{\overline E_\Lambda})
      \,\middle|\,
      \mathscr E_\Lambda
   \right]
   =
   \nu_{J,p,\Lambda,\eta_{E\setminus\overline E_\Lambda}}(f)
   \qquad
   \nu\text{-a.s.}
\end{equation}

Let
\(
   \zeta:=\eta_{E\setminus\overline E_\Lambda}.
\)
Work on the full \(\nu\)-measure event supplied by
Lemma~\ref{UniqueOutside}, applied to the ES marginal \(\nu\) with
\(F=\overline E_\Lambda\). The lemma applies directly to ES marginals,
so this use does not presuppose the FK DLR property being proved.

On this event, after removing the edges in \(\overline E_\Lambda\), the
infinite FK cluster, if it exists, has at most one infinite exterior component, all other exterior FK clusters of \(\zeta\) meeting
\(\partial\Lambda\) being finite. Since only finitely many exterior clusters meet
\(\partial\Lambda\), there exists a finite set
\[
   \Delta=\Delta(\eta,\Lambda)\Subset V,
   \qquad
   \Lambda\subset \Delta,
\]
which contains all finite exterior FK clusters of \(\zeta\) meeting
\(\partial\Lambda\). There is at most one exterior cluster meeting
\(\partial\Lambda\) that is not contained in \(\Delta\), and this cluster is
the unique infinite one, if it exists.

The set \(\Delta\) is a measurable function of the exterior bond
configuration \(\zeta\). Since there are only countably many finite choices
for \(\Delta\), it is enough to argue on the event
\(
   \{\Delta(\eta,\Lambda)=D\}
\)
for a fixed finite \(D\Subset V\).

Condition on
\[
   \omega_{D^c}
   \quad\text{and}\quad
   \eta_{E\setminus\overline E_\Lambda}=\zeta.
\]
By the ES DLR equations for \(\pi\), the conditional distribution inside
\(D\), given the exterior spins, is the finite-volume ES kernel in \(D\).
We then keep fixed all FK edges in
\(
   E\setminus\overline E_\Lambda
\)
and sum over the spin variables in \(D\). The finite
exterior FK clusters meeting \(\partial\Lambda\) have been included in
\(D\), so their signs are averaged. If an infinite exterior FK cluster meets
\(\partial\Lambda\), its sign is fixed by \(\omega_{D^c}\); by uniqueness
there is at most one such fixed exterior cluster.

Therefore Lemma~\ref{lem:finite-quotient-coloring}\textup{(a)} applies.
It gives that the conditional distribution of
\(\eta_{\overline E_\Lambda}\), after summing over the spin variables, is
exactly
\(
   \nu_{J,p,\Lambda,\zeta}.
\)

In particular, the result does not depend on \(\omega_{D^c}\). Hence, on
\(\{\Delta=D\}\),
\[
   \mathbb E_\pi
   \left[
      f(\eta_{\overline E_\Lambda})
      \,\middle|\,
      \omega_{D^c},\eta_{E\setminus\overline E_\Lambda}
   \right]
   =
   \nu_{J,p,\Lambda,\eta_{E\setminus\overline E_\Lambda}}(f).
\]
Taking conditional expectation with respect to \(\mathscr E_\Lambda\) and
summing over the countably many possible choices of \(D\), we obtain
\[
   \mathbb E_\nu
   \left[
      f(\eta_{\overline E_\Lambda})
      \,\middle|\,
      \mathscr E_\Lambda
   \right]
   =
   \nu_{J,p,\Lambda,\eta_{E\setminus\overline E_\Lambda}}(f).
\]
This proves \eqref{eq:FK-DLR-to-prove-new}. Since \(\Lambda\) and \(f\)
were arbitrary,
\(
   \nu\in\mathcal G^{\mathrm{FK}}_{J,p}
\)
and
\(
   \Pi_{\mathrm{FK}}
   \bigl(
      \mathcal G^{\mathrm{ES}}_{\le1}
   \bigr)
   \subset
   \mathcal G^{\mathrm{FK}}_{\le1}.
\)

\smallskip
\noindent
\emph{Step 2: every FK Gibbs state with \(N_\infty\le1\) has an ES lift.}

Let
\(
   \nu\in\mathcal G^{\mathrm{FK}}_{\le1}.
\)
We construct
\(
   \pi\in\mathcal G^{\mathrm{ES}}_{\le1}
\)
such that
\(
   \Pi_{\mathrm{FK}}\pi=\nu.
\)

By the FK hard constraint, \(\nu\)-a.s. every FK configuration
\(\eta\) is compatible with \(J\): its open graph contains no frustrated open
cycle. For such an \(\eta\), every open FK cluster has exactly two compatible
spin assignments. Fix an enumeration \(V=\{v_1,v_2,\ldots\}\).
For each FK cluster \(C\), finite or infinite, let \(x_C\) be its first
vertex in this enumeration. For
\(y\in C\), define
\(
   \tau_C(y)
   :=
   \prod_{b\in\gamma_{x_C,y}} J_b,
\)
where \(\gamma_{x_C,y}\) is any FK-open path in \(C\) from \(x_C\) to \(y\).
This is well-defined because \(\eta\) has no frustrated open cycle. The two
compatible spin assignments on \(C\) are
\[
   \omega_y=\varepsilon_C\,\tau_C(y),
   \qquad y\in C,
\]
where \(\varepsilon_C\in\{-1,+1\}\).

Given \(\eta\), choose the signs
\(
   \{\varepsilon_C:C\text{ an FK cluster of }\eta\}
\)
independently, each with probability \(1/2\). Let
\(
   \kappa_J(d\omega\mid\eta)
\)
be the corresponding coloring kernel, and define
\(
   \pi(d\omega,d\eta)
   :=
   \kappa_J(d\omega\mid\eta)\,\nu(d\eta).
\)
Then
\(
   \Pi_{\mathrm{FK}}\pi=\nu
\)
by construction, and therefore
\(
   \pi(N_\infty\le1)=\nu(N_\infty\le1)=1.
\)

It remains to prove that \(\pi\) satisfies the ES DLR equations. Fix
\(\Lambda\Subset V\), and let
\[
   \mathscr A_{\Lambda^c}^{\mathrm{ES}}
   :=
   \sigma\bigl(
      \omega_x:x\in V\setminus\Lambda;\,
      \eta_b:b\in E\setminus\overline E_\Lambda
   \bigr)
\]
be the ES exterior sigma-field.

We show that, under \(\pi\), the conditional law of
\(
   (\omega_\Lambda,\eta_{\overline E_\Lambda})
\)
given \(\mathscr A_{\Lambda^c}^{\mathrm{ES}}\) is
\(
   \pi_{J,\beta,\Lambda,\omega_{\partial\Lambda}}.
\)
Let
\(
   \zeta:=\eta_{E\setminus\overline E_\Lambda}
\)
and fix an exterior spin configuration
\(
   \alpha:=\omega_{V\setminus\Lambda}
\)
which is compatible with \(\zeta\). This compatibility holds
\(\pi\)-a.s. by construction.

For
\[
   \xi\in\{0,1\}^{\overline E_\Lambda},
   \qquad
   \sigma\in\{-1,+1\}^{\Lambda},
\]
we compute the conditional weight of
\(
   \eta_{\overline E_\Lambda}=\xi,\hskip 0.2 cm
   \omega_\Lambda=\sigma,
\)
given
\(
   \eta_{E\setminus\overline E_\Lambda}=\zeta,\hskip 0.2 cm
   \omega_{V\setminus\Lambda}=\alpha.
\)

By the FK DLR equation for \(\nu\), the conditional FK weight of \(\xi\),
given \(\zeta\), is proportional to
\[
   p^{|\xi_1|}(1-p)^{|\xi_0|}
   2^{K_\Lambda(\xi,\zeta)}
   \mathbf 1_{\{\xi\zeta\sim J\}}.
\]
Given the full FK configuration \(\xi\zeta\), the coloring kernel chooses
one fair sign for each connected component of the quotient graph
\(Q_\Lambda(\xi,\zeta)\). Lemma~\ref{lem:finite-quotient-coloring}\textup{(b)}
therefore gives that the additional probability of realizing the exterior
spins \(\alpha\) and the interior spins \(\sigma\) contributes the factor
\(
   2^{-N_\Lambda(\xi,\zeta)}
   \mathbf 1_{\{\xi\sim_J\sigma\alpha_{\partial\Lambda}\}}.
\)
Multiplying this factor by
\(
   2^{K_\Lambda(\xi,\zeta)}
   =
   2^{N_\Lambda(\xi,\zeta)-m_\Lambda(\zeta)}
\)
leaves the constant \(2^{-m_\Lambda(\zeta)}\), independent of
\(\xi\) and \(\sigma\), times the compatibility indicator
\(
   \mathbf 1_{\{\xi\sim_J\sigma\alpha_{\partial\Lambda}\}}.
\)

Thus the unnormalized conditional weight of \((\sigma,\xi)\), given the
exterior ES configuration, is proportional to
\[
   p^{|\xi_1|}(1-p)^{|\xi_0|}
   \mathbf 1_{\{\xi\sim_J\sigma\alpha_{\partial\Lambda}\}}.
\]
After normalization over all
\(
   \sigma\in\{-1,+1\}^{\Lambda},
   \hskip 0.2 cm
   \xi\in\{0,1\}^{\overline E_\Lambda},
\)
this is precisely the ES kernel
\(
   \pi_{J,\beta,\Lambda,\alpha_{\partial\Lambda}}
      (\sigma,\xi).
\)
Therefore
\[
   \pi\!\left(
      (\omega_\Lambda,\eta_{\overline E_\Lambda})\in\cdot
      \,\middle|\,
      \mathscr A_{\Lambda^c}^{\mathrm{ES}}
   \right)
   =
   \pi_{J,\beta,\Lambda,\omega_{\partial\Lambda}}(\cdot)
   \qquad
   \pi\text{-a.s.}
\]
Since \(\Lambda\) was arbitrary,
\(
   \pi\in\mathcal G^{\mathrm{ES}}_{J,\beta}.
\)
Together with \(\pi(N_\infty\le1)=1\), this gives
\(
   \pi\in\mathcal G^{\mathrm{ES}}_{\le1},
   \hskip 0.2 cm
   \Pi_{\mathrm{FK}}\pi=\nu.
\)
Hence
\(
   \mathcal G^{\mathrm{FK}}_{\le1}
   \subset
   \Pi_{\mathrm{FK}}
   \bigl(
      \mathcal G^{\mathrm{ES}}_{\le1}
   \bigr).
\)

Combining this inclusion with the inclusion proved in Step~1 gives
\[
   \Pi_{\mathrm{FK}}
   \bigl(
      \mathcal G^{\mathrm{ES}}_{\le1}
   \bigr)
   =
   \mathcal G^{\mathrm{FK}}_{\le1}.
\]
The theorem follows.
\end{proof}

\section{Ferromagnetic FK-clusters and signed
ferromagnetic FK percolation} \label{SFFK1}

\subsection{Separation of an FK cluster into two companions}
\label{SFFK11}

Fix $J$ and a configuration $\eta$ of the
infinite-volume FK variables compatible with $J$.
Let $r(x,y)$ be as in \eqref{SignedFerrConn};
we say that two vertices $x,y \in V$
are signed-ferromagnetically-, or sf-FK-connected, or \(x \xLeftrightarrow{+} y\), if $r(x,y)=1$, 
and we say that they are signed-antiferromagnetically-, or sa-FK-connected, or \(x \xLeftrightarrow{-} y\), if $r(x,y)=-1$.
An  sf-FK-cluster is a set of vertices $C$
such that $r(x,y)=1$ for each $x,y \in C$,
and it is maximal with respect to this property.
The definition is well posed by transitivity 
of the ferromagnetic, i.e. $r(x,y)=1$, connectivity;
the set of vertices $V$ then falls  apart into 
 sf-FK-clusters, some possibly with cardinality one.

\begin{lemma}\label{Lemma1t}
Let $C$ be an sf--FK cluster and fix $x\in C$. Define
\[
   \widetilde C \;:=\;\{y\in V:\ r(x,y)=-1\}.
\]
The set $\widetilde C$ is called the \emph{companion} of $C$; when
nonempty, it is an sf--FK cluster. Moreover, $\widetilde C$ does not depend on the choice
of $x\in C$, and every $x\in C$ and $y\in\widetilde C$ satisfy $r(x,y)=-1$.
\end{lemma}

\begin{proof}
Assume $\tilde C$ is defined based on $x$, and take 
$x' \in C$. Then  $r(x,x')=+1$. If $y$ satisfies $r(x,y)=-1$, concatenate a
ferromagnetic FK-open path from $x'$ to $x$ with an antiferromagnetic FK-open
path from $x$ to $y$; the resulting FK-open path from $x'$ to $y$ has parity
$-1$, hence $r(x',y)=-1$. Thus the set $\widetilde C$ is independent of which
$x\in C$ is used in its definition.

If $\widetilde C=\varnothing$, the assertion is immediate. We therefore
assume that $\widetilde C$ is nonempty.

Now take $y,z\in\widetilde C$. Choose FK-open paths $x\to y$ and $x\to z$ of
parity $-1$. Reversing the first and concatenating yields an FK-open path
$y\to z$ of parity $( -1)\cdot(-1)=+1$, so $r(y,z)=+1$. Hence $\widetilde C$ is
sf--FK-connected. Maximality follows because any vertex sf--FK-connected to a
vertex of $\widetilde C$ is also antiferromagnetically connected to $x$ by the
same concatenation argument, and therefore belongs to $\widetilde C$.
\end{proof}

\begin{lemma} \label{Lemma2}
For each  sf-FK-cluster $C$ and companion $ \tilde C$,
the union $C \cup \tilde C$ is an FK-cluster
in the sense of nearest neighbor
FK open connections in $G$, 
and every standard connected FK cluster $\overline C$
decomposes into two sf-FK companion sets, each of which is an
sf-FK-cluster when nonempty.
\end{lemma}

\begin{proof}
Let \( C \) be an  sf-FK-cluster with companion \( \tilde{C} \). Consider any two vertices \( u, v \in C \cup \tilde{C} \). We show there exists a path connecting \( u \) and \( v \) using FK-open
edges of \(G\) that remains entirely within \( C \cup \tilde{C} \).

\textbf{Case 1:} If \( u, v \in C \), then by definition of  sf-FK-clusters, there exists an FK path \( \gamma \subset E \) connecting \( u \) and \( v \) with \( r(u,v) = +1 \). Any vertex $u'$ in  \( \gamma \) is FK-connected
to $u$, and since 
the parity of an FK-open path from $u$ to $u'$ is either positive or
negative, $u'$ belongs to $C$ or $\tilde C$, respectively. Thus, \( \gamma \) provides a nearest-neighbor FK-open path in \( G \) within \( C \cup \tilde{C} \).

\textbf{Case 2:} If \( u \in C \) and \( v \in \tilde{C} \), by Lemma~\ref{Lemma1t}, there exists an FK-open path \( \gamma \subset E \) connecting \( u \) and \( v \) with \( r(u,v) = -1 \); necessarily, \( \gamma \subset C \cup \tilde C \).

\textbf{Case 3:} If \( u, v \in \tilde{C} \), choose any \( x \in C \). By Lemma~\ref{Lemma1t}, there exist FK-open paths \( \gamma_{ux} \) and \( \gamma_{xv} \) within \( C \cup \tilde{C} \). The concatenated path \( \gamma_{ux} \cup \gamma_{xv} \) connects \( u \) and \( v \) through \( x \), remaining entirely within \( C \cup \tilde{C} \).

In all cases, \( C \cup \tilde{C} \) is connected in the nearest-neighbor sense on \( G \).

If any other point $u$ is FK-connected to either
$C$ or $\tilde C$, then it is connected 
either ferromagnetically, or antiferromagnetically,
so it is ferromagnetically connected to either
$C$ or $\tilde C$.
 Since both $C$ and $\tilde C$ are
maximal, \(u \in C \cup \tilde{C} \); hence 
\( C \cup \tilde{C} \) is an FK-cluster. 

Conversely, if $\overline C$ is an FK cluster
in the nearest neighbor sense then take $v \in \overline C$
and all vertices which are sf-FK-connected to $v$.
This set forms an  sf-FK-cluster $C \subseteq \overline C$;
if $C$ does not exhaust $\overline C$ then vertices in
$\tilde C = \overline C \setminus C$ are necessarily 
sa-FK-connected to $v$, hence they form the companion
cluster of $C$ by Lemma \ref{Lemma1t}.
\end{proof}

\subsection{Signed ferromagnetic (sf) percolation} \label{SfPerc}

We now define the FK percolation criterion for nonzero uniform
magnetization density in Theorem~\ref{EquivalenceOneCopy}.
Recall that we restricted to
 $\mathcal G^{\mathrm{ES}}_{\leq 1}$ 
and its image $\mathcal G^{\mathrm{SPIN}}_{\leq 1}$, so that
we assume that the infinite FK cluster is unique.

Consider a fixed van~Hove sequence $(\Lambda_n)_{n\ge1}$.
This sequence is kept fixed in each fixed-\(J\) argument:
in fact, all 
we will be interested in is that there
is a way of distinguishing the two
companions, and not the specific way
of identifying them. As mentioned, we
use a fixed van-Hove sequence and 
density for definiteness;
we also briefly comment below on
the general features that the distinguishing
operator needs to have for the theory
to work.

We are now interested in the fact that there
is a way of distinguishing the density of
the two
companions at infinity.

For each
ordinary FK cluster $C\subset V$, let $C^A$
and $C^B$ indicate its two companions; for a fixed deterministic van Hove sequence
\(\Gamma=(\Lambda_n)_{n\ge1}\),
write
\begin{equation}
\label{density}
   d(C^A)\;=\; d_{\Gamma}(C^A)
   \;:=\;
   \lim_{n\to\infty}\frac{|C^A \cap\Lambda_n|}{|\Lambda_n |},
\end{equation} 
whenever the limit exists, and analogously for $C^B$.
Given $\eta\in\mathcal H$ with $N_\infty(\eta)\le1$, let
$C_\infty^{A}(\eta)$ and $C_\infty^{B}(\eta)$ be the two sf--FK companions
of the (unique, if it exists) infinite FK cluster.

\begin{definition}
\label{def:percolation-selector}
Fix an FK distribution $\nu$
satisfying uniqueness of the infinite FK cluster. 
For a fixed deterministic van Hove sequence
\(\Gamma=(\Lambda_n)_{n\ge1}\), we say that sf--FK percolation occurs
along \(\Gamma\) if the two signed ferromagnetic companions of the
infinite FK cluster have
density limits $d_{\Gamma}(C^A), d_{\Gamma}(C^B)$ along \(\Gamma\), and these limits are unequal with  positive $\nu$
probability. And we say that sf--FK percolation occurs
if it occurs 
along some \(\Gamma\).
\end{definition}

An example in which  one of the two FK companions has 
higher density than the other is of course the low temperature
Ising ferromagnet at $\lambda=1$ in which 
there is only one sf-FK
companion and the other is empty.

When  sf-FK-percolation occurs, 
the density distinguishes the 
two FK companions, and we now show that
this selection is stable outside finite sets under local bond modifications.
For a fixed van Hove sequence $\Gamma$, define the set–valued density based map $\mathcal S_D=\mathcal S_{D,\Gamma}:\mathcal H\to\mathcal P(V)$ by
\begin{equation}
    \label{S}
   \mathcal S_{D,\Gamma}(\eta)
   \;:=\;
   \begin{cases}
     C_\infty^{A}(\eta), & \text{if }d\bigl(C_\infty^{A}(\eta)\bigr)
                            > d\bigl(C_\infty^{B}(\eta)\bigr),\\[2pt]
     C_\infty^{B}(\eta), & \text{if }d\bigl(C_\infty^{B}(\eta)\bigr)
                            > d\bigl(C_\infty^{A}(\eta)\bigr),\\[2pt]
     \varnothing,        & \text{if limits do not exist or are equal.}
   \end{cases}
\end{equation} 
Notice that if $N_\infty(\eta)=0$ then $\mathcal S_{D}(\eta)=\varnothing$.
Then
\begin{lemma}[Stability of the density selector under finite modifications]
\label{lem:tail-pql-density-selector1}
Let \(\Gamma=(\Lambda_n)_{n\ge1}\) be the fixed van~Hove sequence, and let
\(\mathcal S_{D,\Gamma}\) be the density selector defined from the
\(\Gamma\)-density limits of the two signed companions, with the convention
that \(\mathcal S_{D,\Gamma}(\eta)=\varnothing\) if there is no unique
infinite FK cluster, or if the two companion density limits along \(\Gamma\)
do not both exist, or if they are equal.

Consider either an ES measure
\(
   \pi\in\mathcal G^{\mathrm{ES}}_{\le1}
\)
and let \(\bar\nu=\Pi_{\mathrm{FK}}\pi\), or an FK measure
\(
   \nu\in\mathcal G^{\mathrm{FK}}_{\le1}
\)
and let \(\bar\nu=\nu\). Then, for every finite
\(\Lambda\Subset V\), for \(\bar\nu\)-a.e. \(\eta\in\mathcal H\), there exists
a finite set \(M=M(\eta,\Lambda)\Subset V\) such that, for every FK
configuration \(\xi\) compatible with \(J\) and satisfying
\(
   \xi_{E\setminus\overline E_\Lambda}
   =
   \eta_{E\setminus\overline E_\Lambda},
\)
one has
\(
   \mathcal S_{D,\Gamma}(\xi)\cap M^c
   =
   \mathcal S_{D,\Gamma}(\eta)\cap M^c .
\)
\end{lemma}

\begin{proof}
It is enough to work under the FK law \(\bar\nu\). Fix
\(\Lambda\Subset V\). Work on the full \(\bar\nu\)-measure event on which
there is at most one infinite FK cluster and the conclusion of
Lemma~\ref{UniqueOutside} holds with the deleted edge set
\(F=\overline E_\Lambda\).

For such an \(\eta\), consider the exterior configuration
\(
   \eta_{E\setminus\overline E_\Lambda}.
\)
Among the exterior FK clusters which meet the vertices incident to
\(\overline E_\Lambda\), all but possibly one are finite; if an infinite one
exists, it is the unique infinite exterior component. Hence we can choose a
finite set
\(
   M=M(\eta,\Lambda)\Subset V,\qquad \Lambda\subset M,
\)
containing all the finite exterior FK clusters that meet the vertices incident
to \(\overline E_\Lambda\).

Let \(\xi\) be compatible with \(J\) and agree with \(\eta\) on
\(E\setminus\overline E_\Lambda\). Since \(\xi\) differs from \(\eta\) only on
the finite set \(\overline E_\Lambda\), the unique infinite FK cluster, if it
exists, has the same exterior part outside \(M\) for \(\eta\) and for \(\xi\).
If no infinite exterior component exists, then no finite modification inside
\(\overline E_\Lambda\) can create an infinite FK cluster outside \(M\), and
both selectors are empty outside \(M\).

Assume therefore that an infinite exterior FK component \(I\) exists. Since
\(\eta\) and \(\xi\) coincide on all exterior edges, the signed parity relation
between any two vertices of \(I\cap M^c\), computed using exterior FK-open
paths, is the same for \(\eta\) and for \(\xi\). Thus the unordered pair of
signed companions of the infinite FK cluster agrees outside \(M\) for
\(\eta\) and \(\xi\). More explicitly, up to exchanging the two labels, the two
sets
\(
   C_\infty^A(\eta)\cap M^c,
   C_\infty^B(\eta)\cap M^c
\)
coincide with
\(
   C_\infty^A(\xi)\cap M^c,
   C_\infty^B(\xi)\cap M^c .
\)

Changing finitely many vertices does not affect the existence or the values
of the density limits along the fixed van~Hove sequence \(\Gamma\). Therefore
\[
   \mathcal S_{D,\Gamma}(\eta)=\varnothing
   \quad\Longleftrightarrow\quad
   \mathcal S_{D,\Gamma}(\xi)=\varnothing .
\]
When the selector is nonempty, the same member of the unordered companion
pair is selected outside \(M\), namely the one with the larger
\(\Gamma\)-density. Hence
\[
   \mathcal S_{D,\Gamma}(\eta)\cap M^c
   =
   \mathcal S_{D,\Gamma}(\xi)\cap M^c .
\]
This proves the claim.
\end{proof}

The property proven in this lemma is actually
the only property that is used later in the paper, so it is possible to envision
other signed percolation phenomena which rely
on other tail measurable functions distinguishing
the two companions. 
One example could be identifying the companion
with the higher density of left-right edges;
such identification would be
 compatible with the 
two companions having the same density and 
hence not being distinguishable in the sense used above.
Other tail observables can distinguish companions whose vertex densities
are equal. Their relation to thermodynamic order parameters requires
separate identification.

For a translation-invariant joint law of the couplings and FK bonds,
the mean density imbalance along a fixed van~Hove sequence can be
identified with the difference between the probabilities that the origin
belongs to the denser companion and to the other companion. The precise
statement is Theorem~\ref{Th1}. This identification does not assert
translation invariance of a fixed-$J$ state.

\section{Uniform magnetization and signed FK percolation
in \texorpdfstring{\(\mathbb{Z}^d\)}{the integer lattice}} \label{SFFK2}

\subsection{Decomposition of ES states according to the 
            sf--FK companion bias}
\label{subsec:decomposition}

Throughout this subsection we fix a coupling realization
$J\in\{\pm1\}^{E}$ and an inverse temperature $\beta$
(or the edge–occupation parameter $p=1-e^{-2\beta}$).
We also restrict to $\mathcal G^{\mathrm{ES}}_{J,\beta,\leq 1}$; by
Lemma \ref{Lemma1t} the infinite FK cluster splits  (when it exists) into
\emph{two} signed companion sets $C_{\infty}^{A}$ and
$C_{\infty}^{B}$ antiferromagnetically connected between each other.
Write
\(
   \mathcal C_{\infty}(\eta)=C_{\infty}^{A}\cup C_{\infty}^{B}.
\)
and let $E(\mathcal C_{\infty}(\eta))$ be the 
set of edges incident to at least one vertex of 
$\mathcal C_{\infty}(\eta)$.

While the coherence of the finite-volume conditional specifications
between spin and ES distributions implies that the ES distribution
must be invariant when the spins of 
each finite FK cluster are flipped (which originates the well known factor $2^{Cl(\eta)}$, where
$Cl(\eta) =$ number of (finite) clusters, in the FK measure), there is no
such constraint for the infinite cluster. 
The compatible assignments on the infinite cluster can have different
probabilities. We study these probabilities when a companion is selected
by its density.
To this end, fix a van Hove sequence $\Gamma$, consider a given bond configuration 
$\eta_{E^d}$ and label the two companions of the infinite FK cluster
as follows: if the two companions have distinct densities along
$\Gamma$, let
$C_\infty^{A}(\eta)=\mathcal S_{D,\Gamma}(\eta)$, and otherwise label the two companions
as $C_\infty^{A}(\eta)$ and 
$C_{\infty}^{B}$ using a fixed Borel labeling. For an ES state $\pi$, let
$$
q\;:=\pi\bigl(\omega_x=+1 \text{ for all } x\in C_{\infty}^{A}
\mid  (\omega_{\mathbb Z^d \setminus \mathcal C_{\infty} },
\eta_{E^d  })\bigr)
 $$
and denote by $Q(\pi)\subset[0,1]$ the essential range of this
conditional probability. Thus $Q(\pi)=\{q_0\}$ means that the
conditional probability equals $q_0$ almost surely.

\medskip
For fixed couplings $J$, let \[
\begin{aligned}
&\mathcal G^{\mathrm{ES}}_{null}
      :=\Bigl\{
            \pi\in\mathcal G^{\mathrm{ES}}_{J,\beta}:
            \pi(N_{\infty}=0)=1
        \Bigr\},\\[4pt]
&\mathcal G^{\mathrm{ES}}_{1,\mathrm{unsel}}
      :=\Bigl\{
            \pi\in\mathcal G^{\mathrm{ES}}_{J,\beta}:
            \pi(N_{\infty}=1)=1,\; \text{ and  } 
            \mathcal S_{D,\Gamma}(\eta) \equiv \emptyset\\[4pt]
    & \hskip 1cm
         \text{   with $\pi$ probability one }
           \Bigr\}\\[4pt] \end{aligned}
\]
The subscript \(\mathrm{unsel}\) means that the density selector is
empty. If both companion densities exist, they are equal in this class.
Without a density-existence hypothesis, the class also includes failure
of the density limits. This condition imposes no spin-flip invariance
on the ES measure.
On ES configurations, the global spin flip acts by
\[
   F_{\mathrm{ES}}(\omega,\eta):=(-\omega,\eta).
\]
Define the spin-flip-invariant subclass by
\[
   \mathcal G^{\mathrm{ES},\mathrm{flip}}_{1,\mathrm{unsel}}
   :=
   \left\{
      \pi\in\mathcal G^{\mathrm{ES}}_{1,\mathrm{unsel}}:
      (F_{\mathrm{ES}})_*\pi=\pi
   \right\}.
\]

For the density-selected configurations, set
\[
   E_{\mathrm{sel}}
   :=
   \{N_\infty=1,\ \mathcal S_{D,\Gamma}(\eta)\ne\varnothing\}.
\]
We next construct ES distributions over FK states supported on this event.
Notice that sf--FK percolation only requires
\(
   \nu(E_{\mathrm{sel}})>0.
\)
When this probability is positive, we apply the following lemma to the
conditional FK state
\(
   \nu_{\mathrm{sel}}:=\nu(\,\cdot\,|E_{\mathrm{sel}}).
\)
Since \(E_{\mathrm{sel}}\) is a tail event, \(\nu_{\mathrm{sel}}\) is again
an FK Gibbs state for the same specification. Under
\(\nu_{\mathrm{sel}}\), the assumptions \(N_\infty=1\) and
\(\mathcal S_D\neq\varnothing\) hold with probability one.

\begin{lemma}[Gibbs property of the $q$--biased ES lifts]
\label{lem:DLR_for_q}
Fix a coupling realization $J\in\{\pm1\}^{E}$ and
$0\le\beta<\infty$, with $p=1-e^{-2\beta}<1$.
Let
\(
   \nu\in\mathcal G^{\mathrm{FK}}_{J,p}
\)
be an FK Gibbs state such that
\(
   \nu(N_\infty=1)=1
\)
and such that the two signed ferromagnetic companions of the unique
infinite FK cluster are distinguished by the selector $\mathcal S_{D,\Gamma}$
with $\nu$-probability one.  For $q\in[0,1]$, define an ES measure
$\pi_{\nu,q}$ as follows.

First sample $\eta$ according to $\nu$.  On every finite FK cluster,
choose one of the two compatible spin assignments with probability
$1/2$, independently from cluster to cluster.  On the unique infinite
FK cluster, write its two signed ferromagnetic companions as
\[
   C_\infty^A(\eta)=\mathcal S_{D,\Gamma}(\eta),
   \qquad
   C_\infty^B(\eta)=\mathcal C_\infty(\eta)\setminus C_\infty^A(\eta).
\]
Assign spin $+1$ to all vertices of $C_\infty^A(\eta)$, and therefore
spin $-1$ to all vertices of $C_\infty^B(\eta)$, with probability $q$;
assign the opposite spins with probability $1-q$.

Then
\(
   \pi_{\nu,q}\in \mathcal G^{\mathrm{ES}}_{J,\beta}
\)
for every $q\in[0,1]$.  In particular, each $q$--biased lift obeys the
same Edwards--Sokal DLR specification
\(
   \bigl\{
      \pi_{J,\beta,\Lambda,\omega_{\partial\Lambda}}
   \bigr\}_{\Lambda\Subset V}.
\)
\end{lemma}

\begin{proof}
We first construct a spin-flip-invariant state from the given FK measure
$\nu$.

Let $\pi_{\nu,1/2}$ be the ES lift of $\nu$ obtained by assigning
independent fair signs to all FK clusters, including the unique infinite
FK cluster.  By Theorem~\ref{thm:FK-surjective}, or more precisely by the
constructive part of its proof, this lift satisfies the ES DLR equations:
\[
   \pi_{\nu,1/2}\in \mathcal G^{\mathrm{ES}}_{J,\beta},
   \qquad
   \Pi_{\mathrm{FK}}\pi_{\nu,1/2}=\nu .
\]

Next, define
\[
      H(\omega,\eta)
      :=
      \mathbf 1_{\{
          \omega_x=+1 \text{ for all } x\in C_\infty^A(\eta)
      \}},
\]
where \(C_\infty^A(\eta)=\mathcal S_{D,\Gamma}(\eta)\).  Under
$\pi_{\nu,1/2}$, conditional on $\eta$, the spin on the selected infinite
companion is $+1$ with probability $1/2$, and therefore
\[
   \int H\,d\pi_{\nu,1/2}=\frac12.
\]
For $q\in[0,1]$, set
\[
   h_q(\omega,\eta)
   :=
   2q\,H(\omega,\eta)
   +
   2(1-q)\bigl(1-H(\omega,\eta)\bigr);
\]
we have
\(
   \int h_q\,d\pi_{\nu,1/2}=1,
\)

 Define the $q$--biased lift  by
\[
   \pi_{\nu,q}:=h_q\,\pi_{\nu,1/2},
\] or, 
equivalently,
\(
   \frac{d\pi_{\nu,q}}{d\pi_{\nu,1/2}}
   =
   h_q.
\)

We now show that $h_q$ is exterior-measurable for every finite volume.
Fix $\Lambda\Subset V$ and let
\[
   \mathscr A_{\Lambda^c}
   :=
   \sigma(\omega_x:\ x\notin\Lambda)
   \vee
   \sigma(\eta_b:\ b\in E\setminus\overline E_\Lambda).
\]
We claim that $H$, and hence $h_q$, is
$\mathscr A_{\Lambda^c}$-measurable, modulo
$\pi_{\nu,1/2}$-null sets.

Indeed, by Lemma~\ref{lem:tail-pql-density-selector1}, applied with the
finite edge set $\overline E_\Lambda$, for $\nu$-almost every $\eta$
there exists a finite set $M=M(\eta,\Lambda)$ containing $\Lambda$ such
that any configuration $\xi$ compatible with $J$ and agreeing with
$\eta$ outside $\overline E_\Lambda$ satisfies
\[
   \mathcal S_{D,\Gamma}(\eta)\cap M^c
   =
   \mathcal S_{D,\Gamma}(\xi)\cap M^c .
\]
Since $\mathcal S_{D,\Gamma}(\eta)=C_\infty^A(\eta)$ is infinite, choose
\(
   z\in C_\infty^A(\eta)\cap M^c.
\)
If $(\omega',\xi)$ agrees with $(\omega,\eta)$ on
$\mathscr A_{\Lambda^c}$, then in particular $\omega'_z=\omega_z$ and
\(
   z\in C_\infty^A(\xi).
\)
Because all vertices in one signed ferromagnetic companion are connected
by FK-open paths of positive parity, ES compatibility forces their spins
to be equal.  Hence
\[
   H(\omega,\eta)=1
   \quad\Longleftrightarrow\quad
   \omega_z=+1
   \quad\Longleftrightarrow\quad
   \omega'_z=+1
   \quad\Longleftrightarrow\quad
   H(\omega',\xi)=1.
\]
Therefore
\[
   h_q(\omega,\eta)=h_q(\omega',\xi)
\]
whenever $(\omega,\eta)$ and $(\omega',\xi)$ agree on
$\mathscr A_{\Lambda^c}$.  This proves the claimed exterior
measurability.

Now let $f$ be any bounded local function depending only on
\((\omega_\Lambda,\eta_{\overline E_\Lambda})\), and let
$g$ be any bounded $\mathscr A_{\Lambda^c}$-measurable function.
Since $\pi_{\nu,1/2}$ satisfies the ES DLR equations,
\[
   \int f g\,d\pi_{\nu,1/2}
   =
   \int g\,
      \pi_{J,\beta,\Lambda,\omega_{\partial\Lambda}}(f)
      \,d\pi_{\nu,1/2}.
\]
Using the exterior-measurability of $h_q$, we get
\[
\begin{aligned}
   \int f g\,d\pi_{\nu,q}
   &=
   \int f g h_q\,d\pi_{\nu,1/2}  \\
   &=
   \int g h_q\,
      \pi_{J,\beta,\Lambda,\omega_{\partial\Lambda}}(f)
      \,d\pi_{\nu,1/2} \\
   &=
   \int g\,
      \pi_{J,\beta,\Lambda,\omega_{\partial\Lambda}}(f)
      \,d\pi_{\nu,q}.
\end{aligned}
\]
Since this holds for every bounded
$\mathscr A_{\Lambda^c}$-measurable $g$, we have
\[
   \mathbb E_{\pi_{\nu,q}}
   \bigl[
      f
      \mid
      \mathscr A_{\Lambda^c}
   \bigr]
   =
   \pi_{J,\beta,\Lambda,\omega_{\partial\Lambda}}(f)
   \qquad
   \pi_{\nu,q}\text{-a.s.}
\]
This is exactly the ES DLR equation in the volume $\Lambda$.
Since $\Lambda$ was arbitrary,
\[
   \pi_{\nu,q}\in\mathcal G^{\mathrm{ES}}_{J,\beta}.
\]
\end{proof}

Finally, we let 
\[
\begin{aligned}
&\mathcal G^{\mathrm{ES}}_{1,q}
      :=\Bigl\{
            \pi\in\mathcal G^{\mathrm{ES}}_{J,\beta}:
            \pi(N_{\infty}=1)=1,\;
            \mathcal S_{D,\Gamma}(\eta)=C_\infty^{A}(\eta)\ne\varnothing \\[4pt]
    & \hskip 1cm
         \text{  with $\pi$  probability one,
             and }
            Q(\pi)=\{q\}
        \Bigr\},
        \; \text{ for } q\in [0,1]
\end{aligned}
\]

Next, let

\[
\begin{aligned}
 &  \mathcal G^{\mathrm{FK}}_{null}
   \; :=\;
   \bigl\{
        \nu\in\mathcal G^{\mathrm{FK}}_{J,p}:
        \;\nu\bigl(N_{\infty}=0\bigr)=1
   \bigr\},  \\[4pt]
  & \mathcal G^{\mathrm{FK}}_{1,\mathrm{unsel}}
   \;:=\;
   \bigl\{
        \nu\in\mathcal G^{\mathrm{FK}}_{J,p}:
        \;\nu\bigl(N_{\infty}=1 \bigr)=1
    \text{ and  } 
            \mathcal S_{D,\Gamma}(\eta) \equiv \emptyset\\[4pt]
    & \hskip 1cm
         \text{   with $\nu$ probability one }
           \bigr\} \\[4pt]
 &  \mathcal G^{\mathrm{FK}}_{1,\mathrm{sel}}
   \;:=\;\bigl\{
        \nu\in\mathcal G^{\mathrm{FK}}_{J,p}:
        \;\nu\bigl(N_{\infty}=1 \bigr)=1, \;
   \mathcal S_{D,\Gamma}(\eta) =C_\infty^{A}(\eta)\ne\varnothing\\[4pt]
    & \hskip 1cm
         \text{  with $\nu$  probability one}
           \bigr\} 
   \\[4pt]
\end{aligned}
\]

Following and extending {\cite[Thm.\,II.2(ii)–(iv)]{biskup2000gibbs}},
we decompose an ES measure according to the density selector and,
when a companion is selected, its spin orientation. These are tail
conditions, so conditioning on them preserves the ES specification.

\begin{theorem}[Canonical decomposition by the signed infinite companion]
\label{thm:canonical_sf_decomposition}
Let
\[
   \pi\in\mathcal G^{\mathrm{ES}}_{J,\beta}
   \qquad\text{with}\qquad
   \pi(N_\infty\le 1)=1 .
\]
Let
\[
   E_{\mathrm{null}}:=\{N_\infty=0\},
\]
\[
   E_{\mathrm{unsel}}
   :=
   \{N_\infty=1,\ \mathcal S_{D,\Gamma}(\eta)=\varnothing\},
\]
and
\[
   E_{\mathrm{sel}}
   :=
   \{N_\infty=1,\ \mathcal S_{D,\Gamma}(\eta)\neq\varnothing\}.
\]
On \(E_{\mathrm{sel}}\), write
\[
   C_\infty^A(\eta):=\mathcal S_{D,\Gamma}(\eta)
\]
for the selected signed ferromagnetic companion, and define
\[
   H(\omega,\eta)
   :=
   \mathbf 1_{\{
      \omega_x=+1 \text{ for all } x\in C_\infty^A(\eta)
   \}} .
\]
Set
\[
   E_+ := E_{\mathrm{sel}}\cap\{H=1\},
   \qquad
   E_- := E_{\mathrm{sel}}\cap\{H=0\}.
\]
Then
\[
   E_{\mathrm{null}},\quad E_{\mathrm{unsel}},\quad E_+,\quad E_-
\]
are tail events modulo \(\pi\)-null sets and form a measurable partition of
\(\{N_\infty\le1\}\). Consequently \(\pi\) admits the canonical
decomposition
\[
   \pi
   =
   \alpha_{\mathrm{null}}\pi_{\mathrm{null}}
   +
   \alpha_{\mathrm{unsel}}\pi_{\mathrm{unsel}}
   +
   \alpha_+\pi_+
   +
   \alpha_-\pi_-,
\]
where
\[
   \alpha_{\mathrm{null}}=\pi(E_{\mathrm{null}}),\qquad
   \alpha_{\mathrm{unsel}}=\pi(E_{\mathrm{unsel}}),\qquad
   \alpha_+=\pi(E_+),\qquad
   \alpha_-=\pi(E_-),
\]
and, whenever the corresponding coefficient is nonzero,
\[
   \pi_{\mathrm{null}}=\pi(\,\cdot\,|E_{\mathrm{null}}),\qquad
   \pi_{\mathrm{unsel}}=\pi(\,\cdot\,|E_{\mathrm{unsel}}),\qquad
   \pi_+=\pi(\,\cdot\,|E_+),\qquad
   \pi_-=\pi(\,\cdot\,|E_-).
\]
Each nonzero component is again an Edwards--Sokal Gibbs state for the same
specification. This decomposition is unique, with the convention that
zero-weight components are omitted.
\end{theorem}

\begin{proof}
Work on the common full-measure set of Lemma~\ref{UniqueOutside}.
After any finite deletion there is at most one infinite exterior
component. Restoring finitely many bonds cannot create an infinite
cluster from finite components or remove an existing infinite component.
Thus the presence or absence of the unique infinite FK cluster is
unchanged by finite bond modifications on this set. By
Lemma~\ref{lem:tail-pql-density-selector1}, the density selector is stable
outside a finite set, so its nonemptiness is unchanged by compatible
modifications. Hence \(E_{\mathrm{null}}\), \(E_{\mathrm{unsel}}\), and
\(E_{\mathrm{sel}}\) are tail events, modulo \(\pi\)-null sets.

We next check that \(H\) is also tail-measurable on \(E_{\mathrm{sel}}\).
Fix a finite \(\Lambda\Subset V\). By Lemma~\ref{lem:tail-pql-density-selector1},
for \(\pi\)-a.e. \(\eta\) there is a finite set
\(M=M(\eta,\Lambda)\) such that every configuration \(\xi\) compatible
with \(J\) and agreeing with \(\eta\) outside
\(\overline E_\Lambda\) satisfies
\[
   \mathcal S_{D,\Gamma}(\xi)\cap M^c
   =
   \mathcal S_{D,\Gamma}(\eta)\cap M^c .
\]
Since \(C_\infty^A(\eta)=\mathcal S_{D,\Gamma}(\eta)\) is infinite on
\(E_{\mathrm{sel}}\), choose
\(
   z\in C_\infty^A(\eta)\cap M^c .
\)
For every ES-compatible configuration, all spins on a signed
ferromagnetic companion are equal, because any two vertices in the same
companion are connected by an FK-open path of positive parity. Therefore
\[
   H(\omega,\eta)=1
   \quad\Longleftrightarrow\quad
   \omega_z=+1.
\]
The vertex \(z\) lies outside the finite set \(M\), and hence outside
\(\Lambda\). Thus the value of \(H\) is determined by variables outside
\(\Lambda\), up to \(\pi\)-null sets. Since \(\Lambda\) was arbitrary,
\(H\) is tail-measurable. Therefore \(E_+\) and \(E_-\) are tail events.

The four events
\[
   E_{\mathrm{null}},\quad E_{\mathrm{unsel}},\quad E_+,\quad E_-
\]
are disjoint and exhaust \(\{N_\infty\le1\}\). Since
\(\pi(N_\infty\le1)=1\), they form a measurable partition of the support of
\(\pi\). Hence the decomposition
\[
   \pi
   =
   \sum_{E\in\{E_{\mathrm{null}},E_{\mathrm{unsel}},E_+,E_-\}}
   \pi(E)\,\pi(\,\cdot\,|E)
\]
is simply the ordinary decomposition of a probability measure with respect
to a measurable partition. This proves existence and uniqueness of the
displayed decomposition.

It remains only to show that the nonzero conditional measures are still ES
Gibbs states. Let \(E\) be one of the four tail events, with \(\pi(E)>0\), and
write
\(
   \pi_E:=\pi(\,\cdot\,|E).
\)
Fix a finite \(\Lambda\Subset V\). Let
\(
   \mathscr A_{\Lambda^c}
\)
be the sigma-field generated by all spin and bond variables outside
\(\Lambda\) in the ES specification. Since \(E\) is a tail event,
\(\mathbf 1_E\) is \(\mathscr A_{\Lambda^c}\)-measurable, up to null sets.

Let \(f\) be a bounded local function depending only on the variables in
\(\Lambda\), and let \(g\) be a bounded
\(\mathscr A_{\Lambda^c}\)-measurable function. Since \(\pi\) satisfies the ES
DLR equations,
\[
   \int f g\,d\pi
   =
   \int g\,
      \pi_{J,\beta,\Lambda,\omega_{\partial\Lambda}}(f)
      \,d\pi .
\]
Multiplying by \(\mathbf 1_E\), which is exterior-measurable, gives
\[
   \int_E f g\,d\pi
   =
   \int_E g\,
      \pi_{J,\beta,\Lambda,\omega_{\partial\Lambda}}(f)
      \,d\pi .
\]
Dividing by \(\pi(E)\), we obtain
\[
   \int f g\,d\pi_E
   =
   \int g\,
      \pi_{J,\beta,\Lambda,\omega_{\partial\Lambda}}(f)
      \,d\pi_E .
\]
Since this holds for every bounded exterior-measurable \(g\), it says that
the conditional law of the variables in \(\Lambda\) under \(\pi_E\), given the
exterior, is the same ES kernel
\(
   \pi_{J,\beta,\Lambda,\omega_{\partial\Lambda}}.
\)
Thus \(\pi_E\in\mathcal G^{\mathrm{ES}}_{J,\beta}\). Applying this to each of
the four events proves that every nonzero component in the decomposition is
an ES Gibbs state.
\end{proof}

\begin{remark}
The decomposition above is a canonical decomposition by tail events, not
the Choquet decomposition into extremal ES states.  In particular, the
components \(\pi_+\) and \(\pi_-\) need not be extremal, because their FK
marginals may themselves be non-extremal FK Gibbs states.  The Choquet
decomposition would further decompose these components into extremal ES
states.

\end{remark}

\begin{corollary}
\label{cor:extremal_sf_sector}
Let \(\pi\in\mathcal G^{\mathrm{ES}}_{J,\beta}\) be extremal and suppose
\(\pi(N_\infty\le1)=1\). Then exactly one of the four events
\[
   E_{\mathrm{null}},\quad E_{\mathrm{unsel}},\quad E_+,\quad E_-
\]
has \(\pi\)-probability one.
In particular, if \(\pi(E_{\mathrm{sel}})=1\), then either
\[
   \pi(H=1)=1
   \qquad\text{or}\qquad
   \pi(H=0)=1.
\]
Equivalently, in the density-distinguishable tail event an extremal ES state
cannot have a genuine \(q\)-bias with \(0<q<1\).
\end{corollary}

\begin{proof}
The four events are tail events by Theorem~\ref{thm:canonical_sf_decomposition}.
Extremal ES Gibbs states are tail-trivial. Hence each of the four events has
probability \(0\) or \(1\). Since they form a partition of a full-probability
set, exactly one of them has probability \(1\).
\end{proof}

On the selected event, extremality fixes the spin orientation of the
denser companion. The \(q\)-biased lifts describe mixtures of the two
oriented assignments. The component supported on
\(E_{\mathrm{unsel}}\) is specified only by the density selector and
may also fail to be spin-flip invariant.

\begin{proposition}[Restricted FK projection in the single-copy case]
\label{prop:restricted_FK_projection_single}
Assume the \(0/1\)-infinite-cluster property. Let
\[
   \mathcal G^{\mathrm{FK}}_{null},\qquad
   \mathcal G^{\mathrm{FK}}_{1,\mathrm{unsel}},\qquad
   \mathcal G^{\mathrm{FK}}_{1,\mathrm{sel}}
\]
be the classes of FK states supported, respectively, on the following tail events: no infinite FK cluster; one
infinite FK cluster whose signed companions are not selected by
\(\mathcal S_{D,\Gamma}\); and one infinite FK cluster whose signed companions are
selected by \(\mathcal S_{D,\Gamma}\).

Then the following hold.

\begin{enumerate}
\item[\textup{(i)}]
The restriction
\[
   \Pi_{\mathrm{FK}}:
   \mathcal G^{\mathrm{ES}}_{null}
   \longrightarrow
   \mathcal G^{\mathrm{FK}}_{null}
\]
is a bijection.

\item[\textup{(ii)}]
The restriction to the spin-flip-invariant subclass,
\[
   \Pi_{\mathrm{FK}}:
   \mathcal G^{\mathrm{ES},\mathrm{flip}}_{1,\mathrm{unsel}}
   \longrightarrow
   \mathcal G^{\mathrm{FK}}_{1,\mathrm{unsel}}
\]
is a bijection.

\item[\textup{(iii)}]
For each fixed \(q\in[0,1]\), the restriction
\[
   \Pi_{\mathrm{FK}}:
   \mathcal G^{\mathrm{ES}}_{1,q}
   \longrightarrow
   \mathcal G^{\mathrm{FK}}_{1,\mathrm{sel}}
\]
is a bijection.
Equivalently, once the FK marginal and the bias \(q\) are fixed, the ES
lift is uniquely determined.
\end{enumerate}
\end{proposition}

\begin{proof}
Surjectivity in all three cases follows from the ES lifting construction for
FK Gibbs states satisfying the \(0/1\)-infinite-cluster property. Namely, every
\(\nu\in\mathcal G^{\mathrm{FK}}_{\le1}\) has at least one ES lift
\(\pi\in\mathcal G^{\mathrm{ES}}_{\le1}\) with
\(\Pi_{\mathrm{FK}}\pi=\nu\). This is precisely the content of the FK
surjectivity theorem. For the prescribed bias in case~\textup{(iii)},
Lemma~\ref{lem:DLR_for_q} supplies the required $q$-biased ES lift.

We prove injectivity tail event by tail event.

\smallskip
\noindent
\textup{(i)} If \(\nu\in\mathcal G^{\mathrm{FK}}_{null}\), then every FK
cluster is finite. Given the FK configuration \(\eta\), the ES conditional
spin law is forced: each finite FK cluster receives one of its two compatible
spin assignments with probability \(1/2\), independently over clusters.
Thus the conditional spin law given \(\eta\) is uniquely determined by
\(\eta\), and hence the FK marginal determines the ES measure uniquely.

\smallskip
\noindent
\textup{(ii)} Let \(\nu\in\mathcal G^{\mathrm{FK}}_{1,\mathrm{unsel}}\).
Given the FK configuration, finite-cluster colors are independent and
fair by the ES specification. For a spin-flip-invariant ES lift,
\(F_{\mathrm{ES}}\) leaves that configuration fixed and exchanges the
two compatible assignments on the unique infinite cluster. Those
assignments therefore have equal conditional probabilities. Thus the
conditional spin law is the fair-cluster coloring law and is determined
by \(\eta\). This proves injectivity on
\(\mathcal G^{\mathrm{ES},\mathrm{flip}}_{1,\mathrm{unsel}}\).
Conversely, the fair-cluster construction gives a spin-flip-invariant
ES lift for every \(\nu\) in the indicated FK class.

\smallskip
\noindent
\textup{(iii)} Let \(q\in[0,1]\) be fixed and let
\(\nu\in\mathcal G^{\mathrm{FK}}_{1,\mathrm{sel}}\). For \(\nu\)-a.e. \(\eta\),
the selector \(\mathcal S_{D,\Gamma}(\eta)\) identifies one of the two signed
companions of the infinite FK cluster; call it \(C_\infty^A(\eta)\), and call the other
\(C_\infty^B(\eta)\). Given \(\eta\), the ES conditional spin law in the
class \(\mathcal G^{\mathrm{ES}}_{1,q}\) is uniquely prescribed as follows:
all finite FK clusters are colored by independent fair signs, while
\(C_\infty^A(\eta)\) receives spin \(+1\) with probability \(q\) and spin
\(-1\) with probability \(1-q\). The spin on \(C_\infty^B(\eta)\) is then
forced to be opposite, because the two companions are antiferromagnetically
connected.

Thus, for fixed \(q\), the conditional spin law given \(\eta\) is uniquely
determined by the FK configuration. Therefore two ES measures in
\(\mathcal G^{\mathrm{ES}}_{1,q}\) with the same FK marginal must coincide.
This proves injectivity. Together with surjectivity, this proves the
bijection.
\end{proof}

\begin{remark}
The restriction to fixed \(q\) in part \(\textup{(iii)}\) is essential.
Without fixing \(q\), the FK projection is not injective. Indeed, for the same
FK marginal \(\nu\in\mathcal G^{\mathrm{FK}}_{1,\mathrm{sel}}\), the ES lifts
obtained by assigning spin \(+1\) to the selected companion with probabilities
\(q\) and \(q'\) have the same FK marginal but are different whenever
\(q\neq q'\).

Within the classes specified in the proposition, the FK marginal and the
prescribed coloring law determine the ES lift uniquely.
\end{remark}

We next characterize nonzero limiting expected uniform magnetization
density in an extremal state. The result combines the tail decomposition,
Corollary~\ref{cor:extremal_sf_sector}, and the spin--ES isomorphism in
Theorem~\ref{SpinProjSincleCopy}, with existence of the companion densities
along the chosen sequence.

\begin{theorem}\label{EquivalenceOneCopy}
In \(\mathbb Z^d\), fix a finite inverse temperature \(\beta\ge0\),
\(p=1-e^{-2\beta}\), and a coupling realization \(J\).
Let \(\mu_{J,\beta}\) be an extremal Gibbs state, let
\(\pi\in\mathcal G^{\mathrm{ES}}_{J,\beta}\) be its Edwards--Sokal lift,
and let \(\nu_{J,p}:=\Pi_{\mathrm{FK}}\pi\). Assume that
\(\pi(N_\infty\le1)=1\). Fix a deterministic van~Hove sequence
\(\Gamma=(\Lambda_n)_{n\ge1}\), and assume that both signed companion
densities \(d_\Gamma(C_\infty^A)\) and \(d_\Gamma(C_\infty^B)\) exist
\(\nu_{J,p}\)-almost surely whenever the infinite FK cluster exists;
set both densities to zero when it does not exist.

Then the limiting expected magnetization density exists:
\[
   m_\Gamma
   :=\lim_{n\to\infty}\frac{1}{|\Lambda_n|}
      \sum_{x\in\Lambda_n}\mu_{J,\beta}(\omega_x).
\]
It is nonzero if and only if sf--FK percolation occurs along
\(\Gamma\) in \(\nu_{J,p}\). More precisely,
\(|m_\Gamma|=\mathbb E_{\nu_{J,p}}
[\,|d_\Gamma(C_\infty^A)-d_\Gamma(C_\infty^B)|\,]\).
Equivalently, the limiting expected fraction of \(+\) spins differs from
\(1/2\) exactly when the two companions are density-distinguishable along
\(\Gamma\).
\end{theorem}

\begin{proof}
By Theorem~\ref{SpinProjSincleCopy}, the ES lift \(\pi\) is extremal.
Corollary~\ref{cor:extremal_sf_sector} therefore implies that exactly one
of \(E_{\mathrm{null}}\), \(E_{\mathrm{unsel}}\), \(E_+\), and \(E_-\)
has \(\pi\)-probability one.

Conditionally on the FK configuration, every finite FK cluster has its
two compatible spin assignments with equal probability. Each spin in a
finite cluster consequently has conditional mean zero. If
\(\pi(E_{\mathrm{null}})=1\), every finite-volume expected spin average
is therefore zero, and \(m_\Gamma=0\).

If \(\pi(E_{\mathrm{unsel}})=1\), both companion densities exist and are
equal, by the density-existence hypothesis and the definition of the
selector. For any labeling of the companions, compatibility forces the
spin assignment on the infinite FK cluster to be constant on one
companion and opposite on the other. The absolute conditional mean of
the spin average is thus at most
\(\bigl||C_\infty^A\cap\Lambda_n|-|C_\infty^B\cap\Lambda_n|\bigr|
/|\Lambda_n|\). This bound tends to zero almost surely and is bounded by
one. Bounded convergence gives \(m_\Gamma=0\).

If \(\pi(E_+)=1\), the selected, denser companion \(C_\infty^A\)
carries spin \(+1\), and \(C_\infty^B\) carries spin \(-1\).
After conditional averaging, the finite clusters contribute zero and
the infinite cluster contributes
\((|C_\infty^A\cap\Lambda_n|-|C_\infty^B\cap\Lambda_n|)/|\Lambda_n|\).
Bounded convergence yields
\(m_\Gamma=\mathbb E_{\nu_{J,p}}
[d_\Gamma(C_\infty^A)-d_\Gamma(C_\infty^B)]>0\).
If \(\pi(E_-)=1\), the same argument gives the negative of this value.

These cases prove existence of the displayed limit and the stated
absolute-value identity. The absolute density difference vanishes off the selected event
and is strictly positive on it.
Its expectation is positive exactly when that event has positive
\(\nu_{J,p}\)-probability, which is sf--FK percolation along \(\Gamma\).
The final assertion follows from
\(\mathbf1_{\{\omega_x=+1\}}=(1+\omega_x)/2\).
\end{proof}

\begin{corollary}[Coexistence detected by some van Hove sequence]
\label{cor:coexistence_some_vanhove}
Fix a coupling realization \(J\), a finite inverse temperature
\(\beta\ge0\), and \(p=1-e^{-2\beta}\). Work within the class of
extremal Gibbs states whose ES lifts satisfy the
\(0/1\)-infinite-cluster property. In both assertions below, require
that the two signed companion densities exist almost surely under the
corresponding FK marginal along the stated deterministic van~Hove
sequence \(\Gamma\), with both densities zero when there is no infinite
FK cluster. Let \(m_\Gamma\) be the limiting expected magnetization
density from Theorem~\ref{EquivalenceOneCopy}.

Then the following are equivalent.
\begin{enumerate}
\item[\textup{(i)}]
There exist an extremal Gibbs state \(\mu_{J,\beta}\) in this class and
a deterministic van~Hove sequence \(\Gamma\), satisfying the preceding
density-existence requirement, such that \(m_\Gamma\ne0\).
\item[\textup{(ii)}]
There exist an extremal ES state \(\pi\in\mathcal G^{\mathrm{ES}}_{J,\beta}\)
with \(\pi(N_\infty\le1)=1\) and a deterministic van~Hove sequence
\(\Gamma\), satisfying the same density-existence requirement, such that
\(\nu_{J,p}:=\Pi_{\mathrm{FK}}\pi\) has sf--FK percolation along
\(\Gamma\).
\end{enumerate}
In this case the spin-flipped state has limiting expected magnetization
density \(-m_\Gamma\) along the same sequence. Thus there are two
extremal Gibbs states with opposite nonzero limiting expected
magnetization densities within this class.
\end{corollary}

\begin{proof}
Fix a state and a deterministic van~Hove sequence \(\Gamma\) satisfying
the stated density-existence requirement. Theorem~\ref{SpinProjSincleCopy}
preserves extremality between spin states and their ES lifts, and
Theorem~\ref{EquivalenceOneCopy} gives the equivalence between
\(m_\Gamma\ne0\) and sf--FK percolation along \(\Gamma\).
Taking the existential quantifier over these state--sequence pairs
proves the equivalence of \textup{(i)} and \textup{(ii)}.

A global spin flip preserves the Hamiltonian, extremality and the FK
marginal, and hence the companion-density hypothesis. It reverses every
spin expectation and therefore changes \(m_\Gamma\) to \(-m_\Gamma\).
The two states are distinct because these limits are nonzero and
opposite.
\end{proof}

In terms of the order parameter we have the following result.

\begin{theorem}\label{Th1}
In \(\mathbb Z^d\), fix \(\lambda\in[0,1]\) and \(0\le\beta<\infty\),
put \(p=1-e^{-2\beta}<1\), and let
\(\rho^{\mathrm{FK}}\) be a weak limit along a deterministic subsequence
of the joint distributions of couplings \(J\) and FK bond configurations
\(\eta\) with periodic boundary conditions, as in
Lemma~\ref{lem:periodic-joint-ES}. Fix a deterministic van~Hove
sequence \(\Gamma=(\Lambda_n)_{n\ge1}\) for the density averages. Assume that
\(\rho^{\mathrm{FK}}\) satisfies the \(0/1\)-infinite-cluster property for FK clusters.
Assume also that, on the event where the two sf--FK companions of the infinite FK cluster are
distinguished by the density selector \(\mathcal S_{D,\Gamma}\), their densities along
\((\Lambda_n)_{n\ge1}\) exist \(\rho^{\mathrm{FK}}\)-a.s.  More precisely, if
\(C_\infty^A(\eta)=\mathcal S_{D,\Gamma}(\eta)\) denotes the selected, denser companion
and \(C_\infty^B(\eta)\) denotes the other companion, assume that
 \(d_A(\eta)>d_B(\eta)\).
Set \(d_A=d_B=0\) outside that tail event. For standard cubes, existence
of both companion densities follows from
Lemma~\ref{lem:stationary-companion-densities}.

Choose a Borel version of the regular conditional distribution
\(
   \nu_{J,p}:=\rho^{\mathrm{FK}}(\,\cdot\,|J)
\).
These conditional laws satisfy
\[
   \nu_{J,p}\in\mathcal G^{\mathrm{FK}}_{J,p},
   \qquad
   \nu_{J,p}(N_\infty\le1)=1
\]
for \(P_\lambda\)-almost every \(J\).

Let \(\pi^{+}_{J,p}\) be the ES lift of \(\nu_{J,p}\) obtained as follows:
finite FK clusters are assigned their two compatible spin assignments with
probability \(1/2\); on the unselected tail event the companions of the
infinite FK cluster, if present, are also assigned symmetrically; and on the selected tail event the
selected, denser companion \(C_\infty^A(\eta)=\mathcal S_{D,\Gamma}(\eta)\) is assigned
spin \(+1\), while its companion is assigned spin \(-1\). Let
\(
   \mu^{+}_{J,\beta}:=\Pi_{\mathrm{SPIN}}\pi^{+}_{J,p}
\)
be its spin marginal, and let \(\rho^{+}\) be the corresponding quenched joint
distribution of \((J,\omega,\eta)\).

Define the joint disorder--FK event
\[
 \mathcal E^{\mathrm{sf}}_\Gamma
 :=\{(J,\eta):N_\infty(\eta)=1,
              \ \mathcal S_{D,\Gamma}(J,\eta)\ne\varnothing\}.
\]
Define also the events
\begin{eqnarray}
      A_0
      \;&=&\;
      \Bigl\{
         0\text{ belongs to the selected, denser sf--FK companion} \\
         && \hskip 1cm \text{of the unique infinite FK cluster}
      \Bigr\},
      \qquad \\
      B_0
      \;&=&\;
      \Bigl\{
         0\text{ belongs to the other sf--FK companion of the unique
         infinite FK cluster}
      \Bigr\}.
\end{eqnarray}
Both events are joint disorder--FK events, empty off
\(\mathcal E^{\mathrm{sf}}_\Gamma\); their quenched probabilities refer
to their \(J\)-sections.
Then the magnetization of the \(+\)-oriented lift,
\[
      m
      \;:=\;
      \mathbb E_{P_{\lambda}}
      \bigl[
            \mu^{+}_{J,\beta}(\omega_0)
      \bigr]
      \;=\;
      \rho^{+}(\omega_0),
\]
satisfies
\[
      m
      \;=\;
      \rho^{\mathrm{FK}}(A_0)-\rho^{\mathrm{FK}}(B_0)
      \;=\;
      \mathbb E_{P_{\lambda}}
      \Bigl[
           \nu_{J,p}(A_0)-\nu_{J,p}(B_0)
      \Bigr].
\]
Moreover,
\[
      m
      =
      \mathbb E_{\rho^{\mathrm{FK}}}
      \bigl[
         d_A(\eta)-d_B(\eta)
      \bigr].
\]
Consequently,
\[
      m>0
      \quad\Longleftrightarrow\quad
      \rho^{\mathrm{FK}}\bigl(\mathcal E^{\mathrm{sf}}_\Gamma\bigr)>0.
\]
\end{theorem}

\begin{proof}
Lemma~\ref{lem:periodic-joint-ES} supplies a joint ES lift of
\(\rho^{\mathrm{FK}}\) whose conditional laws given \(J\) are ES
Gibbs states. Their FK marginals are \(\nu_{J,p}\), and the assumed
\(0/1\) property holds under these conditional laws by disintegration.
Applying Theorem~\ref{thm:FK-surjective} to these ES states gives
\(\nu_{J,p}\in\mathcal G^{\mathrm{FK}}_{J,p}\) for
\(P_\lambda\)-almost every \(J\).
The existence of the lift \(\pi^+_{J,p}\) follows from the ES lifting
construction under the \(0/1\)-infinite-cluster property, together with the restricted
FK-projection statement for fixed bias \(q\).  In the selected tail event we
are using the case \(q=1\): the selected companion receives spin \(+1\), and
its signed companion receives spin \(-1\).  In the null and unselected tail events
we use the symmetric lift.  Since the conditioning events are tail events and the
corresponding conditional ES measures satisfy the same DLR specification,
the resulting lift is again an ES Gibbs state for the same interaction.
The coloring is a jointly Borel kernel in \((J,\eta)\): enumerate the
cluster roots for the independent fair signs, and use
Lemma~\ref{lem:joint-density-measurability} for the denser-companion choice.
Composition with \(\nu_{J,p}\) makes \(\pi^+_{J,p}\) and
\(\mu^+_{J,\beta}\) measurable kernels and defines \(\rho^+\).

We first compute the conditional expectation of \(\omega_0\) given the
quenched FK configuration.  By construction of \(\pi^+_{J,p}\),
\[
   \mathbb E_{\pi^+_{J,p}}\bigl[\omega_0\,\big|\,\eta\bigr]
   =
   \mathbf 1_{A_0}(\eta)-\mathbf 1_{B_0}(\eta).
\]
Indeed, if \(0\in A_0\), the selected companion carries spin \(+1\), while if
\(0\in B_0\), the other companion carries spin \(-1\).  On the complement of
\(A_0\cup B_0\), either there is no infinite FK cluster, or the infinite signed
companions are not selected by \(\mathcal S_{D,\Gamma}\), or \(0\) lies in a finite FK
cluster.  In all these cases the conditional spin assignment is symmetric, so
the conditional expectation of \(\omega_0\) is \(0\).

Integrating this identity first with respect to \(\pi^+_{J,p}\) and then over
\(J\), we obtain
\[
\begin{aligned}
   m
   &= \mathbb E_{P_\lambda}
      \bigl[
         \mu^+_{J,\beta}(\omega_0)
      \bigr]                                                   
   = \mathbb E_{P_\lambda}
      \bigl[
         \nu_{J,p}(A_0)-\nu_{J,p}(B_0)
      \bigr]                                                   \\
   &= \rho^{\mathrm{FK}}(A_0)-\rho^{\mathrm{FK}}(B_0).
\end{aligned}
\]
Here we used that \(A_0\) and \(B_0\) are FK events, so their probabilities
under the lifted measure \(\rho^+\) are the same as under the FK marginal
\(\rho^{\mathrm{FK}}\).

It remains to identify this difference with the density imbalance.  For
\(x\in\mathbb Z^d\), let \(A_x\) and \(B_x\) be the translates of \(A_0\) and
\(B_0\), namely the events that \(x\) belongs respectively to the selected
and to the other sf--FK companion.  Since \(\rho^{\mathrm{FK}}\) is a
periodic-boundary-condition weak limit, it is translation invariant.
The denser-companion choice along the fixed \(\Gamma\) is translation-covariant:
translating a box changes a density fraction by at most
\(\lvert(\Lambda_n-z)\triangle\Lambda_n\rvert/|\Lambda_n|\to0\).
Hence, for every \(n\),
\[
\begin{aligned}
   \rho^{\mathrm{FK}}(A_0)-\rho^{\mathrm{FK}}(B_0)
   &=
   \frac{1}{|\Lambda_n|}
   \sum_{x\in\Lambda_n}
   \Bigl(
      \rho^{\mathrm{FK}}(A_x)-\rho^{\mathrm{FK}}(B_x)
   \Bigr)                                                     \\
   &=
   \mathbb E_{\rho^{\mathrm{FK}}}
   \left[
      \frac{
         |C_\infty^A(\eta)\cap\Lambda_n|
         -
         |C_\infty^B(\eta)\cap\Lambda_n|
      }{|\Lambda_n|}
   \right],
\end{aligned}
\]
where the last expression is understood to be \(0\) outside the selected
tail event.

By the assumed existence of the densities \(d_A(\eta)\) and \(d_B(\eta)\), and
by bounded convergence,
\[
   \rho^{\mathrm{FK}}(A_0)-\rho^{\mathrm{FK}}(B_0)
   =
   \mathbb E_{\rho^{\mathrm{FK}}}
   \bigl[
      d_A(\eta)-d_B(\eta)
   \bigr].
\]
Therefore
\(
   m
   =
   \mathbb E_{\rho^{\mathrm{FK}}}
   \bigl[
      d_A(\eta)-d_B(\eta)
   \bigr].
\)

By construction,
\(
   d_A(\eta)-d_B(\eta)\ge 0,
\)
and this quantity is strictly positive exactly on
\(\mathcal E^{\mathrm{sf}}_\Gamma\), the event that the two signed
companions of the infinite FK cluster are density-distinguishable by the selector. Hence
\[
   m>0
   \quad\Longleftrightarrow\quad
   \rho^{\mathrm{FK}}
   \bigl(
      d_A(\eta)-d_B(\eta)>0
   \bigr)>0,
\]
which is precisely the final equivalence.
\end{proof}

\section{Replicas}
\subsection{Infinite-volume two-replica ES measures for fixed couplings \texorpdfstring{$J$}{J}}

To study replica overlap we use two replicas with the same quenched
couplings \(J\). We distinguish the full Gibbs simplex associated with
the product two-replica specification from its product subclass,
which describes independent replicas. The joint periodic-boundary
limits constructed in Section~\ref{subsec:periodic-joint-limits} give
conditional laws in the full two-replica class.

Let
\[
   \Omega^{2}:=\Omega^{(1)}\times\Omega^{(2)},
   \qquad
   \mathcal H^{2}:=\mathcal H^{(1)}\times\mathcal H^{(2)} ,
\]
with generic elements
\[
   \omega^{2}=(\omega^{(1)},\omega^{(2)}),
   \qquad
   \eta^{2}=(\eta^{(1)},\eta^{(2)}).
\]
For a finite volume \(\Lambda\Subset V\) and boundary spins
\(\omega^{(1)}_{\partial\Lambda},\omega^{(2)}_{\partial\Lambda}\), define
the two-replica ES kernel by
\[
   \pi^{2}_{J,\beta,\Lambda,\,
             \omega^{(1)}_{\partial\Lambda},
             \omega^{(2)}_{\partial\Lambda}}
   :=
   \pi^{(1)}_{J,\beta,\Lambda,\omega^{(1)}_{\partial\Lambda}}
     \otimes
   \pi^{(2)}_{J,\beta,\Lambda,\omega^{(2)}_{\partial\Lambda}} .
\]
Thus, conditionally on the two boundary spin configurations, the two
replicas are independent.

\begin{definition}[Full two-replica ES Gibbs measures]
\label{def:ES-DLR-2}
The set
\(
   \mathcal G^{\mathrm{ES},2}_{J,\beta}
\)
consists of all probability measures \(\pi^{2}\) on
\(\Omega^{2}\times\mathcal H^{2}\) satisfying the DLR equations with the
above product specification. More precisely, for every finite
$\Lambda\Subset V$, the conditional law of
$(\omega^{(1)}_\Lambda,\omega^{(2)}_\Lambda,
  \eta^{(1)}_{\overline E_\Lambda},\eta^{(2)}_{\overline E_\Lambda})$
given
\[
 \mathcal A^2_{\Lambda^c}:=
 \sigma\bigl(\omega^{(r)}_{V\setminus\Lambda},
              \eta^{(r)}_{E\setminus\overline E_\Lambda}:r=1,2\bigr)
\]
is $\pi^2_{J,\beta,\Lambda,
\omega^{(1)}_{\partial\Lambda},\omega^{(2)}_{\partial\Lambda}}$,
$\pi^2$-almost surely.
\end{definition}

\begin{definition}[Independent two-replica ES measures]
\label{def:ES-product-replicas}
The independent two-replica ES class is
\[
   \mathcal G^{\mathrm{ES},\otimes 2}_{J,\beta}
   :=
   \Bigl\{
      \pi^{(1)}\otimes\pi^{(2)}:
      \pi^{(1)},\pi^{(2)}\in\mathcal G^{\mathrm{ES}}_{J,\beta}
   \Bigr\}.
\]
Similarly,
\[
   \mathcal G^{\mathrm{SPIN},\otimes 2}_{J,\beta}
   :=
   \Bigl\{
      \mu^{(1)}\otimes\mu^{(2)}:
      \mu^{(1)},\mu^{(2)}\in\mathcal G^{\mathrm{SPIN}}_{J,\beta}
   \Bigr\}.
\]
When both replicas are sampled from the same state, we use the diagonal
subclasses
\[
   \mathcal G^{\mathrm{ES},\otimes 2,\mathrm{diag}}_{J,\beta}
   :=
   \Bigl\{
      \pi\otimes\pi:
      \pi\in\mathcal G^{\mathrm{ES}}_{J,\beta}
   \Bigr\},
\]
and
\[
   \mathcal G^{\mathrm{SPIN},\otimes 2,\mathrm{diag}}_{J,\beta}
   :=
   \Bigl\{
      \mu\otimes\mu:
      \mu\in\mathcal G^{\mathrm{SPIN}}_{J,\beta}
   \Bigr\}.
\]
\end{definition}

\begin{remark}
The full class \(\mathcal G^{\mathrm{ES},2}_{J,\beta}\) is convex and contains
all Gibbs states for the product two-replica ES specification.  We use
\(\mathcal G^{\mathrm{SPIN},2}_{J,\beta}\) for the analogous full spin Gibbs
class associated with the product two-replica spin specification.  The product
class \(\mathcal G^{\mathrm{ES},\otimes 2}_{J,\beta}\) is the class relevant
for independent replicas.  It is compact, but it is not convex in general.
Indeed, a convex combination of two product measures need not be a product
measure.
\end{remark}

\begin{theorem}[Two-replica analogue of Theorem~\ref{SpinProjSincleCopy}]
\label{thm:ES-2replica}
For every fixed coupling configuration \(J\in\{\pm1\}^{E}\) and inverse
temperature \(\beta\), the following hold.

\begin{enumerate}
\item[\textup{(a)}]
The full two-replica ES Gibbs class
\(\mathcal G^{\mathrm{ES},2}_{J,\beta}\) is non-empty, weak-* compact and
convex.

\item[\textup{(b)}]
\(\mathcal G^{\mathrm{ES},2}_{J,\beta}\) is the closed convex hull of weak
limits of finite-volume two-replica ES measures with deterministic boundary
conditions.

\item[\textup{(c)}]
\(\mathcal G^{\mathrm{ES},2}_{J,\beta}\) is a Choquet simplex. Its extremal
points are the tail-trivial two-replica ES Gibbs measures.

\item[\textup{(d)}]
Let
\[
   \Pi^{2}_{\mathrm{SPIN}}
   :
   (\omega^{(1)},\omega^{(2)},\eta^{(1)},\eta^{(2)})
   \longmapsto
   (\omega^{(1)},\omega^{(2)})
\]
be the projection onto the two spin configurations.  Then
\[
   \Pi^{2}_{\mathrm{SPIN}}
   :
   \mathcal G^{\mathrm{ES},2}_{J,\beta}
   \longrightarrow
   \mathcal G^{\mathrm{SPIN},2}_{J,\beta}
\]
is an affine isomorphism between the full two-replica ES and spin Gibbs
simplices.

Moreover, this isomorphism restricts to the independent product subclasses:
\[
   \Pi^{2}_{\mathrm{SPIN}}
   :
   \mathcal G^{\mathrm{ES},\otimes 2}_{J,\beta}
   \longrightarrow
   \mathcal G^{\mathrm{SPIN},\otimes 2}_{J,\beta}
\]
is a bijection.  Explicitly, if
\[
   \pi^{2}=\pi^{(1)}\otimes\pi^{(2)},
   \qquad
   \mu^{(i)}=\Pi_{\mathrm{SPIN}}\pi^{(i)},
\]
then
\[
   \Pi^{2}_{\mathrm{SPIN}}\pi^{2}
   =
   \mu^{(1)}\otimes\mu^{(2)}.
\]
Conversely, if
\[
   \mu^{2}=\mu^{(1)}\otimes\mu^{(2)}
   \in \mathcal G^{\mathrm{SPIN},\otimes 2}_{J,\beta},
\]
then its unique independent ES lift is
\[
   \pi^{2}
   =
   \pi^{(1)}\otimes\pi^{(2)},
\]
where \(\pi^{(i)}\) is the unique one-replica ES lift of \(\mu^{(i)}\).

In particular, the same restriction gives a bijection between the diagonal
subclasses:
\[
   \mathcal G^{\mathrm{ES},\otimes 2,\mathrm{diag}}_{J,\beta}
   \longleftrightarrow
   \mathcal G^{\mathrm{SPIN},\otimes 2,\mathrm{diag}}_{J,\beta}.
\]
\end{enumerate}
\end{theorem}

\begin{proof}
The full two-replica specification is the product of two copies of the
one-replica ES specification.  Therefore it is again a finite-range,
quasilocal specification with the same type of hard-core compatibility
constraint.  The usual DLR arguments used in
Theorem~\ref{SpinProjSincleCopy} apply verbatim to this product
specification: non-emptiness, compactness, convexity, the closed-convex-hull
description by finite-volume limits, and the Choquet simplex property all
follow exactly as in the one-replica case.  This proves
\textup{(a)--(c)}.

We next prove the projection statement.  

If
\(\pi^{2}\in\mathcal G^{\mathrm{ES},2}_{J,\beta}\), then its spin marginal
satisfies the two-replica spin DLR equations for the product spin
specification, because summing the two-replica ES kernel over the two bond
configurations gives the corresponding two-replica spin kernel.  Thus
\(
   \Pi^{2}_{\mathrm{SPIN}}\pi^{2}
   \in \mathcal G^{\mathrm{SPIN},2}_{J,\beta}.
\)

Conversely, let
\(
   \mu^{2}\in\mathcal G^{\mathrm{SPIN},2}_{J,\beta}.
\)
Conditionally on \((\omega^{(1)},\omega^{(2)})\), sample the FK variables
\(\eta^{(1)}\) and \(\eta^{(2)}\) independently according to the usual
one-replica Edwards--Sokal conditional bond kernels in the two replicas.
This gives a measure \(\pi^{2}\) on
\(\Omega^{2}\times\mathcal H^{2}\).  By construction, its spin marginal is
\(\mu^{2}\), and its finite-volume conditional distributions are exactly the
two-replica ES kernels.  Hence
\(
   \pi^{2}\in\mathcal G^{\mathrm{ES},2}_{J,\beta}.
\)

The conditional law of the bonds given the spins is fixed by the ES
construction, so this lift is unique.  Therefore
\[
   \Pi^{2}_{\mathrm{SPIN}}
   :
   \mathcal G^{\mathrm{ES},2}_{J,\beta}
   \to
   \mathcal G^{\mathrm{SPIN},2}_{J,\beta}
\]
is an affine bijection.  Since both sides are Choquet simplices, it is an
affine isomorphism of simplices.

Finally, suppose
\[
   \pi^{2}=\pi^{(1)}\otimes\pi^{(2)}
   \in \mathcal G^{\mathrm{ES},\otimes 2}_{J,\beta}.
\]
Then
\[
   \Pi^{2}_{\mathrm{SPIN}}\pi^{2}
   =
   (\Pi_{\mathrm{SPIN}}\pi^{(1)})
   \otimes
   (\Pi_{\mathrm{SPIN}}\pi^{(2)})
   \in \mathcal G^{\mathrm{SPIN},\otimes 2}_{J,\beta}.
\]
Conversely, if
\[
   \mu^{2}=\mu^{(1)}\otimes\mu^{(2)}
   \in \mathcal G^{\mathrm{SPIN},\otimes 2}_{J,\beta},
\]
then by Theorem~\ref{SpinProjSincleCopy} each \(\mu^{(i)}\) has a unique
one-replica ES lift \(\pi^{(i)}\).  Hence
\[
   \pi^{(1)}\otimes\pi^{(2)}
\]
is the unique independent ES lift of \(\mu^{2}\).  This proves the
bijection on the independent product subclasses.  The diagonal statement is
the special case \(\pi^{(1)}=\pi^{(2)}\) and
\(\mu^{(1)}=\mu^{(2)}\).

\end{proof}

\begin{proposition}[Extremal full two-replica states]
\label{prop:full-two-replica-extremals-factor}
Let
\(
   \rho\in\mathcal G^{\mathrm{ES},2}_{J,\beta}
\)
be extremal for the full product two-replica ES specification. Then
\[
   \rho=\pi^{(1)}\otimes\pi^{(2)}
\]
for some extremal one-replica ES states
\(
   \pi^{(1)},\pi^{(2)}\in\mathcal G^{\mathrm{ES}}_{J,\beta}.
\)
Consequently, the extremal points of the full two-replica ES simplex are
products of extremal one-replica ES states.
\end{proposition}

\begin{proof}
Let \(\mathscr T^{(1)}\), \(\mathscr T^{(2)}\), and \(\mathscr T^2\) denote,
respectively, the tail sigma-fields of the first replica, of the second replica,
and of the full two-replica system. Since \(\rho\) is extremal for the full
product specification, it is tail-trivial; hence every event in
\(\mathscr T^2\) has \(\rho\)-probability \(0\) or \(1\). In particular,
\(\mathscr T^{(1)}\) and \(\mathscr T^{(2)}\), viewed as sub-sigma-fields of
\(\mathscr T^2\), are trivial under \(\rho\).

Let \(f_1\) and \(f_2\) be bounded local functions of the first and second
replicas, respectively. Let \((\Lambda_n)_{n\ge1}\) be an increasing sequence
of finite sets exhausting \(V\) and containing the supports of both functions. Since the
finite-volume two-replica ES specification is the product of the one-replica
specifications, the DLR equations give
\[
   \mathbb E_\rho
   \left[
      f_1 f_2
      \mid
      \mathscr A^2_{\Lambda_n^c}
   \right]
   =
   \Gamma^{(1)}_{\Lambda_n}f_1\,
   \Gamma^{(2)}_{\Lambda_n}f_2,
\]
where \(\mathscr A^2_{\Lambda_n^c}\) is the two-replica exterior sigma-field
and \(\Gamma^{(i)}_{\Lambda_n}\) is the one-replica ES specification operator
in replica \(i\). Letting \(n\to\infty\) and using the reverse martingale
convergence theorem yields
\[
   \mathbb E_\rho
   \left[
      f_1 f_2
      \mid
      \mathscr T^2
   \right]
   =
   \mathbb E_\rho
   \left[
      f_1
      \mid
      \mathscr T^{(1)}
   \right]
   \mathbb E_\rho
   \left[
      f_2
      \mid
      \mathscr T^{(2)}
   \right].
\]
By tail triviality, the right-hand side is
\(
   \mathbb E_\rho[f_1]\,\mathbb E_\rho[f_2],
\)
and the left-hand side is \(\mathbb E_\rho[f_1f_2]\). Therefore
\[
   \mathbb E_\rho[f_1f_2]
   =
   \mathbb E_\rho[f_1] \mathbb E_\rho[f_2]
\]
for all bounded local \(f_1,f_2\). By a monotone class argument, the two
replicas are independent under \(\rho\). Thus
\(
   \rho=\pi^{(1)}\otimes\pi^{(2)},
\)
where \(\pi^{(i)}\) is the \(i\)-th marginal of \(\rho\).

Each marginal is a one-replica ES Gibbs state, because summing the full
product DLR equations over the other replica gives the one-replica ES DLR
equations. Each marginal is also tail-trivial, since \(\mathscr T^{(i)}\) is
trivial under \(\rho\). Hence \(\pi^{(1)}\) and \(\pi^{(2)}\) are extremal
one-replica ES states.

Conversely, the same reverse-martingale argument for products of local
functions shows that a product of two tail-trivial one-replica Gibbs
states is tail-trivial for the full two-replica specification, and
hence extremal.
\end{proof}

\begin{remark}
The full class \(\mathcal G^{\mathrm{ES},2}_{J,\beta}\) contains both
auxiliary tail-biased states and the conditional ES laws obtained by
disintegrating joint periodic-boundary limits over the couplings.
The latter construction is given in Lemma~\ref{lem:periodic-joint-ES}.
Corollary~\ref{cor:periodic_pc_overlap} applies the product-extremal
overlap theorem to the extremal components of these conditional laws.
\end{remark}

We also use the independent two-replica FK class
\[
   \mathcal G^{\mathrm{FK},\otimes 2}_{J,p}
   :=
   \Bigl\{
      \nu^{(1)}\otimes\nu^{(2)}:
      \nu^{(1)},\nu^{(2)}\in\mathcal G^{\mathrm{FK}}_{J,p}
   \Bigr\}.
\]

\begin{theorem}[Two--replica FK marginals of independent ES states]
\label{thm:FK-2replica-surjective}
Let
\[
   \mathcal G^{\mathrm{ES},\otimes 2}_{\le1}
   :=
   \Bigl\{
      \pi^{(1)}\otimes \pi^{(2)}:
      \pi^{(r)}\in\mathcal G^{\mathrm{ES}}_{J,\beta},
      \ \pi^{(r)}(N^{(r)}_\infty\le1)=1,\ r=1,2
   \Bigr\},
\]
and
\[
   \mathcal G^{\mathrm{FK},\otimes 2}_{\le1}
   :=
   \Bigl\{
      \nu^{(1)}\otimes \nu^{(2)}:
      \nu^{(r)}\in\mathcal G^{\mathrm{FK}}_{J,p},
      \ \nu^{(r)}(N^{(r)}_\infty\le1)=1,\ r=1,2
   \Bigr\}.
\]
Here \(N^{(r)}_\infty\) denotes the number of ordinary infinite FK clusters
in replica \(r\).

Then the map
\[
   \Pi^{\otimes 2}_{\mathrm{FK}}:
     \mathcal G^{\mathrm{ES},\otimes 2}_{\le1}
     \longrightarrow
     \mathcal G^{\mathrm{FK},\otimes 2}_{\le1}
\]
is surjective.  In other words, every independent product two--replica FK
Gibbs state whose one--replica marginals satisfy the \(0/1\)-infinite-cluster property is the
bond marginal of an independent product two--replica Edwards--Sokal state
for the same coupling realization \(J\).

Moreover, every state in either class has at most one infinite paired
FK cluster (p--FK cluster, defined in Section~\ref{subsec:pc-clusters})
almost surely.
\end{theorem}

\begin{proof}[Proof of Theorem~\ref{thm:FK-2replica-surjective}]
The proof uses the one--replica surjectivity result, applied separately in
the two replicas.

\smallskip
\noindent
\emph{Step 1: the FK marginal of an independent two--replica ES state is an
independent two--replica FK state.}
Let
\[
   \pi^{\otimes 2}
   =
   \pi^{(1)}\otimes\pi^{(2)}
   \in
   \mathcal G^{\mathrm{ES},\otimes 2}_{\le1}.
\]
For each \(r=1,2\), set
\(
   \nu^{(r)}:=\Pi_{\mathrm{FK}}\pi^{(r)}.
\)
Since
\[
   \pi^{(r)}\in\mathcal G^{\mathrm{ES}}_{J,\beta},
   \qquad
   \pi^{(r)}(N^{(r)}_\infty\le1)=1,
\]
the one--replica FK marginal theorem implies
\[
   \nu^{(r)}\in\mathcal G^{\mathrm{FK}}_{J,p},
   \qquad
   \nu^{(r)}(N^{(r)}_\infty\le1)=1.
\]
Therefore
\[
   \Pi^{\otimes 2}_{\mathrm{FK}}\pi^{\otimes 2}
   =
   \nu^{(1)}\otimes\nu^{(2)}
   \in
   \mathcal G^{\mathrm{FK},\otimes 2}_{\le1}.
\]
Thus
\(
   \Pi^{\otimes 2}_{\mathrm{FK}}
   \bigl(
      \mathcal G^{\mathrm{ES},\otimes 2}_{\le1}
   \bigr)
   \subseteq
   \mathcal G^{\mathrm{FK},\otimes 2}_{\le1}.
\)

\smallskip
\noindent
\emph{Step 2: surjectivity.}
Let
\[
   \nu^{\otimes 2}
   =
   \nu^{(1)}\otimes\nu^{(2)}
   \in
   \mathcal G^{\mathrm{FK},\otimes 2}_{\le1}.
\]
By definition,
\[
   \nu^{(r)}\in\mathcal G^{\mathrm{FK}}_{J,p},
   \qquad
   \nu^{(r)}(N^{(r)}_\infty\le1)=1,
   \qquad r=1,2.
\]
By the one--replica surjectivity theorem, for each \(r=1,2\) there exists
\(
   \pi^{(r)}\in\mathcal G^{\mathrm{ES}}_{J,\beta}
\)
such that
\[
   \Pi_{\mathrm{FK}}\pi^{(r)}=\nu^{(r)}
   \qquad\text{and}\qquad
   \pi^{(r)}(N^{(r)}_\infty\le1)=1.
\]
Now define
\(
   \pi^{\otimes 2}
   :=
   \pi^{(1)}\otimes\pi^{(2)}.
\)
Since each factor is an ES Gibbs state for the same quenched coupling
configuration \(J\), the product measure \(\pi^{\otimes 2}\) satisfies the
two--replica product ES specification.  Hence
\(
   \pi^{\otimes 2}
   \in
   \mathcal G^{\mathrm{ES},\otimes 2}_{\le1}.
\)
Moreover,
\[
   \Pi^{\otimes 2}_{\mathrm{FK}}\pi^{\otimes 2}
   =
   (\Pi_{\mathrm{FK}}\pi^{(1)})
   \otimes
   (\Pi_{\mathrm{FK}}\pi^{(2)})
   =
   \nu^{(1)}\otimes\nu^{(2)}
   =
   \nu^{\otimes 2}.
\]
This proves surjectivity.

\smallskip
\noindent
\emph{Step 3: uniqueness of the infinite p--FK cluster.}
It remains only to record why the paired \(0/1\)-infinite-cluster property follows from the
one--replica \(0/1\)-infinite-cluster property in each copy.  A p--FK cluster is an
intersection of an ordinary FK cluster in replica \(1\) with an ordinary FK
cluster in replica \(2\).  If each replica has at most one infinite ordinary
FK cluster, then there is at most one possible intersection of two infinite
ordinary FK clusters.  Any intersection involving a finite ordinary FK
cluster is finite.  Hence there is at most one infinite p--FK cluster almost
surely.

Therefore every state in
\(\mathcal G^{\mathrm{ES},\otimes 2}_{\le1}\) and in
\(\mathcal G^{\mathrm{FK},\otimes 2}_{\le1}\) has
\(N^p_\infty\le1\) almost surely.
\end{proof}

\subsection{Joint periodic-boundary limits}
\label{subsec:periodic-joint-limits}

Fix \(\lambda\in[0,1]\), \(0\le\beta<\infty\), and
\(p=1-e^{-2\beta}<1\). On the torus
\(\mathbb T_L=(\mathbb Z/L\mathbb Z)^d\), sample independent couplings
with \(P(J_e=+1)=\lambda\). Given these couplings, sample \(k\in\{1,2\}\)
independent copies of the normalized periodic ES measure. Extend the
couplings, spins, and bonds periodically to \(\mathbb Z^d\), and denote
their joint law by \(\rho_L^{\mathrm{ES},k}\).

\begin{lemma}[Quenched Gibbs laws in joint periodic limits]
\label{lem:periodic-joint-ES}
Every weak limit \(\rho^{\mathrm{ES},k}\) of these laws along a
deterministic subsequence \(L\to\infty\) is translation invariant and
has coupling marginal \(P_\lambda\). In a disintegration
\[
 \rho^{\mathrm{ES},k}(dJ,dz)=P_\lambda(dJ)\,\pi_J^k(dz),
\]
where \(z\) denotes all spin and bond variables, the conditional law
belongs to \(\mathcal G^{\mathrm{ES}}_{J,\beta}\) for \(k=1\), and to
\(\mathcal G^{\mathrm{ES},2}_{J,\beta}\) for \(k=2\), for
\(P_\lambda\)-almost every \(J\).

If each replica has at most one infinite ordinary FK cluster under the
joint limit, then each one-replica FK marginal of \(\pi_J^k\) belongs
to \(\mathcal G^{\mathrm{FK}}_{J,p}\) for almost every \(J\).
Every joint weak limit of the corresponding finite-volume periodic FK
laws is the FK marginal of a joint ES limit with the asserted stationarity
and quenched ES Gibbs properties.
\end{lemma}

\begin{proof}
The finite joint laws are invariant under translations of the torus,
and hence their periodic extensions are translation invariant. This
property passes to weak limits. Their coupling marginals converge to
\(P_\lambda\), because every fixed finite set of lattice edges embeds
without identifications in all sufficiently large tori.

Write \(\gamma^{k}_{J,\Lambda}\) for the one-replica ES specification
kernel when \(k=1\), and for the product kernel in
Definition~\ref{def:ES-DLR-2} when \(k=2\). Fix a finite \(\Lambda\),
a bounded local function \(f\) of the ES variables, and a bounded
cylinder function \(g\) of \(J\) and the exterior ES variables.
For sufficiently large \(L\), all coordinates entering these tests
and the local kernel embed in the torus without identifications.
The finite-volume DLR equations then give
\[
 \int f g\,d\rho_L^{\mathrm{ES},k}
 =
 \int g\,\gamma^{k}_{J,\Lambda}f\,d\rho_L^{\mathrm{ES},k}.
\]
Both integrands are continuous cylinder functions: the ES kernels
depend on finitely many coupling, spin, and bond coordinates.
The identity therefore passes to the joint weak limit. A monotone
class argument extends the exterior tests. Disintegrating over \(J\)
and using a countable determining family of local DLR identities shows
that \(\pi_J^k\) satisfies the required quenched ES specification for
almost every \(J\).

The one-replica marginals of a full two-replica ES Gibbs state are
one-replica ES Gibbs states. The assumed \(0/1\) properties pass to
these conditional marginals by disintegration, so
Theorem~\ref{thm:FK-surjective} gives the asserted FK Gibbs property.

Finally, given a convergent sequence of joint periodic FK laws, take
their finite-volume ES lifts. Compactness of the joint configuration
space supplies a further deterministic subsequence along which these
lifts converge. Projection preserves weak convergence, so the resulting
ES limit has the prescribed joint FK marginal. The preceding arguments
apply to this lift.
\end{proof}

For two replicas, the conditional law \(\pi_J^2\) may be a mixture of
products of extremal one-replica states. Its extremal decomposition is
described by Theorem~\ref{thm:ES-2replica} and
Proposition~\ref{prop:full-two-replica-extremals-factor}.

\section{Paired coherent FK-clusters and paired coherent FK percolation}

Sections~\ref{SFFK1} and~\ref{SFFK2} developed the signed FK description
of uniform magnetization density. In this and the following section we
develop the corresponding description of replica overlap.

The main difference is that we now consider two replicas with the same
quenched couplings. We define a paired coherent FK connection, or
pc--FK connection, between two vertices \(x,y\) as the existence of an
FK-open path from \(x\) to \(y\) in each replica, with the two paths
having the same parity. If \(x\) and \(y\) are FK-connected in both
replicas but the two path parities are opposite, we say that \(x\) and
\(y\) are paired incoherently FK-connected, or pi--FK-connected. The FK
paths in the two replicas need not use the same edges.

Equivalently, in an ES-compatible two-replica spin configuration, write
\[
   q_x:=\omega^{(1)}_x\omega^{(2)}_x .
\]
If \(x\) and \(y\) are pc--FK-connected, then \(q_x=q_y\); if they are
pi--FK-connected, then \(q_x=-q_y\). Thus pc--FK connectivity preserves
the two classes
\[
\left\{
\begin{bmatrix}1\\ 1\end{bmatrix},
\begin{bmatrix}-1\\ -1\end{bmatrix}
\right\},
\qquad
\left\{
\begin{bmatrix}1\\ -1\end{bmatrix},
\begin{bmatrix}-1\\ 1\end{bmatrix}
\right\},
\]
whereas pi--FK connectivity switches them.

We now formalize this notion of connection. Fix \(J\), and consider two
independent FK Gibbs states with the same quenched couplings,
\(
   \nu^{\otimes 2}_{J,p}
   =
   \nu^{(1)}_{J,p}\otimes \nu^{(2)}_{J,p}.
\)
Fix a pair of FK bond configurations
\(
   \eta^{\otimes 2}=(\eta^{(1)},\eta^{(2)})
\)
sampled from this product measure and compatible with \(J\).

For \(r=1,2\), if \(x\) and \(y\) are FK-connected in replica \(r\), let
\(
   \varepsilon_r(x,y)\in\{-1,+1\}
\)
be the product of the couplings \(J_b\) along any FK-open path from \(x\)
to \(y\) in replica \(r\). This is well-defined because the FK
configuration is compatible with \(J\), so no FK-open circuit has negative
coupling product. Define
\[
   s(x,y)
   =
   \begin{cases}
      +1, & \text{if \(x\) and \(y\) are FK-connected in both replicas and }
             \\ &\varepsilon_1(x,y)=\varepsilon_2(x,y),\\[2pt]
      -1, & \text{if \(x\) and \(y\) are FK-connected in both replicas and }
            \\& \varepsilon_1(x,y)=-\varepsilon_2(x,y),\\[2pt]
       0, & \text{otherwise.}
   \end{cases}
\]
This agrees with the definition in~\eqref{SignedDoubleConn}.

We say that \(x,y\in V\) are paired coherently FK-connected, or
pc--FK-connected, and write \(x\xLeftrightarrow{c} y\), if \(s(x,y)=+1\),
and paired incoherently FK-connected, or
pi--FK-connected, and write \(x\xLeftrightarrow{i} y\), if \(s(x,y)=-1\).

\subsection{Example (paired connectivity along different routes)}
\label{subsec:paired-example}
Consider the square with vertices \(x,a,y,b\) in cyclic order, and
identify each FK configuration with its set of occupied bonds:
\[
  \eta^{(1)}=\{xa,ay\},\qquad
  \eta^{(2)}=\{xb,by\}.
\]
The two replicas use the same coupling realization \(J\) in each
case shown in Figure~\ref{fig:paired-routes}. The paired classes and coherent companions used in this
example are defined in Section~\ref{subsec:pc-clusters}.
Write \(\varepsilon_r(x,y)\) for the product of the
couplings along the occupied route from \(x\) to \(y\) in replica \(r\).

\begin{figure}[!htb]
\centering
\begin{tikzpicture}[x=1.02cm,y=1.02cm,line cap=round,line join=round,
  vertex/.style={circle,draw=black,fill=white,inner sep=0pt,minimum size=5pt},
  first/.style={draw=replicaone,line width=1.8pt},
  second/.style={draw=replicatwo,line width=1.8pt,dash pattern=on 5pt off 2.6pt}]
  \begin{scope}
    \node[font=\bfseries] at (1,2.65) {(a) Coherent connection};
    \draw[first] (0,2)--(2,2)--(2,0);
    \draw[second] (0,2)--(0,0)--(2,0);
    \node[vertex,label=above left:{$x$}] at (0,2) {};
    \node[vertex,label=above right:{$a$}] at (2,2) {};
    \node[vertex,label=below right:{$y$}] at (2,0) {};
    \node[vertex,label=below left:{$b$}] at (0,0) {};
    \node[above,font=\small] at (1,2) {$+1$};
    \node[right,font=\small] at (2,1) {$+1$};
    \node[left,font=\small] at (0,1) {$+1$};
    \node[below,font=\small] at (1,0) {$+1$};
    \node at (1,-0.68) {$\varepsilon_1=+1,\quad\varepsilon_2=+1$};
    \node at (1,-1.08) {$s(x,y)=+1$};
  \end{scope}
  \begin{scope}[xshift=6.3cm]
    \node[font=\bfseries] at (1,2.65) {(b) Incoherent connection};
    \draw[first] (0,2)--(2,2)--(2,0);
    \draw[second] (0,2)--(0,0)--(2,0);
    \node[vertex,label=above left:{$x$}] at (0,2) {};
    \node[vertex,label=above right:{$a$}] at (2,2) {};
    \node[vertex,label=below right:{$y$}] at (2,0) {};
    \node[vertex,label=below left:{$b$}] at (0,0) {};
    \node[above,font=\small] at (1,2) {$+1$};
    \node[right,font=\small] at (2,1) {$+1$};
    \node[left,font=\small] at (0,1) {$+1$};
    \node[below,font=\small] at (1,0) {$-1$};
    \node at (1,-0.68) {$\varepsilon_1=+1,\quad\varepsilon_2=-1$};
    \node at (1,-1.08) {$s(x,y)=-1$};
  \end{scope}
  \draw[first] (-0.8,-1.73)--(-0.05,-1.73);
  \node[anchor=west,font=\small] at (0.12,-1.73) {Occupied only in replica 1};
  \draw[second] (5.0,-1.73)--(5.75,-1.73);
  \node[anchor=west,font=\small] at (5.92,-1.73) {Occupied only in replica 2};
\end{tikzpicture}

\caption{Paired connectivity along different routes. Bond labels
indicate coupling signs, independently of the line styles that indicate
occupation by replica.}
\label{fig:paired-routes}
\end{figure}
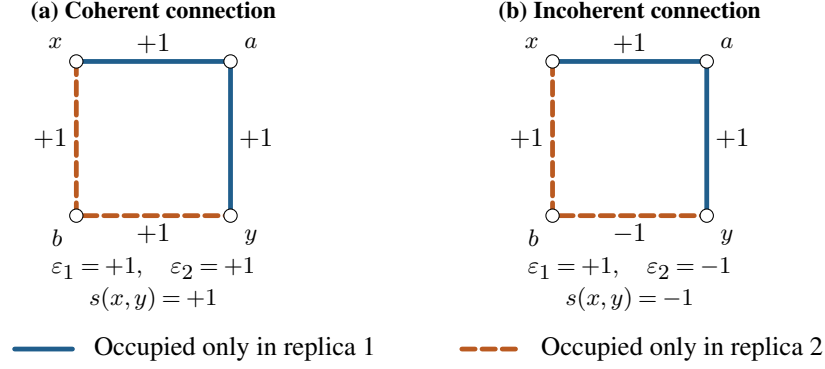

In both cases, \(x\) and \(y\) are connected in each replica, although
\(\eta^{(1)}\cap\eta^{(2)}=\varnothing\), so they are not connected by
a path of commonly occupied bonds. Their ordinary FK clusters are
\[
 C^{(1)}(x)=\{x,a,y\},\qquad C^{(2)}(x)=\{x,b,y\},
\]
and their paired class is the vertex-set intersection
\[
 \overline C(x)=C^{(1)}(x)\cap C^{(2)}(x)=\{x,y\}.
\]
In particular, the paired class is not itself nearest-neighbor connected.

In case~(a), all couplings are positive. The two route parities agree,
so \(\{x,y\}\) is one coherent companion and the other is empty.
In case~(b), only \(J_{by}\) is negative. The route parities are
opposite, so the two coherent companions are \(\{x\}\) and \(\{y\}\).
Each replica's occupied graph is a forest and is therefore unfrustrated,
even though their union is a frustrated circuit in case~(b).

For any compatible pair of spin assignments, define
\(q_z=\omega_z^{(1)}\omega_z^{(2)}\). Compatibility along the two routes gives
\[
 q_y=\varepsilon_1(x,y)\varepsilon_2(x,y)q_x=s(x,y)q_x.
\]
Thus the endpoint overlaps are equal in case~(a) and opposite in
case~(b). Coherence also occurs when both route parities are negative:
for example, if \(J_{xa}=J_{xb}=-1\) and \(J_{ay}=J_{by}=+1\), then
\(\varepsilon_1(x,y)=\varepsilon_2(x,y)=-1\) and \(s(x,y)=+1\).

\subsection{pc-FK clusters}
\label{subsec:pc-clusters}

A pc--FK cluster is a maximal set of vertices \(C\subseteq V\) such that
\[
   s(x,y)=+1
   \qquad\text{for all }x,y\in C.
\]
The definition is well posed because pc--FK connectivity is transitive:
if \(s(x,y)=+1\) and \(s(y,z)=+1\), then \(s(x,z)=+1\). Hence \(V\)
is partitioned into pc--FK clusters, some possibly with one vertex.

\begin{lemma}
\label{Lemma1tt}
For each pc--FK cluster \(C\) there is a possibly empty companion
set \(\widetilde C\) such that each pair of vertices
\(x\in C\) and \(y\in\widetilde C\) is pi--FK-connected.
When nonempty, \(\widetilde C\) is a pc--FK cluster.
\end{lemma}

\begin{proof}
Let \(C\) be a pc--FK cluster and fix \(x\in C\). Define
\[
   \widetilde C:=\{y\in V:\ s(x,y)=-1\}.
\]
The definition does not depend on the choice of
\(x\in C\): if \(x'\in C\), then \(s(x,x')=+1\); hence, for any
\(y\) with \(s(x,y)=-1\),
\[
   s(x',y)=s(x',x)s(x,y)=(+1)(-1)=-1.
\]
Thus the same set \(\widetilde C\) is obtained from any \(x\in C\).

If \(\widetilde C=\varnothing\), the assertion is immediate.
We therefore assume that \(\widetilde C\) is nonempty.

Now take \(y,z\in\widetilde C\). Then
\[
   s(y,z)=s(y,x)s(x,z)=(-1)(-1)=+1,
\]
so \(y\) and \(z\) are pc--FK-connected. Hence \(\widetilde C\) is
pc--FK-connected.

Finally, if \(w\) is pc--FK-connected to some \(y\in\widetilde C\), then
\[
   s(x,w)=s(x,y)s(y,w)=(-1)(+1)=-1,
\]
so \(w\in\widetilde C\). Thus \(\widetilde C\) is maximal and therefore
is a pc--FK cluster. By construction, every \(x\in C\) and
\(y\in\widetilde C\) satisfy \(s(x,y)=-1\).
\end{proof}

Notice that the pc--FK clusters, and even the union of two companion
pc--FK clusters, are generally not nearest-neighbor connected clusters.
They are long-range vertex clusters defined by paired connectivity in the
two replicas. The following associated notion will be relevant, especially
when the cluster is infinite.

\begin{definition}
The union
\[
   \overline C=C\cup\widetilde C
\]
of a pc--FK cluster \(C\) and its companion \(\widetilde C\) is called a
paired FK cluster, or p--FK cluster. Here ``p--FK'' stands for ``paired
FK'' and is not related to the edge-occupation parameter
\(p=1-e^{-2\beta}\).
\end{definition}

\begin{lemma}
\label{lem:pairedAsIntersectionOfFKClusters}
Let the FK bond configurations in the two replicas be
\(\eta^{(1)}\) and \(\eta^{(2)}\). A nonempty vertex set
\(\overline C\subseteq V\) is a p--FK cluster if and only if there exist
ordinary FK clusters \(C^{(1)}(\eta^{(1)})\) and
\(C^{(2)}(\eta^{(2)})\) in the two replicas such that
\[
      \overline C
      =
      C^{(1)}(\eta^{(1)})\cap C^{(2)}(\eta^{(2)}).
\]
\end{lemma}

\begin{proof}
First suppose that \(\overline C=C\cup\widetilde C\) is a p--FK cluster,
where \(C\) is a pc--FK cluster and \(\widetilde C\) is its pi--companion.
Fix \(x\in C\). For \(r=1,2\), let
\[
      C^{(r)}(\eta^{(r)})
      :=
      \{y\in V:\ x \stackrel{\eta^{(r)}}{\longleftrightarrow} y\}
\]
be the ordinary FK cluster of \(x\) in replica \(r\).

Every \(y\in\overline C\) is either pc--FK-connected or pi--FK-connected
to \(x\), and hence is FK-connected to \(x\) in both replicas. Therefore
\[
   \overline C
   \subseteq
   C^{(1)}(\eta^{(1)})\cap C^{(2)}(\eta^{(2)}).
\]
Conversely, if
\[
   y\in C^{(1)}(\eta^{(1)})\cap C^{(2)}(\eta^{(2)}),
\]
then \(x\) and \(y\) are FK-connected in both replicas, so
\(s(x,y)=+1\) or \(s(x,y)=-1\). In the first case \(y\in C\), by
maximality of \(C\); in the second case \(y\in\widetilde C\), by the
definition of the companion. Thus
\[
   C^{(1)}(\eta^{(1)})\cap C^{(2)}(\eta^{(2)})
   \subseteq
   \overline C.
\]
This proves the first implication.

Conversely, suppose
\[
   \overline C
   =
   C^{(1)}(\eta^{(1)})\cap C^{(2)}(\eta^{(2)})
\]
is nonempty, and fix \(x\in\overline C\). Let \(C\) be the maximal
pc--FK cluster containing \(x\), and let \(\widetilde C\) be its companion
from Lemma~\ref{Lemma1tt}. If \(y\in\overline C\), then \(x\) and \(y\)
are FK-connected in both replicas, so \(s(x,y)=+1\) or \(s(x,y)=-1\);
hence \(y\in C\cup\widetilde C\). Conversely, every vertex in
\(C\cup\widetilde C\) is FK-connected to \(x\) in both replicas, and
therefore lies in
\(C^{(1)}(\eta^{(1)})\cap C^{(2)}(\eta^{(2)})=\overline C\).
Thus
\(
   \overline C=C\cup\widetilde C,
\)
so \(\overline C\) is a p--FK cluster.
\end{proof}

It follows that two distinct p--FK clusters must arise from distinct
ordinary FK clusters in at least one replica.

For future use, we record the following deterministic consequence of
one-replica uniqueness.

\begin{lemma}
\label{Lemma5}
Suppose that each of the two FK configurations
\(\eta^{(1)}\) and \(\eta^{(2)}\) has at most one infinite ordinary FK
cluster. Then the paired configuration \(\eta^{\otimes 2}\) has at most
one infinite p--FK cluster.

Consequently, in \(\mathbb Z^d\), for joint weak limits of the two-replica
periodic-boundary-condition FK/ES measures, the number of infinite
p--FK clusters is at most one almost surely whenever
the one-replica \(0/1\)-infinite-cluster property holds in each replica.
\end{lemma}

\begin{proof}
By Lemma~\ref{lem:pairedAsIntersectionOfFKClusters}, every p--FK cluster
is the intersection of an ordinary FK cluster in replica \(1\) with an
ordinary FK cluster in replica \(2\). If such an intersection is infinite,
then both ordinary FK clusters involved must be infinite. Since each
replica has at most one infinite ordinary FK cluster, there is at most one
possible intersection of two infinite ordinary FK clusters. Therefore
there is at most one infinite p--FK cluster.

The final statement follows by applying the one-replica uniqueness result
to each replica, for instance in the periodic-boundary-condition limits
covered by~\cite{gandolfi1992uniqueness}.
\end{proof}

\medskip
\subsection{Paired coherent (pc) percolation}\label{subsec:pc-percolation}

Following Section~\ref{SfPerc}, we now define the second version of signed
FK percolation, which will be the basis of our analysis of replica
overlap. We work with two independent replicas with the same quenched
couplings \(J\). Thus, for fixed \(J\), we consider a product FK state
\[
   \nu^{\otimes 2}_{J,p}
   =
   \nu^{(1)}_{J,p}\otimes \nu^{(2)}_{J,p},
\]
and we restrict to the class in which each one-replica FK configuration has
at most one infinite FK cluster. By Lemma~\ref{Lemma5}, this implies that
there is at most one infinite p--FK cluster. We then look for ways of
distinguishing the two paired coherent companions of this infinite p--FK
cluster at infinity.
The geometric definitions below also apply to mixtures of product FK
states, as used in Corollary~\ref{cor:periodic_pc_overlap}.

Let \((\Lambda_n)_{n\ge1}\) be a fixed van~Hove sequence as in
Section~\ref{SfPerc}. If \(\overline C\) is a p--FK cluster and
\(
   \overline C=C^A\cup C^B
\)
is its decomposition into the two pc--FK companions, define
as before
\[
   d_{\Gamma}(C^A)
   :=
   \lim_{n\to\infty}
   \frac{|C^A\cap \Lambda_n|}{|\Lambda_n|},
\]
and analogously for \(C^B\), if the limits exist.

Given a two-replica FK configuration
\[
   \eta^{\otimes 2}=(\eta^{(1)},\eta^{(2)})
   \in \mathcal H^{(1)}\times\mathcal H^{(2)}
\]
with \(N^p_\infty(\eta^{\otimes 2})\le1\), let
\(C_\infty^A(\eta^{\otimes 2})\) and
\(C_\infty^B(\eta^{\otimes 2})\) denote the two pc--FK companions of the
unique infinite p--FK cluster, when it exists.
\begin{definition}[Paired coherent FK percolation]
\label{def:subseq-pc-density-imbalance}
Fix a deterministic van~Hove sequence \(\Gamma=(\Lambda_n)_{n\ge1}\)
for the fixed-$J$ argument.

Let \(\nu^{\otimes 2}_{J,p}\) be a two-replica FK state satisfying
\(N^p_\infty\le1\) almost surely. When there is no infinite p--FK
cluster, both companion sets and their densities are taken to be empty
and zero, respectively. Assume that the limits
$d_{\Gamma}(C^A), d_{\Gamma}(C^B)$ exist \(\nu^{\otimes 2}_{J,p}\)-almost surely, hence also in 
$L^1$ by bounded convergence. 

We say that \emph{pc--FK percolation occurs along \(\Gamma\)} if
\[
   \nu^{\otimes 2}_{J,p}(d_{\Gamma}(C^A) \neq  d_{\Gamma}(C^B))>0 .
\]
\end{definition}

When pc--FK-percolation occurs, the density distinguishes the two coherent
companions. For the fixed van~Hove sequence \((\Lambda_n)_{n\ge1}\), define
the two-replica density selector
\[
   \mathcal S^{(2)}_D:\mathcal H^{(1)}\times\mathcal H^{(2)}
   \longrightarrow \mathcal P(V)
\]
by
\[
   \mathcal S^{(2)}_D(\eta^{\otimes 2})
   :=
   \begin{cases}
     C_\infty^A(\eta^{\otimes 2}),
       & \text{if }
       d\bigl(C_\infty^A(\eta^{\otimes 2})\bigr)
       >
       d\bigl(C_\infty^B(\eta^{\otimes 2})\bigr),\\[3pt]
     C_\infty^B(\eta^{\otimes 2}),
       & \text{if }
       d\bigl(C_\infty^B(\eta^{\otimes 2})\bigr)
       >
       d\bigl(C_\infty^A(\eta^{\otimes 2})\bigr),\\[3pt]
     \varnothing,
       & \text{if limits do not exist or are equal.}
   \end{cases}
\]
If \(N^p_\infty(\eta^{\otimes 2})=0\), we also set
\(
   \mathcal S^{(2)}_D(\eta^{\otimes 2})=\varnothing .
\)

The analogue of Lemma~\ref{lem:tail-pql-density-selector1} is the following.

\begin{lemma}[Stability of the two-replica density selector under finite modifications]
\label{lem:tail-pql-density-selector3}
Let \(\pi^2\in\mathcal G^{\mathrm{ES},2}_{J,\beta}\) be a full two-replica ES
measure whose two one-replica FK marginals satisfy the \(0/1\)-infinite-cluster property, or
let \(\nu^2\) be a two-replica FK measure whose two one-replica marginals
are FK Gibbs states satisfying the \(0/1\)-infinite-cluster property. This includes, in particular, product measures
\(\pi^{(1)}\otimes\pi^{(2)}\) and \(\nu^{(1)}\otimes\nu^{(2)}\) with
one-replica marginals in the corresponding \(\le1\) classes. Assume also that
\(N^p_\infty\le1\) almost surely.

For a van~Hove sequence \(\Gamma=(\Lambda_n)_{n\ge1}\), the selector
\(\mathcal S^{(2)}_{D,\Gamma}\) satisfies the following property. For every finite
\(\Lambda\Subset V\) and for almost every
\[
   \eta^{\otimes 2}
   =
   (\eta^{(1)},\eta^{(2)})
   \in \mathcal H^{(1)}\times \mathcal H^{(2)},
\]
there exists a finite set
\[
   M=M(\eta^{\otimes 2},\Lambda)\subset V
\]
such that, for every
\[
   \xi^{\otimes 2}
   =
   (\xi^{(1)},\xi^{(2)})
   \in \mathcal H^{(1)}\times \mathcal H^{(2)}
\]
with both \(\xi^{(i)}\) compatible with \(J\) and satisfying
\[
   \eta^{(i)}_{E\setminus \overline E_\Lambda}
   =
   \xi^{(i)}_{E\setminus \overline E_\Lambda},
   \qquad i=1,2,
\]
one has
\[
   \mathcal S^{(2)}_D(\eta^{\otimes 2})\cap M^c
   =
   \mathcal S^{(2)}_D(\xi^{\otimes 2})\cap M^c .
\]
\end{lemma}

\begin{proof}
We give the proof for the FK measure; the ES case follows by applying the
same argument to the FK marginal.

Fix a finite \(\Lambda\Subset V\), and work on the common
full-probability event supplied by Lemma~\ref{UniqueOutside} in each
replica. In particular, deleting \(\overline E_\Lambda\) leaves at most
one infinite exterior component in each replica.
Let
\(
   \eta^{\otimes 2}=(\eta^{(1)},\eta^{(2)})
\)
be a configuration in this full-probability event. For \(i=1,2\), let
\(C_\infty^{(i)}(\eta^{(i)})\) denote the unique infinite ordinary FK cluster
in replica \(i\), if it exists.

Only finitely many exterior clusters meet \(\partial\Lambda\) in
either replica. Choose \(M=M(\eta^{\otimes2},\Lambda)\) to contain
\(\Lambda\cup\partial\Lambda\) and all finite such exterior clusters
in both replicas. By Lemma~\ref{UniqueOutside}, any vertex of an ordinary
infinite FK cluster outside \(M\) lies in the unique infinite component
of the corresponding exterior graph on \(E\setminus\overline E_\Lambda\).
Consequently local modifications inside \(\overline E_\Lambda\) do not
change the ordinary infinite FK clusters outside \(M\). More precisely,
if
\(
   \eta^{(i)}_{E\setminus \overline E_\Lambda}
   =
   \xi^{(i)}_{E\setminus \overline E_\Lambda},
   \qquad i=1,2,
\)
then, for each \(i=1,2\),
\[
   C_\infty^{(i)}(\eta^{(i)})\cap M^c
   =
   C_\infty^{(i)}(\xi^{(i)})\cap M^c,
\]
with the convention that both sides are empty when no infinite cluster exists
in replica \(i\).

Using Lemma~\ref{lem:pairedAsIntersectionOfFKClusters}, the infinite p--FK
cluster is the intersection
\[
   C_\infty^{(1)}(\eta^{(1)})
   \cap
   C_\infty^{(2)}(\eta^{(2)}),
\]
when this intersection is infinite. Hence the infinite p--FK cluster, if it
exists, is unchanged outside \(M\) by any modification inside
\(\overline E_\Lambda\).

It remains to check that the decomposition into the two pc--FK companions is
also unchanged outside \(M\). Suppose first that an infinite p--FK cluster
exists. Choose a vertex
\[
   z\in C_\infty^{(1)}(\eta^{(1)})\cap
        C_\infty^{(2)}(\eta^{(2)})\cap M^c .
\]
For any other vertex \(y\) in the infinite p--FK cluster outside
\(M\), both \(z\) and \(y\) lie, in each replica, in the unique infinite
exterior component supplied by Lemma~\ref{UniqueOutside}. They can
therefore be joined by an FK-open path using only edges in
\(E\setminus\overline E_\Lambda\). The path may pass through
\(M\setminus\Lambda\); connectedness outside \(M\) is not required.
Since \(\eta^{\otimes2}\) and \(\xi^{\otimes2}\) agree on these edges,
the two path parities, and hence \(s(z,y)\), are unchanged. This preserves
the pc/pi classification of vertices outside \(M\).

Consequently, up to a possible interchange of the labels \(A\) and \(B\), the
two pc--FK companions agree outside \(M\) for \(\eta^{\otimes 2}\) and
\(\xi^{\otimes 2}\). After matching the labels outside \(M\), changing
finitely many vertices preserves both the existence and the values of
the density limits along the fixed van~Hove sequence. Hence the density
selector either remains empty in both configurations, or selects the
same companion outside \(M\). Therefore
\[
   \mathcal S^{(2)}_D(\eta^{\otimes 2})\cap M^c
   =
   \mathcal S^{(2)}_D(\xi^{\otimes 2})\cap M^c .
\]

If no infinite p--FK cluster exists, the same stability argument shows that
no finite modification inside \(\overline E_\Lambda\) can create one outside
\(M\). Thus both selectors are empty outside \(M\). This completes the proof.
\end{proof}

We will use only the stability under finite modifications established in
Lemma~\ref{lem:tail-pql-density-selector3}.

If the relevant two-replica FK distribution is translation invariant, then
the density imbalance can often be read in terms of the probability that the
origin belongs to one or the other coherent companion. In such a setting,
one can equivalently view pc--FK-percolation through an origin-based signed
count: vertices pc--FK-connected to the origin contribute \(+1\), while
vertices pi--FK-connected to the origin contribute \(-1\). A nonzero
asymptotic imbalance of this signed count corresponds to density
distinguishability of the two paired coherent companions.

\section{Replica overlap and paired coherent FK percolation
in \texorpdfstring{\(\mathbb{Z}^d\)}{the integer lattice}}

\subsection{Decomposition of two-replica ES states according to the
            pc--FK companion bias}
\label{subsec:decomposition2}

Throughout this subsection we fix a coupling realization
\(J\in\{\pm1\}^{E}\) and an inverse temperature \(\beta\)
(or edge--occupation parameter \(p=1-e^{-2\beta}\)).
We work here with full two-replica ES Gibbs states
\(
   \pi^{2}\in
   \mathcal G^{\mathrm{ES},2}_{J,\beta}
\)
because the auxiliary agreement-biased constructions below need not be
products of two one-replica ES states. The independent-replica overlap theorem
at the end of the section will, however, be stated for products of extremal
one-replica states; see Theorem~\ref{thm:product_extremal_pc_overlap}.
We assume, as in the preceding section, that under the FK
marginal there is almost surely at most one infinite p--FK cluster and that the
one-replica FK marginals satisfy the \(0/1\)-infinite-cluster property.

Define the relative spin flip on the two-replica ES space by
\[
\begin{aligned}
 F_2(\omega^{(1)},\omega^{(2)},\eta^{(1)},\eta^{(2)})
 :=
 (\omega^{(1)},-\omega^{(2)},\eta^{(1)},\eta^{(2)}).
\end{aligned}
\]
A two-replica ES measure has relative spin-flip symmetry if
\((F_2)_*\pi^2=\pi^2\). This transformation leaves both FK configurations
and their density selector unchanged and reverses every site overlap.
Thus the selected or unselected status of the FK geometry imposes no
relative spin-flip symmetry on the ES measure.

Recall that, under this assumption, when an infinite p--FK cluster exists it
decomposes into two coherent companion pc--FK clusters
\(
   C_\infty^A(\eta^{\otimes 2}),
   \hskip 0.2cm
   C_\infty^B(\eta^{\otimes 2}),
\)
where
\(
   \eta^{\otimes 2}=(\eta^{(1)},\eta^{(2)}).
\)
When the density selector from Section~\ref{subsec:pc-percolation} is nonempty, we label the
denser selected companion by
\(
   C_\infty^A(\eta^{\otimes 2})
   :=
   \mathcal S^{(2)}_D(\eta^{\otimes 2}),
\)
 we let \(C_\infty^B(\eta^{\otimes 2})\) be the other coherent companion,
 and we write
\(
   \mathcal C_\infty(\eta^{\otimes 2})
   :=
   C_\infty^A(\eta^{\otimes 2})
   \cup
   C_\infty^B(\eta^{\otimes 2}).
\)

On the event where the selector is nonempty, define
\begin{eqnarray} \label{H}
   H(\omega^{\otimes 2},\eta^{\otimes 2})
   :=
   \mathbf 1_{\{
      \omega^{(1)}_x=\omega^{(2)}_x
      \text{ for all }x\in C_\infty^A(\eta^{\otimes 2})
   \}} .
\end{eqnarray}
Thus \(H=1\) means that the two replicas agree on the selected coherent
companion, while \(H=0\) means that they disagree on it. By the pc/pi
structure, if the replicas agree on one coherent companion of the infinite
p--FK cluster, then they disagree on the other one, and conversely.

Set
\[
   E_{\mathrm{null}}
   :=
   \{N_\infty^p=0\},
\]
\[
   E_{\mathrm{unsel}}
   :=
   \{N_\infty^p=1,\ \mathcal S^{(2)}_D(\eta^{\otimes 2})=\varnothing\},
\]
and
\[
   E_{\mathrm{sel}}
   :=
   \{N_\infty^p=1,\ \mathcal S^{(2)}_D(\eta^{\otimes 2})\neq\varnothing\}.
\]
On \(E_{\mathrm{unsel}}\), the two companion densities are equal
whenever they both exist. Without a density-existence hypothesis,
this event also includes failure of the density limits.
On \(E_{\mathrm{sel}}\), define
\[
   E_{\mathrm{agr}}
   :=
   E_{\mathrm{sel}}\cap\{H=1\},
   \qquad
   E_{\mathrm{dis}}
   :=
   E_{\mathrm{sel}}\cap\{H=0\}.
\]
The relative spin flip \(F_2\) exchanges \(E_{\mathrm{agr}}\) and
\(E_{\mathrm{dis}}\) and leaves \(E_{\mathrm{unsel}}\) unchanged.

\begin{theorem}[Canonical two-replica decomposition by the pc--FK companion]
\label{thm:decomposition_replica}
Let
\[
   \pi^{2}\in
   \mathcal G^{\mathrm{ES},2}_{J,\beta}
   \qquad\text{with}\qquad
   \pi^2(N_\infty^p\le 1)=1 .
\]
Then
\[
   E_{\mathrm{null}},\quad
   E_{\mathrm{unsel}},\quad
   E_{\mathrm{agr}},\quad
   E_{\mathrm{dis}}
\]
are tail events which form a measurable partition of
\(\{N_\infty^p\le1\}\). Consequently,
\(\pi^2\) admits the canonical decomposition
\[
   \pi^2
   =
   \alpha_{\mathrm{null}}\pi^{2}_{\mathrm{null}}
   +
   \alpha_{\mathrm{unsel}}\pi^{2}_{\mathrm{unsel}}
   +
   \alpha_{\mathrm{agr}}\pi^{2}_{\mathrm{agr}}
   +
   \alpha_{\mathrm{dis}}\pi^{2}_{\mathrm{dis}},
\]
where
\[
   \alpha_{\mathrm{null}}=\pi^2(E_{\mathrm{null}}),
   \quad
   \alpha_{\mathrm{unsel}}=\pi^2(E_{\mathrm{unsel}}),
   \quad
   \alpha_{\mathrm{agr}}=\pi^2(E_{\mathrm{agr}}),
   \quad
   \alpha_{\mathrm{dis}}=\pi^2(E_{\mathrm{dis}}),
\]
and, whenever the corresponding coefficient is nonzero,
\[
   \pi^{2}_{\mathrm{null}}
   =
   \pi^2(\,\cdot\,|E_{\mathrm{null}}),
   \qquad
   \pi^{2}_{\mathrm{unsel}}
   =
   \pi^2(\,\cdot\,|E_{\mathrm{unsel}}),
\]
\[
   \pi^{2}_{\mathrm{agr}}
   =
   \pi^2(\,\cdot\,|E_{\mathrm{agr}}),
   \qquad
   \pi^{2}_{\mathrm{dis}}
   =
   \pi^2(\,\cdot\,|E_{\mathrm{dis}}).
\]
Each nonzero component is again a two-replica Edwards--Sokal Gibbs state
for the same specification. The decomposition is unique, with the convention
that zero-weight components are omitted.
\end{theorem}

\begin{proof}
Work on the common full-measure set supplied by
Lemma~\ref{UniqueOutside} in each replica. As shown in the proof of
Lemma~\ref{lem:tail-pql-density-selector3}, finite compatible bond
modifications change each ordinary infinite FK cluster only within a
finite vertex set. The same is therefore true of their intersection.
Its infinitude, and the nonemptiness of the density selector, are
unchanged. Hence \(E_{\mathrm{null}}\), \(E_{\mathrm{unsel}}\), and
\(E_{\mathrm{sel}}\) are tail events, modulo \(\pi^2\)-null sets.

We next show that \(H\) is tail-measurable on \(E_{\mathrm{sel}}\).
Fix a finite \(\Lambda\Subset V\), and let
\[
   \mathcal A^{(2)}_{\Lambda^c}
   :=
   \sigma\!\left(
      \omega^{(1)}_{V\setminus\Lambda},
      \omega^{(2)}_{V\setminus\Lambda},
      \eta^{(1)}_{E\setminus \overline E_\Lambda},
      \eta^{(2)}_{E\setminus \overline E_\Lambda}
   \right)
\]
be the exterior two-replica ES sigma-field.

By Lemma~\ref{lem:tail-pql-density-selector3}, for
\(\pi^2\)-a.e. \(\eta^{\otimes 2}\) there is a finite set
\(M=M(\eta^{\otimes 2},\Lambda)\) such that every
\(\xi^{\otimes 2}\) compatible with \(J\) in both replicas and agreeing
with \(\eta^{\otimes 2}\) outside \(\overline E_\Lambda\) satisfies
\[
   \mathcal S^{(2)}_D(\xi^{\otimes 2})\cap M^c
   =
   \mathcal S^{(2)}_D(\eta^{\otimes 2})\cap M^c .
\]
On \(E_{\mathrm{sel}}\), the selected companion
\(
   C_\infty^A(\eta^{\otimes 2})
   =
   \mathcal S^{(2)}_D(\eta^{\otimes 2})
\)
is infinite. Hence we may choose
\(
   z\in C_\infty^A(\eta^{\otimes 2})\cap M^c .
\)
For every ES-compatible two-replica configuration, the overlap
\(
   \omega^{(1)}_x\omega^{(2)}_x
\)
is constant on each pc--FK cluster. Indeed, pc--FK connectivity means that
the two FK paths have the same parity, and therefore preserve the equality
or inequality of the two replica spins. Thus
\[
   H(\omega^{\otimes 2},\eta^{\otimes 2})=1
   \quad\Longleftrightarrow\quad
   \omega^{(1)}_z\omega^{(2)}_z=+1 .
\]
Since \(z\in M^c\), this value is determined by variables outside the finite
set \(M\), hence outside \(\Lambda\). Because \(\Lambda\) was arbitrary,
\(H\) is tail-measurable. Therefore \(E_{\mathrm{agr}}\) and
\(E_{\mathrm{dis}}\) are tail events.

The four events
\[
   E_{\mathrm{null}},\quad
   E_{\mathrm{unsel}},\quad
   E_{\mathrm{agr}},\quad
   E_{\mathrm{dis}}
\]
are disjoint and exhaust \(\{N_\infty^p\le1\}\). Since
\(\pi^2(N_\infty^p\le1)=1\), they form a measurable partition of
the support of \(\pi^2\), and the displayed decomposition is just
the decomposition of \(\pi^2\) with respect to this partition.
This proves existence and uniqueness.

It remains to prove that the nonzero conditional measures are still
two-replica ES Gibbs states. Let \(E\) be one of the four tail events with
\(\pi^2(E)>0\), and set
\(
   \pi^{2}_E
   :=
   \pi^2(\,\cdot\,|E).
\)
Since \(E\) is a tail event, \(\mathbf 1_E\) is
\(\mathcal A^{(2)}_{\Lambda^c}\)-measurable up to null sets for every finite
\(\Lambda\).

Let \(f\) be a bounded local function depending on the variables in
\(\Lambda\), and let \(g\) be a bounded
\(\mathcal A^{(2)}_{\Lambda^c}\)-measurable function. Since
\(\pi^2\) satisfies the two-replica ES DLR equations,
\[
   \int f g\,d\pi^2
   =
   \int g\,
      \pi^{2}_{J,\beta,\Lambda,
        \omega^{(1)}_{\partial\Lambda},
        \omega^{(2)}_{\partial\Lambda}}(f)
      \,d\pi^2.
\]
Multiplying by \(\mathbf 1_E\), which is exterior-measurable, gives
\[
   \int_E f g\,d\pi^2
   =
   \int_E g\,
      \pi^{2}_{J,\beta,\Lambda,
        \omega^{(1)}_{\partial\Lambda},
        \omega^{(2)}_{\partial\Lambda}}(f)
      \,d\pi^2.
\]
Dividing by \(\pi^2(E)\), we obtain
\[
   \int f g\,d\pi^{2}_E
   =
   \int g\,
      \pi^{2}_{J,\beta,\Lambda,
        \omega^{(1)}_{\partial\Lambda},
        \omega^{(2)}_{\partial\Lambda}}(f)
      \,d\pi^{2}_E .
\]
Since this holds for every bounded exterior-measurable \(g\), the conditional
law of the variables in \(\Lambda\) under \(\pi^{2}_E\), given the
exterior, is the same two-replica ES kernel. Thus
\(
   \pi^{2}_E
   \in
   \mathcal G^{\mathrm{ES},2}_{J,\beta}.
\)
Applying this to each of the four events proves the theorem.
\end{proof}

\begin{corollary}
\label{cor:extremal_pc_sector}
Let
\(
   \pi^{2}\in
   \mathcal G^{\mathrm{ES},2}_{J,\beta}
\)
be extremal, and suppose
\(
   \pi^2(N_\infty^p\le1)=1 .
\)
Then exactly one of the four events
\[
   E_{\mathrm{null}},\quad
   E_{\mathrm{unsel}},\quad
   E_{\mathrm{agr}},\quad
   E_{\mathrm{dis}}
\]
has \(\pi^2\)-probability one. 

In particular, if
\(
   \pi^2(E_{\mathrm{sel}})=1,
\)
then either
\(
   \pi^2(H=1)=1\) or
   \(
   \pi^2(H=0)=1 .
\)
Equivalently, in the density-distinguishable paired coherent tail event, an
extremal two-replica ES state cannot have a genuine bias
\(0<q<1\) between agreement and disagreement on the selected companion.
\end{corollary}

\begin{proof}
The four events are tail events by
Theorem~\ref{thm:decomposition_replica}. Extremal ES Gibbs states are
tail-trivial. Hence each of the four events has probability \(0\) or \(1\),
and since they form a partition of a full-probability set, exactly one has
probability \(1\).
\end{proof}

We next record the \(q\)-biased construction. This construction is useful for
producing two-replica Gibbs states with a prescribed agreement bias on the
selected coherent companion. For \(0<q<1\), the \(q\)-biased state is a
convex combination of the two oriented states, as in the one-replica
construction.

\begin{proposition}[DLR property of two-replica \(q\)-biased lifts]
\label{lem:DLR_replica_q}
Let
\(
   \nu^{\otimes 2}_{J,p}
   \in
   \mathcal G^{\mathrm{FK},\otimes 2}_{J,p}
\)
be an independent product two-replica FK Gibbs state such that
\[
   \nu^{\otimes 2}_{J,p}(N_\infty^p=1)=1,
   \qquad
   \nu^{\otimes 2}_{J,p}(\mathcal S^{(2)}_D\neq\varnothing)=1 .
\]
Let \(\pi^{2}_{\nu,1/2}\) be the ES lift invariant under \(F_2\) obtained by
coloring the FK clusters in each replica with independent fair compatible
spin assignments. Define \(H\) as in \eqref{H} and, for \(q\in[0,1]\), set
\(
   h_q
   :=
   2q\,H+2(1-q)(1-H),
\) and
\(
   \pi^{2}_{\nu,q}
   :=
   h_q\,\pi^{2}_{\nu,1/2}.
\)

Then
\(
   \pi^{2}_{\nu,q}
   \in
   \mathcal G^{\mathrm{ES},2}_{J,\beta} \)
   for every \(q\in[0,1],\)
and
\(
   \Pi^{\otimes 2}_{\mathrm{FK}}
   \pi^{2}_{\nu,q}
   =
   \nu^{\otimes 2}_{J,p}.
\)
Moreover,
\[
   \pi^{2}_{\nu,q}
   =
   q\,\pi^{2}_{\nu,1}
   +
   (1-q)\,\pi^{2}_{\nu,0}.
\]
\end{proposition}

\begin{proof}
By the product two-replica FK--ES lifting result, the symmetric cluster-coloring
construction gives the independent product ES Gibbs state
\[
   \pi^{2}_{\nu,1/2}
   \in
   \mathcal G^{\mathrm{ES},\otimes 2}_{J,\beta}
   \subset
   \mathcal G^{\mathrm{ES},2}_{J,\beta}
\]
with FK marginal \(\nu^{\otimes 2}_{J,p}\). The tilted measures
\(\pi^2_{\nu,q}\), however, will in general belong only to the full class
\(\mathcal G^{\mathrm{ES},2}_{J,\beta}\), not to the product subclass.

The symmetric lift is invariant under flipping all spins in the second
replica. This transformation preserves the FK variables and exchanges
agreement with disagreement on every pc--FK companion. Therefore
\(
   \int H\,d\pi^{2}_{\nu,1/2}=\frac12,
\)
and so \(\pi^{2}_{\nu,q}\) is a probability measure.

We now verify the DLR equations. Fix a finite \(\Lambda\Subset V\), and let
\[
   \mathcal A^{(2)}_{\Lambda^c}
   =
   \sigma\!\left(
      \omega^{(1)}_{V\setminus\Lambda},
      \omega^{(2)}_{V\setminus\Lambda},
      \eta^{(1)}_{E\setminus \overline E_\Lambda},
      \eta^{(2)}_{E\setminus \overline E_\Lambda}
   \right).
\]
By Lemma~\ref{lem:tail-pql-density-selector3}, the selected coherent
companion is unchanged outside a finite set by modifications inside
\(\overline E_\Lambda\). Since the overlap
\(\omega^{(1)}_x\omega^{(2)}_x\) is constant on a pc--FK companion, the event
\[
   \{\omega^{(1)}=\omega^{(2)}
     \text{ on }C_\infty^A(\eta^{\otimes 2})\}
\]
is determined by the overlap at any exterior vertex of the selected
companion. Hence \(H\), and therefore \(h_q\), is
\(\mathcal A^{(2)}_{\Lambda^c}\)-measurable.

Let \(f\) be a bounded local function depending on the variables in
\(\Lambda\). Using the exterior-measurability of \(h_q\) and the DLR equation
for \(\pi^{2}_{\nu,1/2}\), we get
\[
\begin{aligned}
   \int f\,d\pi^{2}_{\nu,q}
   &=
   \int f h_q\,d\pi^{2}_{\nu,1/2}  =
   \mathbb E_{\pi^{2}_{\nu,1/2}}
   \Bigl[
      \mathbb E(fh_q\mid\mathcal A^{(2)}_{\Lambda^c})
   \Bigr]  \\
   &=
   \mathbb E_{\pi^{2}_{\nu,1/2}}
   \Bigl[
      h_q\,\mathbb E(f\mid\mathcal A^{(2)}_{\Lambda^c})
   \Bigr]  =
   \mathbb E_{\pi^{2}_{\nu,q}}
   \Bigl[
      \mathbb E(f\mid\mathcal A^{(2)}_{\Lambda^c})
   \Bigr].
\end{aligned}
\]
The inner conditional expectation is the two-replica ES kernel because
\(\pi^{2}_{\nu,1/2}\) is an ES Gibbs state. Hence
\(\pi^{2}_{\nu,q}\) satisfies the same DLR equations.

Finally, fair independent coloring in the two replicas gives
$\mathbb E_{\pi^2_{\nu,1/2}}[H\mid\eta^{\otimes2}]=1/2$, and hence
$\mathbb E_{\pi^2_{\nu,1/2}}[h_q\mid\eta^{\otimes2}]=1$.
Therefore the change of density does not change the FK marginal. Thus
\(
   \Pi^{\otimes 2}_{\mathrm{FK}}
   \pi^{2}_{\nu,q}
   =
   \nu^{\otimes 2}_{J,p}.
\)
The identity
\(
   \pi^{2}_{\nu,q}
   =
   q\,\pi^{2}_{\nu,1}
   +
   (1-q)\,\pi^{2}_{\nu,0}
\)
is immediate from
\(
   h_q=2qH+2(1-q)(1-H).
\)
\end{proof}

\begin{remark}
The preceding proposition gives \(q\)-biased ES lifts over a fixed
two-replica FK marginal supported on \(E_{\mathrm{sel}}\).
These are auxiliary full two-replica Gibbs
states. For \(0<q<1\), they need not represent independent replicas and need
not belong to the product subclass
\(\mathcal G^{\mathrm{ES},\otimes 2}_{J,\beta}\). However, by
Proposition~\ref{prop:full-two-replica-extremals-factor}, the extremal full
two-replica ES states are products of extremal one-replica ES states.

The \(q\)-biased construction therefore shows that the paired coherent tail
events are genuine Gibbs tail events, while the physical independent-replica
overlap is read on the product extremal components. It does not give a unique
convex decomposition of an arbitrary ES state into \(q\)-biased states. The
canonical decomposition is the tail-event decomposition of
Theorem~\ref{thm:decomposition_replica}. In particular, for \(0<q<1\),
\[
   \pi^{2}_{\nu,q}
   =
   q\,\pi^{2}_{\nu,1}
   +
   (1-q)\,\pi^{2}_{\nu,0},
\]
so a \(q\)-biased state is already a mixture of the two oriented tail events.

Also, unless one additionally fixes the one-replica infinite-cluster
orientations, the FK projection is not necessarily globally
bijective on all two-replica ES tail events. The
agreement/disagreement tail decomposition associated with the paired coherent
companions is canonical.
\end{remark}

For the deterministic configurations \(J_e=+1\) and
\(\eta_e^{(1)}=\eta_e^{(2)}=1\) for every edge \(e\) of
\(\mathbb Z^d\), the paired FK cluster is all of \(\mathbb Z^d\).
Its two coherent companions are \(\mathbb Z^d\) and the empty set,
with densities \(1\) and \(0\) along every van~Hove sequence.

\subsection{FK characterization of mutual overlap}
\label{subsec:mutual-overlap}

Our next goal is to express the two-replica mutual overlap in terms of the
FK variables alone. For a fixed coupling realization \(J\), consider the
spatial two-replica overlap along a fixed van~Hove sequence
\((\Lambda_n)_{n\ge1}\):
\[
   Q_{\Lambda_n}(\omega^{\otimes 2})
   :=
   \frac{1}{|\Lambda_n|}
   \sum_{x\in\Lambda_n}
      \omega^{(1)}_x\omega^{(2)}_x .
\]
For products of extremal one-replica states, the limit of its expectation is
their mutual overlap. When both replicas are sampled independently from the
same extremal state, this limit is its Edwards--Anderson self-overlap;
see Corollary~\ref{cor:diagonal_pc_EA}.
For stationary joint periodic-boundary limits,
Corollary~\ref{cor:periodic_pc_overlap} identifies the expected absolute
value of the sampled limiting overlap with the mean companion-density
imbalance.

\begin{theorem}[Product-extremal two-replica overlap]
\label{thm:product_extremal_pc_overlap}
Fix a coupling realization \(J\in\{\pm1\}^E\), a finite inverse
temperature \(0\le\beta<\infty\), and \(p=1-e^{-2\beta}<1\). Let
\(
   \mu^{(1)}_{J,\beta},\ \mu^{(2)}_{J,\beta}
   \in
   \mathcal G^{\mathrm{SPIN}}_{J,\beta}
\)
be two extremal one-replica Gibbs states, and let
\(
   \pi^{(1)}_{J,\beta},\ \pi^{(2)}_{J,\beta}
   \in
   \mathcal G^{\mathrm{ES}}_{J,\beta}
\)
be their ES lifts. Assume that both lifts satisfy the \(0/1\)-infinite-cluster property.
Set
\[
   \mu^{\otimes 2}_{J,\beta}
   :=
   \mu^{(1)}_{J,\beta}\otimes\mu^{(2)}_{J,\beta},
   \qquad
   \pi^{\otimes 2}_{J,\beta}
   :=
   \pi^{(1)}_{J,\beta}\otimes\pi^{(2)}_{J,\beta},
\]
and let
\(
   \nu^{\otimes 2}_{J,p}
   :=
   \Pi^{\otimes 2}_{\mathrm{FK}}\pi^{\otimes 2}_{J,\beta}
\)
be the product FK marginal. Then \(N_\infty^p\le1\)
\(\nu^{\otimes 2}_{J,p}\)-a.s.

Fix a deterministic van~Hove sequence
\(
   \Gamma.
\)
Assume that the density limits required in
Definition~\ref{def:subseq-pc-density-imbalance} exist
\(\nu^{\otimes 2}_{J,p}\)-a.s. along \(\Gamma\). On the selected
tail event, write
\(
   C_\infty^A(\eta^{\otimes 2})
   =
   \mathcal S^{(2)}_{D,\Gamma}(\eta^{\otimes 2})
\)
for the selected pc--FK companion and \(C_\infty^B(\eta^{\otimes 2})\) for
the other one. Define
\[
   D_\Gamma(\eta^{\otimes 2})
   :=
   d_\Gamma(C_\infty^A(\eta^{\otimes 2}))
   -
   d_\Gamma(C_\infty^B(\eta^{\otimes 2})),
\]
and set \(D_\Gamma=0\) outside the selected tail event. Thus
\(D_\Gamma\ge0\), and pc--FK percolation along \(\Gamma\) occurs exactly when
\(
   \nu^{\otimes 2}_{J,p}(D_\Gamma>0)>0.
\)

Let
\[
   Q_{\Lambda_n}(\omega^{\otimes 2})
   :=
   \frac{1}{|\Lambda_n|}
   \sum_{x\in\Lambda_n}
      \omega^{(1)}_x\omega^{(2)}_x .
\] 
Then the limit
\begin{eqnarray} \label{qgamma}
   q_\Gamma(J,\beta)
   :=
   \lim_{n\to\infty}
   \mu^{\otimes 2}_{J,\beta}\bigl(Q_{\Lambda_n}\bigr)
\end{eqnarray}
exists. Moreover, there exists
\(
   \varepsilon_J\in\{+1,-1\}
\)
such that
\[
   q_\Gamma(J,\beta)
   =
   \varepsilon_J\,
   \mathbb E_{\nu^{\otimes 2}_{J,p}}
   \bigl[
      D_\Gamma(\eta^{\otimes 2})
   \bigr].
\]
Consequently,
\[
   |q_\Gamma(J,\beta)|>0
   \quad\Longleftrightarrow\quad
   \nu^{\otimes 2}_{J,p}(D_\Gamma>0)>0.
\]
Equivalently, the product-extremal two-replica state has nonzero
\(\Gamma\)-overlap if and only if pc--FK percolation occurs along
\(\Gamma\) in its product FK marginal.
\end{theorem}

\begin{proof}
By the one-replica spin--ES isomorphism, the ES lifts
\(\pi^{(1)}_{J,\beta}\) and \(\pi^{(2)}_{J,\beta}\) are extremal. Hence their
product is tail-trivial for the product two-replica tail sigma-field. The
one-replica \(0/1\)-infinite-cluster property in the two copies implies \(N_\infty^p\le1\) for
the paired FK configuration, by Lemma~\ref{Lemma5}.

Apply the two-replica tail-event decomposition from
Theorem~\ref{thm:decomposition_replica}. Since
\(\pi^{\otimes 2}_{J,\beta}\) is tail-trivial, exactly one of the corresponding
tail events has probability one. We first justify the zero limit in the
null and unselected cases. Given the two FK configurations, let
$I=C_\infty^{(1)}\cap C_\infty^{(2)}$, with $I=\varnothing$ when one
replica has no infinite FK cluster. Outside $I$, the conditional overlap
expectation is zero: at least one ordinary FK cluster is finite, its
compatible coloring is fair, and the two replicas are independent.
If $I$ is finite, therefore,
\[
 \left|\mathbb E_{\pi^{\otimes2}_{J,\beta}}
 [Q_{\Lambda_n}\mid\eta^{\otimes2}]\right|
 \le \frac{|I\cap\Lambda_n|}{|\Lambda_n|}\longrightarrow0.
\]
This includes a finite intersection of two infinite ordinary FK clusters.
If $I$ is infinite, choose a Borel labeling $C^0,C^1$ of its coherent
companions and a vertex $a\in C^0$. Compatibility gives a conditional
mean $t(\eta^{\otimes2})s(a,x)$ at $x\in I$, where
$t=\mathbb E[\omega_a^{(1)}\omega_a^{(2)}\mid\eta^{\otimes2}]$
and $|t|\le1$. Consequently the absolute conditional mean of
$Q_{\Lambda_n}$ is at most
\[
 \left|\frac{|C^0\cap\Lambda_n|-|C^1\cap\Lambda_n|}{|\Lambda_n|}\right|.
\]
On the unselected event the assumed existence of both companion densities
makes this expression tend to zero. Bounded convergence proves the claimed
zero limiting expectation in both cases. No pointwise vanishing at finite
$n$ is asserted.

On the selected tail event, the agreement event
\(
   \{\omega^{(1)}=\omega^{(2)}
      \text{ on }C_\infty^A(\eta^{\otimes 2})\}
\)
is a tail event. Hence it has probability either zero or one. If it has
probability one, then the replicas agree on \(C_\infty^A\) and disagree on
\(C_\infty^B\); if it has probability zero, the signs are reversed. Therefore
there is \(\varepsilon_J\in\{+1,-1\}\) such that, conditional on the FK
variables,
\[
   \mathbb E_{\pi^{\otimes 2}_{J,\beta}}
   \bigl[
      \omega^{(1)}_x\omega^{(2)}_x
      \,\big|\,
      \eta^{\otimes 2}
   \bigr]
   =
   \begin{cases}
      \varepsilon_J, & x\in C_\infty^A(\eta^{\otimes 2}),\\
      -\varepsilon_J, & x\in C_\infty^B(\eta^{\otimes 2}),\\
      0, & x\notin C_\infty^A(\eta^{\otimes 2})
              \cup C_\infty^B(\eta^{\otimes 2}),
   \end{cases}
\]
where the last case gives zero because at least one of the one-replica FK
clusters containing \(x\) is finite and is colored symmetrically.

Thus
\[
\begin{aligned}
   \mu^{\otimes 2}_{J,\beta}
   \bigl(Q_{\Lambda_n}\bigr)
   &=
   \varepsilon_J\,
   \mathbb E_{\nu^{\otimes 2}_{J,p}}
   \left[
      \frac{
         |C_\infty^A(\eta^{\otimes 2})\cap\Lambda_n|
         -
         |C_\infty^B(\eta^{\otimes 2})\cap\Lambda_n|
      }{
         |\Lambda_n|
      }
   \right],
\end{aligned}
\]
This identity is used when the selected tail event has probability one.
Letting \(n\to\infty\) and using bounded convergence gives
\[
   q_\Gamma(J,\beta)
   =
   \varepsilon_J\,
   \mathbb E_{\nu^{\otimes 2}_{J,p}}
   \bigl[
      D_\Gamma(\eta^{\otimes 2})
   \bigr]
\]
on the selected tail event, and \(q_\Gamma(J,\beta)=0\) otherwise.
Since \(D_\Gamma\ge0\), its expectation is positive if and only if
\(\nu^{\otimes 2}_{J,p}(D_\Gamma>0)>0\). This proves the equivalence.
\end{proof}

\begin{corollary}[Nonzero mutual overlap along some van Hove sequence]
\label{cor:spin_glass_some_vanhove}
Fix a coupling realization \(J\), an inverse temperature \(\beta\), and
\(p=1-e^{-2\beta}<1\). Work within the class of products of extremal
one-replica Gibbs states whose ES lifts satisfy the one-replica
\(0/1\)-infinite-cluster property. In both assertions below, require
that the two coherent companion densities in
Definition~\ref{def:subseq-pc-density-imbalance} exist almost surely
under the corresponding product FK marginal along the stated
deterministic sequence \(\Gamma\), with both densities zero when there
is no infinite paired FK cluster.

For a deterministic van~Hove sequence
\(
   \Gamma,
\)
say that a product state
\(
   \mu^{\otimes 2}_{J,\beta}
   =
   \mu^{(1)}_{J,\beta}\otimes\mu^{(2)}_{J,\beta}
\)
has nonzero \(\Gamma\)-overlap if \(q_\Gamma(J,\beta)\)
defined in \eqref{qgamma}
exists and is nonzero.

Then the following are equivalent.

\begin{enumerate}
\item[\textup{(i)}]
There exist extremal one-replica Gibbs states
\(\mu^{(1)}_{J,\beta},\mu^{(2)}_{J,\beta}\) in the  class 
of states satisfying the \(0/1\)-infinite-cluster property,
and a
deterministic van~Hove sequence \(\Gamma\), satisfying the preceding
density-existence requirement, such that
\(\mu^{(1)}_{J,\beta}\otimes\mu^{(2)}_{J,\beta}\) has nonzero
\(\Gamma\)-overlap.

\item[\textup{(ii)}]
There exist extremal one-replica ES states
\(\pi^{(1)}_{J,\beta},\pi^{(2)}_{J,\beta}\) in the class 
of states satisfying the \(0/1\)-infinite-cluster property, and a
deterministic van~Hove sequence \(\Gamma\), satisfying the same
density-existence requirement, such that the product FK marginal
\[
   \nu^{\otimes 2}_{J,p}
   :=
   \Pi^{\otimes 2}_{\mathrm{FK}}
   \bigl(
      \pi^{(1)}_{J,\beta}\otimes\pi^{(2)}_{J,\beta}
   \bigr)
\]
has pc--FK percolation along \(\Gamma\).
\end{enumerate}

In this case, flipping all spins in one of the two replicas changes the sign
of the \(\Gamma\)-overlap but leaves the FK marginal unchanged. Thus pc--FK
percolation along some deterministic van~Hove sequence detects coexistence
of the two relative overlap orientations.
\end{corollary}

\begin{proof}
Fix a product state and a deterministic van~Hove sequence \(\Gamma\)
satisfying the stated density-existence requirement. All density limits,
selectors and pc--FK percolation are defined relative to that same
sequence. The hypotheses of
Theorem~\ref{thm:product_extremal_pc_overlap} are therefore satisfied,
and it gives the equivalence between nonzero limiting expected
\(\Gamma\)-overlap and pc--FK percolation along \(\Gamma\).
Taking the existential quantifier over these state--sequence pairs
proves the equivalence of \textup{(i)} and \textup{(ii)}.

Finally, a global spin flip in one replica preserves the Hamiltonian and the
FK marginal, but changes
\(\omega^{(1)}_x\omega^{(2)}_x\) into its negative for every \(x\). Hence the
two relative orientations have opposite \(\Gamma\)-overlaps.
\end{proof}

\begin{theorem}[Quenched averaged oriented product-extremal overlap]
\label{Th1_replica_product_extremal}
In \(\mathbb Z^d\), fix \(\lambda\in[0,1]\), \(0\le\beta<\infty\),
and \(p=1-e^{-2\beta}<1\), and fix a deterministic van~Hove sequence
\(
   \Gamma.
\)
For \(P_\lambda\)-a.e. coupling realization \(J\), let
\(
   \mu^{(1)}_{J,\beta},\ \mu^{(2)}_{J,\beta}
\)
be extremal one-replica Gibbs states whose ES lifts
\(
   \pi^{(1)}_{J,\beta},\ \pi^{(2)}_{J,\beta}
\)
satisfy the $0/1$ property.
Assume that the choices \(J\mapsto \mu^{(r)}_{J,\beta}\) are measurable
with respect to the disorder.

Set
\[
   \mu^{\otimes 2}_{J,\beta}
   :=
   \mu^{(1)}_{J,\beta}\otimes\mu^{(2)}_{J,\beta},
   \qquad
   \pi^{\otimes 2}_{J,\beta}
   :=
   \pi^{(1)}_{J,\beta}\otimes\pi^{(2)}_{J,\beta},
   \quad \text{ and } \quad
   \nu^{\otimes 2}_{J,p}
   :=
   \Pi^{\otimes 2}_{\mathrm{FK}}\pi^{\otimes 2}_{J,\beta}.
\]
Assume that the density limits required in
Definition~\ref{def:subseq-pc-density-imbalance} exist for
\[
   P_\lambda(dJ)\,\nu^{\otimes 2}_{J,p}(d\eta^{\otimes 2})\text{-a.e.}
\]
configuration.

Choose the relative orientation measurably in \(J\), as in the proof below,
and retain the same notation for the oriented states, so that on the
selected tail event the selected companion
\[
   C_\infty^A(\eta^{\otimes 2})
   =
   \mathcal S^{(2)}_{D,\Gamma}(\eta^{\otimes 2})
\]
carries agreement:
\[
   \omega^{(1)}=\omega^{(2)}
   \quad\text{on }C_\infty^A(\eta^{\otimes 2}).
\]
Let
\[
   D_\Gamma(\eta^{\otimes 2})
   :=
   d_\Gamma(C_\infty^A(\eta^{\otimes 2}))
   -
   d_\Gamma(C_\infty^B(\eta^{\otimes 2})),
\]
with \(D_\Gamma=0\) outside the selected tail event. Then, for
\(P_\lambda\)-a.e. \(J\), the limit
\[
   q_\Gamma^+(J,\beta)
   :=
   \lim_{n\to\infty}
   \mu^{\otimes 2}_{J,\beta}
   \left[
      \frac{1}{|\Lambda_n|}
      \sum_{x\in\Lambda_n}
      \omega^{(1)}_x\omega^{(2)}_x
   \right]
\]
exists and satisfies
\[
   q_\Gamma^+(J,\beta)
   =
   \mathbb E_{\nu^{\otimes 2}_{J,p}}
   \bigl[
      D_\Gamma(\eta^{\otimes 2})
   \bigr].
\]
Consequently,
\[
   q_\Gamma^+(J,\beta)>0
   \quad\Longleftrightarrow\quad
   \nu^{\otimes 2}_{J,p}(D_\Gamma>0)>0.
\]

If \(q_\Gamma(J,\beta)\) denotes the overlap before the orientation choice,
then \(q_\Gamma^+(J,\beta)=|q_\Gamma(J,\beta)|\).
After averaging over the couplings, define the oriented averaged mutual overlap by
\[
   q_{\mathrm{mut},\Gamma}^{+}
   :=
   \mathbb E_{P_\lambda}
   \bigl[
      q_\Gamma^+(J)
   \bigr].
\]
Then
\[
   q_{\mathrm{mut},\Gamma}^{+}
   =
   \mathbb E_{P_\lambda}
   \left[
      \mathbb E_{\nu^{\otimes 2}_{J,p}}
      \bigl[
         D_\Gamma(\eta^{\otimes 2})
      \bigr]
   \right].
\]
In particular,
\[
   q_{\mathrm{mut},\Gamma}^{+}>0
   \quad\Longleftrightarrow\quad
   P_\lambda\!\left(
      \nu^{\otimes 2}_{J,p}(D_\Gamma>0)>0
   \right)>0.
\]
Equivalently, the oriented averaged mutual overlap along \(\Gamma\) is
nonzero precisely when pc--FK percolation occurs along \(\Gamma\) with
positive \(P_\lambda\)-probability.
\end{theorem}

\begin{proof}
For each fixed \(J\), apply
Theorem~\ref{thm:product_extremal_pc_overlap}. The ES bond-sampling kernel
is jointly Borel, so the assumed measurable spin-state choices give
measurable ES lifts and product FK kernels. Lemma~\ref{lem:joint-density-measurability}
therefore makes the density difference and its quenched expectation
measurable in the disorder.

Before orienting the states, set
\[
 q_n(J):=\frac1{|\Lambda_n|}\sum_{x\in\Lambda_n}
       \mu^{(1)}_{J,\beta}(\omega_x)\,
       \mu^{(2)}_{J,\beta}(\omega_x),
 \qquad q(J):=\lim_{n\to\infty}q_n(J),
\]
with \(q=0\) off the convergence set. Each \(q_n\) is Borel, and the
fixed-\(J\) theorem gives convergence for almost every \(J\), so \(q\)
is measurable. If \(F\omega=-\omega\), replace the second state by
\[
 \widehat\mu^{(2)}_{J,\beta}:=
 \begin{cases}
    \mu^{(2)}_{J,\beta},&q(J)\ge0,\\
    F_*\mu^{(2)}_{J,\beta},&q(J)<0.
 \end{cases}
\]
This is a measurable choice, preserves extremality and the FK marginal,
and keeps the replicas independent at fixed \(J\). The flip depends on
\(J\), not on sampled replica configurations. By the fixed-\(J\) identity,
it chooses agreement on the denser companion whenever the selected
event has positive probability. Thus, after this orientation choice,
\(q_\Gamma^+(J)=|q(J)|\) and
\[
   q_\Gamma^+(J)
   =
   \mathbb E_{\nu^{\otimes 2}_{J,p}}
   \bigl[
      D_\Gamma(\eta^{\otimes 2})
   \bigr]
\]
for \(P_\lambda\)-a.e. \(J\). The equivalence
\[
   q_\Gamma^+(J)>0
   \quad\Longleftrightarrow\quad
   \nu^{\otimes 2}_{J,p}(D_\Gamma>0)>0
\]
follows from \(D_\Gamma\ge0\).

Integrating the identity over \(P_\lambda(dJ)\) gives
\[
   q_{\mathrm{mut},\Gamma}^{+}
   =
   \mathbb E_{P_\lambda}
   \left[
      \mathbb E_{\nu^{\otimes 2}_{J,p}}
      \bigl[
         D_\Gamma(\eta^{\otimes 2})
      \bigr]
   \right].
\]
Since \(D_\Gamma\ge0\), this averaged quantity is positive exactly when
\(D_\Gamma>0\) with positive joint
\(P_\lambda(dJ)\nu^{\otimes 2}_{J,p}(d\eta^{\otimes 2})\)-probability, which
is precisely pc--FK percolation along \(\Gamma\) with positive
\(P_\lambda\)-probability.
\end{proof}

\begin{corollary}[Diagonal replicas and the Edwards--Anderson self-overlap]
\label{cor:diagonal_pc_EA}
Under the hypotheses of
Theorem~\ref{thm:product_extremal_pc_overlap}, suppose that the two
one-replica states coincide:
\[
   \mu^{(1)}_{J,\beta}
   =
   \mu^{(2)}_{J,\beta}
   =:
   \mu_{J,\beta}.
\]
Let \(\pi_{J,\beta}\) be their common ES lift and
\(\nu_{J,p}:=\Pi_{\mathrm{FK}}\pi_{J,\beta}\).
The replicas are independent, with spin law
\(\mu_{J,\beta}\otimes\mu_{J,\beta}\) and FK law
\(\nu_{J,p}\otimes\nu_{J,p}\).

The Edwards--Anderson self-overlap of the chosen extremal state along
\(\Gamma\),
\[
   q_{\mathrm{EA},\Gamma}(J,\beta;\mu_{J,\beta})
   :=
   \lim_{n\to\infty}
   \frac{1}{|\Lambda_n|}
   \sum_{x\in\Lambda_n}
      \mu_{J,\beta}(\omega_x)^2,
\]
exists and satisfies
\[
   q_{\mathrm{EA},\Gamma}(J,\beta;\mu_{J,\beta})
   =
   \mathbb E_{\nu_{J,p}\otimes\nu_{J,p}}
   \bigl[D_\Gamma(\eta^{\otimes2})\bigr].
\]
Consequently,
\[
   q_{\mathrm{EA},\Gamma}(J,\beta;\mu_{J,\beta})>0
   \quad\Longleftrightarrow\quad
   (\nu_{J,p}\otimes\nu_{J,p})(D_\Gamma>0)>0.
\]
No relative spin-flip orientation is needed in this diagonal case.

For a measurable family \(J\mapsto\mu_{J,\beta}\) satisfying the
hypotheses of Theorem~\ref{Th1_replica_product_extremal} with these
diagonal products and the same deterministic sequence \(\Gamma\), define
\[
   q_{\mathrm{EA},\Gamma}
   :=
   \mathbb E_{P_\lambda}
   \bigl[
      q_{\mathrm{EA},\Gamma}(J,\beta;\mu_{J,\beta})
   \bigr].
\]
Then
\[
   q_{\mathrm{EA},\Gamma}
   =
   \mathbb E_{P_\lambda}
   \left[
      \mathbb E_{\nu_{J,p}\otimes\nu_{J,p}}
      \bigl[D_\Gamma(\eta^{\otimes2})\bigr]
   \right],
\]
and
\[
   q_{\mathrm{EA},\Gamma}>0
   \quad\Longleftrightarrow\quad
   P_\lambda\!\left(
      (\nu_{J,p}\otimes\nu_{J,p})(D_\Gamma>0)>0
   \right)>0.
\]
\end{corollary}

\begin{proof}
Independence of the two replicas gives, for every \(n\),
\[
   (\mu_{J,\beta}\otimes\mu_{J,\beta})(Q_{\Lambda_n})
   =
   \frac{1}{|\Lambda_n|}
   \sum_{x\in\Lambda_n}
      \mu_{J,\beta}(\omega_x)^2
   \ge0.
\]
Theorem~\ref{thm:product_extremal_pc_overlap} gives existence of the
limit and identifies its absolute value with
\(\mathbb E_{\nu_{J,p}\otimes\nu_{J,p}}D_\Gamma\).
Since the limit is nonnegative, the absolute value can be removed.
In particular, the orientation procedure in
Theorem~\ref{Th1_replica_product_extremal} leaves the second replica
unchanged in the diagonal case.

The self-overlap is measurable in \(J\), being the almost-sure limit
of the displayed bounded measurable spatial averages.
Averaging the fixed-\(J\) identity over the disorder gives the asserted
averaged identity. Both positivity equivalences follow from
\(D_\Gamma\ge0\).
\end{proof}

\noindent\emph{Relation between the two density criteria.}
Fix \(J,\beta,\Gamma\), and let \(\mu_{J,\beta}\) be an extremal state
with FK marginal \(\nu_{J,p}\) satisfying the hypotheses of both
Theorem~\ref{EquivalenceOneCopy} and
Corollary~\ref{cor:diagonal_pc_EA} along the same sequence \(\Gamma\).
Then sf--FK percolation under \(\nu_{J,p}\) implies pc--FK percolation
under \(\nu_{J,p}\otimes\nu_{J,p}\), since
\[
   m_\Gamma(\mu_{J,\beta})^2
   \le q_{\mathrm{EA},\Gamma}(J,\beta;\mu_{J,\beta}).
\]
Either criterion implies that \(\mu_{J,\beta}\) and its global spin
flip are distinct extremal Gibbs states, as in
Section~\ref{subsec:order-parameters}.

\noindent\emph{Disorder-dependent density sequences.}
Theorem~\ref{Th1_replica_product_extremal} and
Corollary~\ref{cor:diagonal_pc_EA} also apply with a
\(P_\lambda\)-completion-measurable choice \(J\mapsto\Gamma(J)\),
using a Borel version outside a null set. Every density and limit is then
evaluated along that one sequence at the fixed \(J\), and the joint
almost-sure density-existence hypothesis is retained. The preceding
measurability and orientation proof is unchanged, with
\(\Lambda_n\) replaced by \(\Lambda_n(J)\). Proposition~\ref{prop:measurable-density-sequences}
supplies measurable witnesses when the corresponding sections are
nonempty. For a stationary periodic-boundary joint two-replica FK law,
Lemma~\ref{lem:stationary-companion-densities} instead supplies a common
deterministic cube sequence. The resulting overlap formula for the
actual joint limit is given in Corollary~\ref{cor:periodic_pc_overlap}.
Stationarity of that law does not require its individual extremal
components to be stationary.

\subsection{Periodic-boundary application}
\label{subsec:periodic-overlap}

\begin{corollary}[Absolute overlap in stationary joint laws]
\label{cor:periodic_pc_overlap}
Fix \(\lambda\in[0,1]\) and \(0\le\beta<\infty\), and put
\(p=1-e^{-2\beta}<1\). Let \(\rho^2\) be a translation-invariant joint
law of the couplings and two-replica ES configurations on
\(\mathbb Z^d\), with coupling marginal \(P_\lambda\). Write
\[
 \rho^2(dJ,d\omega^2,d\eta^2)
 =P_\lambda(dJ)\,\pi_J^2(d\omega^2,d\eta^2),
\]
and suppose that \(\pi_J^2\in\mathcal G^{\mathrm{ES},2}_{J,\beta}\)
for \(P_\lambda\)-almost every \(J\). Assume that each replica has at
most one infinite ordinary FK cluster, \(\rho^2\)-almost surely.
Denote the joint FK marginal by \(\rho^{\mathrm{FK},2}\), and the FK
marginal of \(\pi_J^2\) by \(\nu_J^2\).

Let \(\Gamma=(\Lambda_n)_{n\ge1}\), where
\(\Lambda_n=[-n,n]^d\cap\mathbb Z^d\). The two coherent companion
densities exist along \(\Gamma\), almost surely. Let
\(D_\Gamma(J,\eta^2)\) be their absolute difference, with
\(D_\Gamma=0\) when there is no infinite paired FK cluster.
The sampled overlap
\[
 Q(\omega^2)
 :=\lim_{n\to\infty}
       \frac1{|\Lambda_n|}
       \sum_{x\in\Lambda_n}\omega_x^{(1)}\omega_x^{(2)}
\]
also exists \(\rho^2\)-almost surely. Moreover,
\[
 \pi_J^2(|Q|)=\nu_J^2(D_\Gamma)
 \qquad\text{for }P_\lambda\text{-almost every }J,
\]
and consequently
\[
 \mathbb E_{\rho^2}|Q|
 =\mathbb E_{\rho^{\mathrm{FK},2}}D_\Gamma.
\]
In particular,
\[
 \rho^2(Q\ne0)>0
 \quad\Longleftrightarrow\quad
 \rho^{\mathrm{FK},2}(D_\Gamma>0)>0.
\]
These conclusions apply to the joint periodic-boundary limits of
Lemma~\ref{lem:periodic-joint-ES} whenever the stated one-replica
\(0/1\)-infinite-cluster properties hold.
\end{corollary}

\begin{proof}
Lemma~\ref{lem:stationary-companion-densities} gives almost-sure
existence of the coherent companion densities along \(\Gamma\).
The pointwise ergodic theorem applied to
\(\omega_0^{(1)}\omega_0^{(2)}\) gives almost-sure convergence of the
sampled overlap. Set \(Q=0\) on its null nonconvergence set. Changing
finitely many spins changes each spatial average by a quantity tending
to zero, so \(Q\) is measurable with respect to the two-replica tail
sigma-field.

Fix a coupling realization \(J\) for which the quenched ES DLR
equations, the one-replica \(0/1\) properties, the density limits,
and convergence of \(Q\) hold almost surely. By
Theorem~\ref{thm:ES-2replica} and
Proposition~\ref{prop:full-two-replica-extremals-factor}, the extremal
decomposition of \(\pi_J^2\) has the form
\[
 \pi_J^2
 =\int \pi^{(1)}\otimes\pi^{(2)}
       \,\kappa_J(d\pi^{(1)},d\pi^{(2)}),
\]
where both factors are extremal one-replica ES states. Each of the
preceding almost-sure properties holds under \(\kappa_J\)-almost
every component.

On each such component, tail triviality makes \(Q\) almost surely
constant. Bounded convergence identifies this constant with the
limiting expected mutual overlap of its spin marginals. Writing
\(\nu^{(r)}=\Pi_{\mathrm{FK}}\pi^{(r)}\),
Theorem~\ref{thm:product_extremal_pc_overlap} therefore gives
\[
 (\pi^{(1)}\otimes\pi^{(2)})(|Q|)
 =(\nu^{(1)}\otimes\nu^{(2)})(D_\Gamma).
\]
Integration over the extremal decomposition proves
\(\pi_J^2(|Q|)=\nu_J^2(D_\Gamma)\). Both sides are measurable
quenched expectations; measurability of the density function follows
from Lemma~\ref{lem:joint-density-measurability}. Integrating this
identity over \(P_\lambda\) proves the averaged formula. The
probability equivalence follows because \(|Q|\) and \(D_\Gamma\)
are nonnegative.

For the periodic-boundary application,
Lemma~\ref{lem:periodic-joint-ES} supplies joint translation invariance
and the quenched full two-replica ES Gibbs property. Thus the hypotheses
above hold whenever the one-replica \(0/1\) properties hold.
\end{proof}

The quantity \(\mathbb E_{\rho^2}|Q|\) averages absolute mutual
overlaps over the pairs of extremal states in the quenched
decomposition. For a pair consisting of two copies of the same
extremal state, the corresponding overlap is its Edwards--Anderson
self-overlap, as in Corollary~\ref{cor:diagonal_pc_EA}.

\begin{appendix}
\section*{Disorder measurability and choices of density sequences}
\label{app:disorder-measurability}

In a fixed-\(J\) argument, a deterministic density sequence is independent of
all sampled spin and FK configurations; it may depend on \(J\). The
existential quantifier over sequences in
Definition~\ref{def:percolation-selector} is outside the FK probability.
Thus one sequence must witness positive probability for the state at the
fixed \(J\). Compactness of the two density fractions only gives a convergent
subsequence for each individual configuration; it does not by itself give
such a common sequence. Nor does nonconvergence alone imply unequal
subsequential limits. All density-existence hypotheses are understood
with this order of quantifiers, also for two replicas.

For disorder probabilities and averages we use the completion of
\(P_\lambda\). The state families entering such averages are measurable
probability kernels, with Borel versions chosen outside a disorder-null
set. Write \(\mathfrak J=\{-1,+1\}^{E}\) and equip it, and the spin and bond
spaces, with their product Borel sigma-fields. We work on a countable
locally finite graph. Dependence on \(J\), suppressed in the fixed-coupling
notation elsewhere, is made explicit in this appendix.

\begin{lemma}[Joint measurability of companion densities]
\label{lem:joint-density-measurability}
Let \(\mathfrak V\) be the space of admissible van~Hove sequences, encoded
as a Borel subset of \(\operatorname{Fin}(V)^{\mathbb N}\), where
\(\operatorname{Fin}(V)\) is the countable discrete space of nonempty finite
vertex sets. A prescribed class of admissible subsequences of reference
boxes can be used instead. Let
\(\mathcal E^{\mathrm{sf}}_\Gamma\) be the event that the unique infinite
FK cluster exists and both companion density limits along \(\Gamma\)
exist and are unequal. On this event define
\[
 D^{\mathrm{sf}}_\Gamma(J,\eta)
 =\bigl|d_\Gamma(C^0(J,\eta))-d_\Gamma(C^1(J,\eta))\bigr|,
\]
and put \(D^{\mathrm{sf}}_\Gamma=0\) otherwise. Here \(C^0,C^1\) are any
Borel labeling of the two companions. The event
\(\{(J,\Gamma,\eta):(J,\eta)\in\mathcal E^{\mathrm{sf}}_\Gamma\}\), the
function \(D^{\mathrm{sf}}_\Gamma\), and each membership indicator of
\(\mathcal S_{D,\Gamma}(J,\eta)\) are jointly Borel. Consequently, for a
Borel FK kernel \(J\mapsto\nu_J\), both
\[
 \theta^{\mathrm{sf}}(J,\Gamma)
 :=\nu_J(\mathcal E^{\mathrm{sf}}_{\Gamma,J}),
 \qquad
 (J,\Gamma)\longmapsto\int D^{\mathrm{sf}}_\Gamma(J,\eta)\,\nu_J(d\eta)
\]
are Borel, where \(\mathcal E^{\mathrm{sf}}_{\Gamma,J}\) denotes the
\(J\)-section. The same assertions hold for the paired coherent companions
of Section~7 and any Borel two-replica FK kernel, with
\(\eta\) replaced by \((\eta^{(1)},\eta^{(2)})\).
\end{lemma}

\begin{proof}
Fix an enumeration of \(V\). Connectivity with prescribed parity is a
countable union of finite-path events, jointly Borel in the couplings and
bonds. Compatibility excludes the countably many frustrated open cycles.
Infinitude and uniqueness of a cluster are also countable Borel conditions.
On the compatible unique-infinite-cluster event, choose its first vertex
in the enumeration and label the companions by their parities relative to
that vertex. Set both labeled sets empty off that event. This is a Borel
labeling, and
\[
 u_n^i(J,\Gamma,\eta)
 :=\frac{|C^i(J,\eta)\cap\Lambda_n|}{|\Lambda_n|},\qquad i=0,1,
\]
is jointly Borel. The Cauchy criterion expresses existence of both limits
by countably many Borel conditions. Comparing their limits proves the
claims about the event, density difference and selector. The labeling
does not affect either the event or the denser companion.
Integration of a bounded Borel function against a Borel kernel is Borel,
by a monotone-class argument.

For two replicas, label the companions from the first vertex of the unique
infinite paired cluster, using the product of the two path parities.
The required connectivity relations and the infinitude of the intersection
are Borel by the same argument. This proves the two-replica claims.
Nesting, exhaustion and the van~Hove conditions on a sequence, including
any additionally required volume-ratio condition, are countable Borel
conditions, as are the increasing-index conditions for subsequences.
\end{proof}

\begin{proposition}[Existential events and measurable witnessing sequences]
\label{prop:measurable-density-sequences}
For a Borel FK kernel \(J\mapsto\nu_J\), the disorder event
\[
 \mathcal A^{\mathrm{sf}}
 :=\{J:\text{there exists }\Gamma\in\mathfrak V
                  \text{ with }\theta^{\mathrm{sf}}(J,\Gamma)>0\}
\]
is universally measurable. It therefore has a well-defined
\(P_\lambda\)-probability without any choice of witnesses. On this event a
universally measurable witnessing map \(J\mapsto\Gamma(J)\) exists and
has a Borel version outside a \(P_\lambda\)-null set. The same holds for
two replicas, including the additional almost-sure density-convergence
requirement in Definition~\ref{def:subseq-pc-density-imbalance}.
\end{proposition}

\begin{proof}
By Lemma~\ref{lem:joint-density-measurability},
\(\mathcal R=\{(J,\Gamma):\theta^{\mathrm{sf}}(J,\Gamma)>0\}\)
is Borel. Its projection is analytic and hence universally measurable;
the analytic selection theorem supplies a universally measurable selector
on that projection.\footnote{N.~El Karoui and X.~Tan (2013),
\emph{Capacities, Measurable Selection and Dynamic Programming, Part~I:
Abstract Framework}, arXiv:1310.3363, Theorem~2.23.}
For a fixed probability law, a completion-measurable map into a standard
Borel space admits a Borel version almost surely. The additional condition
that both densities converge with kernel probability one is Borel, again
by Lemma~\ref{lem:joint-density-measurability}, and can be intersected with
\(\mathcal R\) before taking its projection. The two-replica proof is identical.
\end{proof}

Thus an arbitrary, possibly nonmeasurable choice of witnessing sequences
is unnecessary for the existential definition. Whenever the chosen
densities, selectors or orientations themselves enter disorder averages,
a measurable witness is used. Different witnesses need not give the same
density values. Measurable selection does not establish that a witnessing
sequence exists; it only selects one from each nonempty section. Once
\(J\) is fixed, \(\Gamma(J)\) is fixed under local modifications of the
sampled configurations, so the fixed-\(J\) tail arguments are unchanged.

For the existential statements over extremal states in
Corollaries~\ref{cor:coexistence_some_vanhove} and
\ref{cor:spin_glass_some_vanhove}, state parameters may also be included
in the projection. At fixed finite \(\beta\), the joint ES Gibbs graph
\(\{(J,\pi):\pi\in\mathcal G^{\mathrm{ES}}_{J,\beta}\}\) is closed by
the countable local DLR equations. Its extremal subgraph is Borel:
nonextremality is the union over \(m\ge1\) of the compact projections of
\[
 \pi=\tfrac12(\pi_1+\pi_2),\qquad
 \pi_i\in\mathcal G^{\mathrm{ES}}_{J,\beta},\qquad
 \operatorname{dist}(\pi_1,\pi_2)\ge m^{-1},
\]
where \(\operatorname{dist}\) metrizes weak convergence. The \(0/1\)
property and the density conditions are Borel. Hence the same projection
argument applies when the witnessing state, or pair of states, is also
quantified existentially.

\begin{lemma}[A common density sequence under stationarity]
\label{lem:stationary-companion-densities}
Let \(\rho^{\mathrm{FK}}\) be a translation-invariant joint law of
\((J,\eta)\) on \(\mathbb Z^d\), supported on compatible configurations
with at most one infinite FK cluster. Both signed companions of that
cluster have density limits along the deterministic cubes
\(\Lambda_n=[-n,n]^d\cap\mathbb Z^d\), almost surely. The same conclusion
holds for the two coherent companions under a translation-invariant joint
law of \((J,\eta^{(1)},\eta^{(2)})\) with compatible configurations and
the one-replica \(0/1\) property in each replica.
\end{lemma}

\begin{proof}
Before any density-based selection, color each FK cluster independently
with its two compatible spin assignments having probability \(1/2\),
including the infinite cluster. The coloring kernel is jointly Borel by
the enumeration construction and is intrinsically translation-covariant,
since its law does not depend on the chosen roots. Thus the colored joint
law is stationary. Apply the pointwise ergodic theorem along cubes to
\[
 f_\pm(J,\eta,\omega)
 :=\mathbf1_{\{0\in I_\infty(\eta),\ \omega_0=\pm1\}}.
\]
Their spatial averages are the density fractions of the two companions
labeled by the assigned spin; both averages converge almost surely.
The pointwise ergodic theorem requires stationarity, not ergodicity, for
this conclusion.\footnote{See E.~Lindenstrauss (2001),
\emph{Pointwise theorems for amenable groups}, Inventiones mathematicae
\textbf{146}, 259--295; deterministic cubes are tempered F\o lner sets.}
Existence of both limits is independent of exchanging the two spin
labels, so it holds for the original FK law.

For two replicas, color the clusters in each replica fairly and
independently, and apply the same theorem to
\[
 f^{(2)}_\pm
 :=\mathbf1_{\{0\in I^{(1)}_\infty\cap I^{(2)}_\infty,
                       \ \omega^{(1)}_0\omega^{(2)}_0=\pm1\}}.
\]
On an infinite intersection these averages count the two coherent
companions. A finite intersection has zero density. Projecting out the
colors gives the assertion.
\end{proof}

In particular, disintegrating a stationary periodic-boundary joint limit
as \(\rho^{\mathrm{FK}}(dJ,d\eta)=P_\lambda(dJ)\nu_J(d\eta)\) shows that
these same deterministic cubes work \(\nu_J\)-almost surely for
\(P_\lambda\)-almost every \(J\). The same conclusion holds for the
two-replica conditional FK laws. Under an extremal decomposition of the
associated conditional ES state, the density-convergence event has
probability one on almost every product-extremal component. This is the
common sequence used in Corollary~\ref{cor:periodic_pc_overlap}.
The subsequence used to take a weak
limit of the finite-volume joint laws can be chosen deterministically;
it is distinct from a density-averaging sequence. Joint weak convergence
does not assert quenched weak convergence along that subsequence.
Stationarity is a property of the actual joint law and is not implied by
an arbitrary measurable choice of extremal states or density sequences.

\end{appendix}

\begin{acks}[Acknowledgments]
Daniel Stein is also affiliated with the Department of Physics, New York University.
\end{acks}



\begin{thebibliography}{17}
\bibitem[Biskup et al.(2000)]{biskup2000gibbs}
\textsc{Biskup, M.}, \textsc{Borgs, C.}, \textsc{Chayes, J. T.} and \textsc{Koteck{\'y}, R.} (2000). Gibbs states of graphical representations of the Potts model with external fields. \textit{Journal of Mathematical Physics} \textbf{41} 1170--1210.

\bibitem[Carlson et al.(1990)]{carlson1990bethe}
\textsc{Carlson, J. M.}, \textsc{Chayes, J. T.}, \textsc{Chayes, L.}, \textsc{Sethna, J. P.} and \textsc{Thouless, D. J.} (1990). Bethe lattice spin glass: the effects of a ferromagnetic bias and external fields. I. Bifurcation analysis. \textit{Journal of Statistical Physics} \textbf{61} 987--1067.

\bibitem[Duminil-Copin(2019)]{duminil2018sixty}
\textsc{Duminil-Copin, H.} (2019). Sixty years of percolation. In \textit{Proceedings of the International Congress of Mathematicians: Rio de Janeiro 2018} 2829--2856. World Scientific.

\bibitem[Edwards and Sokal(1988)]{edwards1988generalization}
\textsc{Edwards, R. G.} and \textsc{Sokal, A. D.} (1988). Generalization of the {Fortuin--Kasteleyn--Swendsen--Wang} representation and {Monte Carlo} algorithm. \textit{Physical Review D} \textbf{38} 2009--2012.

\bibitem[Edwards and Anderson(1975)]{edwards1975theory}
\textsc{Edwards, S. F.} and \textsc{Anderson, P. W.} (1975). Theory of spin glasses. \textit{Journal of Physics F: Metal Physics} \textbf{5} 965--974.

\bibitem[Fortuin and Kasteleyn(1972)]{fortuin1972random}
\textsc{Fortuin, C. M.} and \textsc{Kasteleyn, P. W.} (1972). On the random-cluster model: I. Introduction and relation to other models. \textit{Physica} \textbf{57} 536--564.

\bibitem[Gandolfi et al.(1992)]{gandolfi1992uniqueness}
\textsc{Gandolfi, A.}, \textsc{Keane, M. S.} and \textsc{Newman, C. M.} (1992). Uniqueness of the infinite component in a random graph with applications to percolation and spin glasses. \textit{Probability Theory and Related Fields} \textbf{92} 511--527.

\bibitem[Grimmett(2004)]{grimmett2004random}
\textsc{Grimmett, G.} (2004). The random-cluster model. In \textit{Probability on Discrete Structures} 73--123. Springer.

\bibitem[Harary(1953)]{harary1953balance}
\textsc{Harary, F.} (1953). On the notion of balance of a signed graph. \textit{Michigan Mathematical Journal} \textbf{2} 143--146.

\bibitem[Machta et al.(2008)]{machta2008percolation}
\textsc{Machta, J.}, \textsc{Newman, C. M.} and \textsc{Stein, D. L.} (2008). The percolation signature of the spin glass transition. \textit{Journal of Statistical Physics} \textbf{130} 113--128.

\bibitem[M{\"u}nster and Weigel(2023)]{munster2023cluster}
\textsc{M{\"u}nster, L.} and \textsc{Weigel, M.} (2023). Cluster percolation in the two-dimensional {Ising} spin glass. \textit{Physical Review E} \textbf{107} 054103.

\bibitem[M{\"u}nster and Weigel(2026)]{munster2026cluster}
\textsc{M{\"u}nster, L.} and \textsc{Weigel, M.} (2026).
Cluster percolation in the three-dimensional $\pm J$ random-bond
{Ising} model. \textit{Physical Review E} \textbf{113} 024139.
\href{https://doi.org/10.1103/y464-wxbk}{doi:10.1103/y464-wxbk}.

\bibitem[Pei(2026)]{pei2026blue}
\textsc{Pei, Y. R.} (2026). Towards the spin-glass transition in finite dimensions via blue percolation. Preprint. \href{https://arxiv.org/abs/2609.22301v1}{arXiv:2609.22301v1}.

\bibitem[Ruelle(2004)]{ruelle2004thermodynamic}
\textsc{Ruelle, D.} (2004). \textit{Thermodynamic Formalism: The Mathematical Structure of Equilibrium Statistical Mechanics}. Cambridge University Press.

\bibitem[Swendsen and Wang(1987)]{swendsen1987nonuniversal}
\textsc{Swendsen, R. H.} and \textsc{Wang, J.-S.} (1987). Nonuniversal critical dynamics in {Monte Carlo} simulations. \textit{Physical Review Letters} \textbf{58} 86--88.

\bibitem[Zaslavsky(1982)]{zaslavsky1982signed}
\textsc{Zaslavsky, T.} (1982). Signed graphs. \textit{Discrete Applied Mathematics} \textbf{4} 47--74.

\bibitem[Zaslavsky(1989)]{zaslavsky1989biased}
\textsc{Zaslavsky, T.} (1989). Biased graphs. I. Bias, balance, and gains. \textit{Journal of Combinatorial Theory, Series B} \textbf{47} 32--52.

\end{thebibliography}
\end{document}